\documentclass{amsart}
\usepackage{hyperref}
\usepackage{graphicx}
\usepackage{curves}
\usepackage{tikz}
\usepackage{enumerate}
\usepackage{enumitem}
\setitemize{leftmargin=*}
\setenumerate{leftmargin=*}
\nonstopmode
\newcommand{\arxiv}[1]{\href{http://arxiv.org/abs/#1}{arXiv:#1}}
\newcommand*{\mailto}[1]{\href{mailto:#1}{\nolinkurl{#1}}}

\newtheorem{theorem}{Theorem}[section]
\newtheorem{RHproblem}{RH problem}
\newtheorem*{modelRHP}{Model RH problem}
\newtheorem{lemma}[theorem]{Lemma}
\newtheorem{corollary}[theorem]{Corollary}

\theoremstyle{definition}
\newtheorem{remark}[theorem]{Remark}

\newcommand{\R}{\mathbb{R}}
\newcommand{\Z}{\mathbb{Z}}

\newcommand{\C}{\mathbb{C}}
\newcommand{\T}{\mathbb{T}}
\newcommand{\M}{\mathbb{M}}

\newcommand{\X}{\mathbb{X}}

\newcommand{\nn}{\nonumber}
\newcommand{\be}{\begin{equation}}
\newcommand{\ee}{\end{equation}}
\newcommand{\beq}{\begin{equation}}
\newcommand{\eeq}{\end{equation}}
\newcommand{\bea}{\begin{eqnarray}}
\newcommand{\eea}{\end{eqnarray}}

\newcommand{\ol}{\overline}
\newcommand{\pa}{\partial}
\newcommand{\ti}{\tilde}

\newcommand{\id}{\mathbb{I}}
\newcommand{\I}{\mathrm{i}}
\newcommand{\E}{\mathrm{e}}

\newcommand{\clos}{\mathop{\mathrm{clos}}}
\newcommand{\re}{\mathop{\mathrm{Re}}}
\newcommand{\im}{\mathop{\mathrm{Im}}}
\DeclareMathOperator{\Mod}{mod}

\DeclareMathOperator{\res}{Res}
\newcommand{\noprint}[1]{}

\makeatletter
\newcommand{\dlmf}[1]{%
\cite[%
  \def\nextitem{\def\nextitem{, }}%
  \@for \el:=#1\do{\nextitem\href{http://dlmf.nist.gov/\el}{(\el)}}%
]{dlmf}%
}
\makeatother

\newcommand{\si}{\sigma}
\newcommand{\la}{\lambda}

\numberwithin{equation}{section}

\newcommand{\rI}{\begin{pmatrix}  1 & 1 \end{pmatrix}}

\begin{document}

\title[Toda shock waves in the modulation region]{Long-time asymptotics for Toda shock waves in the modulation region}

\author[I. Egorova]{Iryna Egorova}
\address{B. Verkin Institute for Low Temperature Physics and Engineering\\ 47, Nauky ave\\ 61103 Kharkiv\\ Ukraine}
\email{\href{mailto:iraegorova@gmail.com}{iraegorova@gmail.com}}

\author[J. Michor]{Johanna Michor}
\address{Faculty of Mathematics\\ University of Vienna\\
Oskar-Morgenstern-Platz 1\\ 1090 Wien\\ Austria}
\email{\href{mailto:Johanna.Michor@univie.ac.at}{Johanna.Michor@univie.ac.at}}
\urladdr{\href{http://www.mat.univie.ac.at/~jmichor/}{http://www.mat.univie.ac.at/\string~jmichor/}}

\author[A. Pryimak]{Anton Pryimak}
\address{B. Verkin Institute for Low Temperature Physics and Engineering\\ 47, Nauky ave\\ 61103 Kharkiv\\ Ukraine}
\email{\href{mailto:pryimakaa@gmail.com}{pryimakaa@gmail.com}}

\author[G. Teschl]{Gerald Teschl}
\address{Faculty of Mathematics\\ University of Vienna\\
Oskar-Morgenstern-Platz 1\\ 1090 Wien\\ Austria\\ and Erwin Schr\"odinger International
Institute for Mathematics and Physics\\ Boltzmanngasse 9\\ 1090 Wien\\ Austria}
\email{\href{mailto:Gerald.Teschl@univie.ac.at}{Gerald.Teschl@univie.ac.at}}
\urladdr{\href{http://www.mat.univie.ac.at/~gerald/}{http://www.mat.univie.ac.at/\string~gerald/}}

\thanks{Research supported by the Austrian Science Fund (FWF) under Grant No.\ P31651.}

\begin{abstract}
We show that Toda shock waves are asymptotically close to a modulated 
finite gap solution in the region separating the soliton and the elliptic wave regions.
We previously derived  formulas for the leading terms of the asymptotic 
expansion of these shock waves  in all principal regions and conjectured that in the modulation region the 
next term is of order $O(t^{-1})$. In the present paper we prove this fact and investigate 
how resonances and eigenvalues influence the leading asymptotic behaviour.
Our main contribution is the solution of the local parametrix Riemann--Hilbert problems and a rigorous 
justification of the analysis. In particular, this involves the construction of a proper singular matrix model solution.
\end{abstract}

\maketitle

\section{Introduction}
A Toda shock wave is a solution of the initial value problem for
the Toda lattice (\cite{tjac, toda})
\begin{align} \label{tl}
	\begin{split}
\dot b(n,t) &= 2(a(n,t)^2 -a(n-1,t)^2),\\
\dot a(n,t) &= a(n,t) (b(n+1,t) -b(n,t)),
\end{split} \quad (n,t) \in \Z \times \R_+,
\end{align}
with a steplike initial profile
\begin{align} \label{ini1}
\begin{split}	
& a(n,0)\to a, \quad b(n,0) \to b, \quad \mbox{as $n \to -\infty$}, \\
& a(n,0)\to \frac{1}{2} \quad b(n,0) \to 0, \quad \mbox{as $n \to +\infty$},
\end{split}
\end{align}
where $a>0$, $b\in \R$ satisfy the condition
\be \label{main}
 b + 2 a < - 1.
\ee
Originally the Toda shock wave  was associated with symmetric initial data (\cite{vdo})
\be\label{deift}
a(n,0)=a(-n, 0)\to\frac{1}{2}, \quad b(-n,0)=-b(n,0)\to \pm b, \quad n\to \pm\infty,\quad b>1.
\ee 
Such a  model is closely related to the motion of driving particles in a container filled with gas ahead of  a piston compressing the content of the container.
From the viewpoint of spectral theory, this model corresponds to two non-intersecting  spectral intervals of equal length, which are associated with left and right background Jacobi operators with constant coefficients; the left background spectrum lies to the left. It is therefore natural to extend the notion of a shock wave to background spectra of different lengths. By scaling and shifting the spectral parameter, one can always assume that the right spectrum coincides with the interval $[-1, 1]$. 

We are interested in the long-time behavior of the solution of the Cauchy problem \eqref{tl}--\eqref{main}, also referred to as {\it Toda shock problem}.
This problem was first investigated for initial data \eqref{deift}  on a physical level of 
rigor by Bloch and Kodama (\cite{bk, bk2}) using the Whitham approach. Venakides, Deift and Oba showed in \cite{vdo} (see also \cite{km0}) based on the  Lax--Levermore approach that in a middle region of the half plane $(n,t)\in \Z\times \R_+$, the solution to \eqref{tl}, \eqref{deift} is asymptotically close to a periodic solution of period two with spectrum $[-1-b, 1-b]\cup [-1 + b, 1+b]$. The classical inverse scattering transform was used to analyze the soliton regions (\cite{toda}) and a transition region behind the leading wave front, where the train of asymptotic solitons was evaluated ({\cite{bek, BE}). The nonlinear steepest descent (NSD) method applied to the Toda shock problem 
yields the most interesting results in the regime $n\to\infty$, $t\to \infty$ with the ratio $n/t$ close to a constant. 
In \cite{emt14} three of us showed that for the solution of \eqref{tl}--\eqref{main} there are five principal regions 
in the $(n,t)$ half plane with different qualitative behavior: the left and right soliton regions,
the left and right modulation regions and the elliptic region or middle region first discussed in \cite{vdo} (see \cite{m15} for an overview).

\subsection{The main asymptotic regions}
The continuous spectrum of the underlying Jacobi operator (the Lax operator for the Toda lattice)
\be \label{ht}
\aligned
H(t)y(n)&=a(n-1,t)y(n-1) + b(n,t) y(n) + a(n,t) y(n+1)\\
&=\la y(n), \quad \la\in\C,
\endaligned
\ee
consists of two intervals $[b-2a, b+2a]$ and $[-1, 1]$ which are the spectra  
of the left and right constant background operators,
\[
\aligned
H_{\ell}y(n) & = a\,y(n-1) + a\,y(n+1) + b\,y(n), \\
H_r y(n) & =\frac{1}{2} y(n-1) +\frac 1 2 y(n+1),
\endaligned \quad  n\in\mathbb Z.
\]
Two parameters $z$ and $\zeta$ are associated with the right and left background; they are connected with the spectral parameter $\la$ by the Joukowsky transform
\be\label{spec5}
\la=\frac{1}{2}\left(z + z^{-1}\right)=b + a\left(\zeta + \zeta^{-1}\right),\quad |z|\leq 1,
\quad |\zeta|\leq 1.
\ee 
In the NSD approach, the behavior of the solution essentially depends on the location of the stationary phase points, that is, the nodal points of the level lines where the real part of the phase function vanishes. In our case both the right  phase function
\be\label{phase4}\Phi(z,\xi)=\frac{z - z^{-1}}{2} + \xi\log z,\quad \xi:=\frac{n}{t},
\ee 
and the left phase function
\be\label{phil}
\Phi_\ell(z, \xi)=a(\zeta^{-1} - \zeta) - \xi\log\zeta,
\ee 
 take part in this characterization. Two soliton regions corresponding to the domains of 
$n$ and $t$ for which $\frac{n}{t}>\xi_{cr}$ or $\frac{n}{t}<\xi_{cr,1}$ are naturally identified, where the solution to \eqref{tl}--\eqref{main} is asymptotically close as $t \to\infty$ to the respective constant background solution
plus a finite number of solitons generated by the discrete spectrum (if any). The right leading wave front
\be\label{ksikr}
\xi_{cr}=\frac{\sqrt{(2a-b)^2 -1}}{\log (2a-b + \sqrt{(2a-b)^2-1})}
\ee 
corresponds to the case when the level line $\re \Phi(z,\xi)=0$ crosses the real line at the point $z(b-2a)$ (the left endpoint of the left spectrum). The left wave front
\be\label{cr1} \xi_{cr,1}=\frac{\sqrt{(1-b)^2-4a^2}}{\log(2a) - \log \big(1-b + \sqrt{(1-b)^2-4a^2}\big)}
\ee
is the value where the stationary phase point of the left phase $\Phi_\ell$ coincides with the right endpoint of the right spectrum, $z=1$.  The region
$\xi_{cr,1}<\frac{n}{t}<\xi_{cr}$ consists of three sectors  with different type of quasi-periodic behavior of the solution. These sectors are divided by rays corresponding to the critical values $\xi_{cr, 1}^\prime$ and $\xi_{cr}^\prime$ of the parameter $\xi$   such that
$\xi_{cr,1}<\xi_{cr,1}^\prime<\xi_{cr}^\prime<\xi_{cr}$. In the modulation regions $\xi_{cr,1}<\frac{n}{t}<\xi_{cr,1}^\prime$
and  $\xi_{cr}^\prime<\frac{n}{t}<\xi_{cr}$, the main terms of the expansion of the solution (with respect to large $t$)
are modulated elliptic waves (\cite{emt14}). In the middle region $\xi_{cr,1}^\prime<\frac{n}{t}<\xi_{cr}^\prime$, the solution is asymptotically close to a finite gap (two band) solution of the Toda equation if the discrete spectrum is absent 
in the gap $(b+2a, -1)$. 

\subsection{Modulated elliptic waves}
A finite gap solution of the Toda equation is completely characterized by the geometry of its continuous spectrum and by the initial Dirichlet divisor on the hyperelliptic Riemann surface associated with the spectrum. In the shock problem we deal with spectra consisting of two bands and one initial Dirichlet eigenvalue in the gap between the bands. The sign necessary to lift this eigenvalue to the Riemann surface is the sign of the respective half-axis, where the corresponding eigenvector is supported.

Let us first discuss the region $\xi\in [\xi_{cr}^\prime, \xi_{cr})$. For any such $\xi$ consider a  point $\gamma(\xi)\in (b-2a, b+2a]$ which moves monotonically and continuously with respect to $\xi$ covering the interval $(b-2a, b+2a]$, with $\gamma(\xi_{cr}^\prime)=b+2a$. Associated with the set 
\be\label{sixi}\si(\xi):=[b-2a, \gamma(\xi)]\cup [-1, 1]
\ee 
is the two-sheeted Riemann surface $\mathbb M(\xi)$. The upper sheet of $\mathbb M(\xi)$ is treated as the complex plane of the spectral parameter $\la$ with cuts along $\si(\xi)$. Let $\Omega(\la,\xi)$ and $\omega(\la,\xi)$ be the normalized Abel integrals of the second and the third kind  on the upper sheet of $\mathbb M(\xi)$, with zero $\mathfrak a$-periods along the gap $(\gamma(\xi), -1)$. The linear combination $g(\la,\xi)=\Omega(\la,\xi) + \xi \omega(\la,\xi)$ is then another Abel integral with zero $\mathfrak a$-period. The nominator of the function $\frac{\pa g(\la, \xi)}{\pa\la}$ has two real zeros $\nu(\xi)$ and $\mu(\xi)$ with at least one zero in the gap, say $\mu(\xi)\in (\gamma(\xi), -1)$.   Moreover, for $\la\to\infty$ the function $g(\la,\xi)$ has the same asymptotic behavior as the phase function $\Phi(z(\la), \xi)$ up to a constant term.  These properties hold for any choice of $\gamma(\xi)$. The peculiarity of our 
choice for $\gamma(\xi)$  when $\xi\in [\xi_{cr}^\prime, \xi_{cr})$ is that we require the second zero $\nu(\xi)$ of $g(\la, \xi)$ to match with $\gamma(\xi)$, that is, 
\be\label{gfunc}
\frac{\pa}{\pa \la}\left(\Omega(\la,\xi) + \xi\, \omega(\la,\xi)\right)=\frac{(\la -\mu(\xi))\sqrt{\la - \gamma(\xi)}}{\sqrt{(\la - b + 2a)(\la^2 -1)}}.
\ee
Such a point $\gamma(\xi)$ is unique for every $\xi$, and $\gamma(\xi)$ satisfies the same continuity and monotonicity properties as described above;  it defines the moving edge of the Whitham zone (cf.\ \cite{bk2}).
From this construction it follows that
\[
\xi_{cr}^\prime = - 2 a - \frac{\int_{b+ 2 a}^{-1}\la\
Q(\la)d\la}{\int_{b+ 2 a}^{-1}
Q(\la)d\la},\qquad Q(\la)=\sqrt{\frac{\la - b-2a}{(\la - b +2a)(\la^2 - 1)}}
\]
and $\si(\xi_{cr}^\prime)=[b-2a, b+2a]\cup[-1,1]$.
In contrast to the phase function $\Phi(z, \xi)$, the properties of $g(\la, \xi)$ allow us to apply the lens construction for the RHP approach to all contours whenever needed. For this reason we replace $\Phi(z, \xi)$ by $g(\la, \xi)$, which plays the role of the $g$-function (\cite{dvz}) in the NSD method. Given $\gamma(\xi)$,  let 
\be\label{shum}
\big\{\hat a\big(n,t,\xi\big),  \hat b\big(n,t,\xi\big)\big\}
\ee 
be the finite gap solution for the Toda lattice associated with the spectrum \eqref{sixi} and with an initial Dirichlet eigenvalue defined via the initial scattering data for \eqref{ini1}, \eqref{main} by the Jacobi inversion problem. This Dirichlet eigenvalue was computed in \cite[Equ.\ (5.25)]{emt14}.
The functions $\big\{\hat a\big(n,t,\tfrac{n}{t}\big),  \hat b\big(n,t,\tfrac{n}{t}\big)\big\}$ are then well defined in the region
\be\label{mathcalI}
\left\{(n,t)\in \Z\times \R_+: \tfrac{n}{t}\in [\xi_{cr}^\prime, \xi_{cr}-\varepsilon]\right\},
\ee
where $\varepsilon$ is arbitrary small.
In analogy to the KdV shock case, we call them modulated elliptic waves. They are  the main terms of the asymptotic expansion for the Toda shock wave with respect to large $t$ in the region \eqref{mathcalI}.

The middle region $\frac{n}{t}\in (\xi_{cr,1}^\prime, \,\xi_{cr}^\prime)$, where
\[
\xi_{cr,1}^\prime= b+ 1 - \frac{\int_{b+ 2 a}^{-1}\la \
Q_1(\la)d\la}{\int_{b+ 2 a}^{-1}
Q_1(\la)d\la},\qquad Q_1(\la)=\sqrt{\frac{\la +1}{((\la - b)^2 -4a^2)(\la - 1)}},
\]
 is  associated with the gap $(b+2a,-1)$, although we cannot claim that the stationary phase point of $\Phi(z(\la), \xi)$ for such $\xi$ is located
in this gap. A suitable $g$-function here is simply
$\Omega(\la) + \xi\, \omega(\la)$, where $\Omega(\la)$ and $\omega(\la)$ are the Abel integrals as defined above
associated with the spectrum 
\[
\si(\xi_{cr}^\prime)=\si(\xi_{cr,1}^\prime)=[b-2a, b+2a]\cup [-1, 1].
\]
The level line  $\re g(\la, \xi)=0$ in this case intersects the real axis at a  point $\la_0(\xi)$ inside the gap, which moves continuously along the gap when $\xi$ moves along 
$(\xi_{cr,1}^\prime, \,\xi_{cr}^\prime)$.  The main asymptotic term for the solution of \eqref{tl}--\eqref{main} is  the classical two band solution of the Toda lattice $\{\hat a(n,t), \hat b(n,t)\}$ with the initial Dirichlet eigenvalue depending on $\xi$ if the discrete spectrum inside the gap is nonempty. The phase summand in the 
theta function representation for this two band solution (cf.\ \cite[Equ.\ (9.48)]{tjac}) contains information on the initial 
scattering data for \eqref{tl}--\eqref{main}, and undergoes a shift when $\la_0(\xi)$ hits a point of the discrete spectrum. This agrees with the effect of adding a single eigenvalue as can be done using the double commutation  method
(cf.\ \cite{GzT} and  \cite[Lem.\ 11.26]{tjac}). 
For $\frac{n}{t}=\xi_{cr}^\prime$, the solution \eqref{shum} coincides with the two band solution $\{\hat a(n,t), \hat b(n,t)\}$ above. The same is valid for the second boundary of the middle region $\xi=\xi_{cr,1}^\prime$, because the construction of the $g$-function in $[\xi_{cr,1}^\prime, \xi_{cr,1})$ is the same as for the modulated elliptic waves above. It is associated with $[b-2a, b+2a]\cup [\gamma(\xi), 1]$ where $\gamma(\xi)\in[-1, 1).$ 

\subsection{Main result}
In \cite{emt14} we derived the precise formula for the modulated finite-gap solution \eqref{shum} using the NSD approach for vector RHPs and more restrictive initial data: we assumed that there are 
no resonances on the edges of the spectrum of $H(t)$ and that the discrete spectrum consists of a single point 
 in the spectral gap. We did not justify  the asymptotic expansion for the solution of \eqref{tl}--\eqref{main} and only conjectured that the next term is of order $O(t^{-1})$. The  aim of the present paper is to prove this fact by solving  local parametrix problems and finishing the conclusive analysis. We will implement a rigorous asymptotic analysis in the region \be\label{mathcalIres}
\mathcal D:=\{(n,t)\in \Z\times \R_+: \tfrac{n}{t}\in \mathcal I_{\varepsilon}:=[\xi_{cr}^\prime+\varepsilon, \xi_{cr}-\varepsilon]\}
\ee
to prove

\begin{theorem} \label{thor7}
For $(n, t)\in \mathcal D$, $n, t\to\infty$ uniformly with respect to $\frac{n}{t}\in\mathcal I_\varepsilon$, 
the Toda shock wave $\{a(n,t), b(n,t)\}$ given by \eqref{tl}--\eqref{main}, \eqref{decay} has 
the following asymptotic behavior:
\begin{align*}
a(n,t)^2 + a(n-1, t)^2 &= \hat a\big(n,t,\tfrac{n}{t}\big)^2 +\hat a\big(n-1,t,\tfrac{n}{t}\big)^2 + O(t^{-1}),\\
b(n,t) &=\hat b\big(n,t,\tfrac n t \big) +O(t^{-1}),
\end{align*}
where $\{\hat a(n,t, \xi), \hat b(n,t,\xi)\}$ is the finite gap solution of the Toda lattice associated with the spectrum 
$[b-2a, \gamma(\xi)]\cup [-1, 1]$ and the initial divisor $(\la(0,0), \pm)$ which is the only zero of the function
$\theta(2 A(z) - \frac{1}{2} - \frac{\Delta}{2\pi} \mid 2\tau)$ (see \eqref{defchi}, \eqref{Bla},  \eqref{Delta}, \eqref{defVell}, \eqref{Thetta}) on the Riemann surface $\mathbb M(\xi)$ with projection on the gap $[\gamma(\xi), -1]$.
\end{theorem}
For the remaining two regions the asymptotic analysis can be done similarly, see Sec.~\ref{sec:disc}.
However, we essentially improve the estimate on the error term in the middle region 
$\frac n t \in (\xi_{cr,1}^\prime, \xi_{cr}^\prime)$ in \cite{EM}, where we also describe the influence of resonances and
the discrete spectrum in the gap on the asymptotic.

\subsection{Remarks on the method of proof}
\begin{itemize}
\item
As in \cite{emt14} we deal with vector statements of RHPs. They are more natural in the Toda case than matrix statements, because the matrix statements are ill-posed for certain values of $n$ and $t$ in the class of invertible 
matrices with $L^2$-integrable singularities. This fact for Toda can be established similarly as for the KdV shock wave (\cite{EPT}). 
The vector statement requires additional symmetries to be posed on the contours, jump matrices 
and on the solution itself to guarantee uniqueness of the solution.  
\item
In \cite{emt14} the RHP was stated in  terms of the spectral variable $\la$, that is, on the two-sheeted Riemann surface with sheets glued along the cuts $[b-2a, b+2a]\cup [-1, 1]$. In the present paper we use the standard approach via the Joukowsky map $z(\la)$ in \eqref{spec5}: 
 the upper sheet of the Riemann surface is identified with the inner part of the circle $|z|<1$ without
the cut $[z(b-2a), z(b+2a)]$, and the lower sheet with $|z|>1$ without $[z^{-1}(b-2a), z^{-1}(b+2a)]$.
We formulate the initial RHP and reformulate all transformations leading to the model RHP in terms of $z$, taking into account the discrete spectrum and resonances, which produce  singularities in the jump matrix and  require a more sophisticated analysis and additional proofs of the uniqueness results. 
We also solve the vector model problem independently, and derive the asymptotics using a new, 
more convenient formula \eqref{asympnew}.
\item
To prove the asymptotics within the framework of the NSD, the traditional  approach first requires to solve the matrix analog of the model RHP,  then to find matrix solutions of the local parametrix RHPs, and finally to derive the singular integral equation for the error vector and to estimate its norm. 
But with this approach one fails to obtain uniform estimates in $n$ and $t$ for both KdV and Toda due to the singular behavior of the matrix model solution (\cite{EPT}). An alternative approach was proposed in \cite{p} for KdV: instead of constructing a matrix model solution, it evaluates the smallness of the difference between initial and model vector solutions as solutions of the associated singular integral equations with slightly different kernels. 
But this approach does not seem to work for Toda, because we have less control on the behavior of
the vector solutions $m(\la)=(m_1(\la), m_2(\la))$ of the initial and model RHPs at infinity. Indeed, for KdV one knows that $m(\la)\to (1, 1)$ as $\la\to \infty$, but for Toda we only know that $m_1(\la)m_2(\la)\to 1$, which is not sufficient to apply the technique of \cite{p}.

In \cite{min1, min2}  a  singular matrix model solution is proposed for the  KdV shock case, which has  a pole at 
$\la=0$, but the respective error vector does not have pole-like singularities. We use this idea to construct the matrix model solution for Toda shock.  It has simple poles at the edges of the right background spectrum, but the error vector does not (see Theorem \ref{properr}).

\item
To show that the expansion with respect to $z$ of the product of components of the initial RHP solution is asymptotically close to the expansion for the model RHP solution, we have to prove that the error vector  $\nu(z)$ (cf.\ \eqref{defnui}, \eqref{nusa}) is asymptotically close to the vector $(1,1)$ as $z\to 0$ up to a term $O(t^{-1})$. To achieve this, we carry out all conjugations and deformations related to NSD by strictly respecting the symmetry\footnote{Here $\si_1$ is the first Pauli matrix.} $m(z)=m(z^{-1})\sigma_1$ and normalization $m_1(0)m_2(0)=1$ for all subsequent RHPs. To preserve  
$\nu(z^{-1})=\nu(z)\si_1$ for the error vector, in the respective singular integral equation we have to use a special matrix Cauchy kernel with entries that have zeros at $z=0$ and $z=\infty$ (cf.\ \cite{KTb}, Equ.\ (B.8)). 
\end{itemize}

\subsection{Visualisation}
In Fig.~\ref{fig:shock} the numerically computed solution generated 
by pure step initial data $a(n,0)=\frac{1}{2}$, 
$b(n,0)=0$ as $n\geq 0$ and $a(n,0)=1$, $b(n,0)=-4$ as $n<0$ is plotted at a frozen time $t=799$ for $5000$ plotpoints around the origin.
\begin{figure}[ht]
\centering
\includegraphics[width=12.6cm]{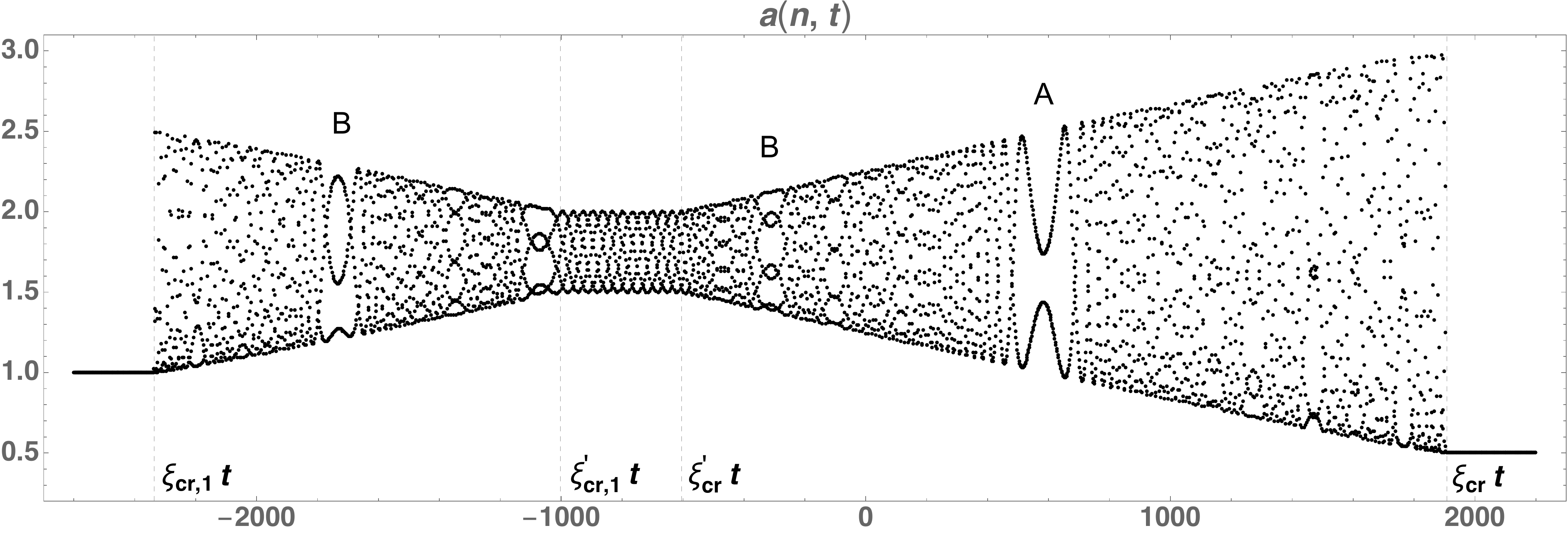}
\hfill
\includegraphics[width=12.6cm]{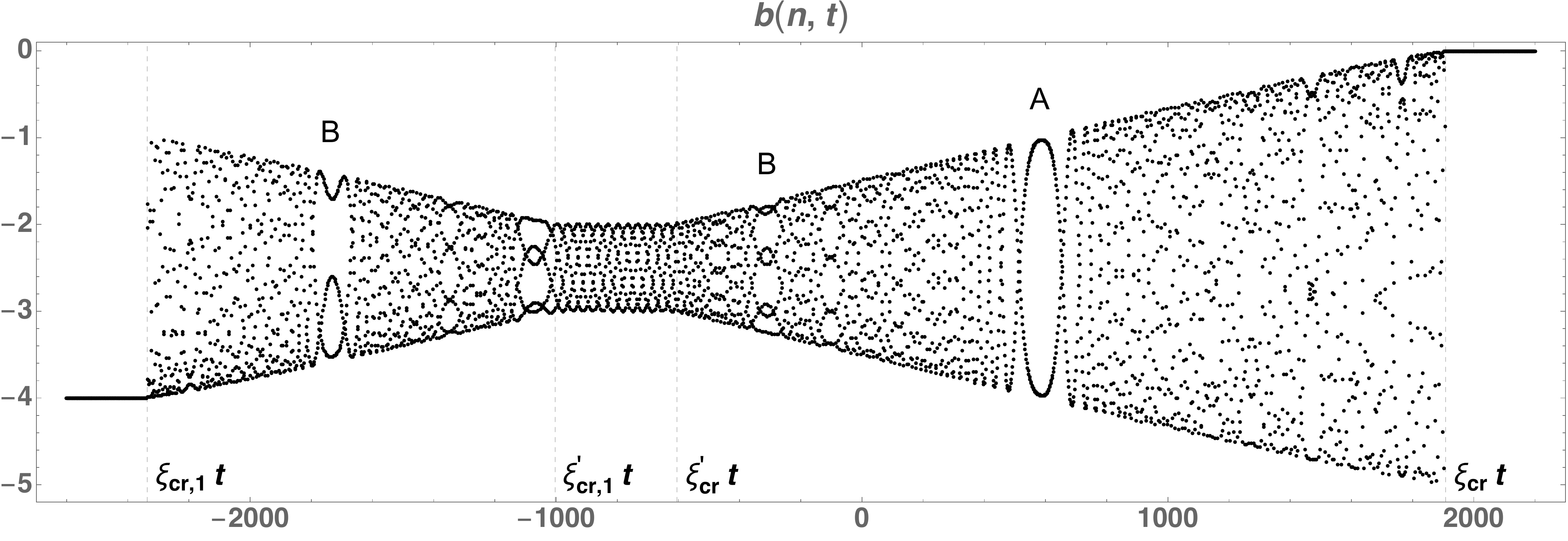}
\caption{{\small Solutions $a(n,t)$, $b(n,t)$ at $t=799$ for $\sigma(H_\ell)=[-6, -2]$ and  $\sigma(H_r)=[-1,1]$.
The critical values (times $t$) are plotted as vertical lines,
$\xi_{cr}t=1907.65$, $\xi_{cr}^\prime t =-604.39$, $\xi_{cr,1}^\prime t =-1002.66$, and $\xi_{cr,1} t =-2336.92$.}} \label{fig:shock}
\end{figure}
We can clearly distinguish the left and right regions corresponding to the modulated elliptic waves
and the middle region, where the quasi-periodic finite gap solution is associated with the full spectrum 
$[-6, -2]\cup [-1,1]$. Another phenomenon is nicely visualized in areas as for example $A$, $B$:
Recall that quasi-periodicity or pure periodicity of the finite gap solution with spectrum $[b-2a, b+2a]\cup [-1,1]$ 
depends on the ratio $r$ of the frequencies of its quasi-momentum $\omega(\la)$,
\[r=\frac{\omega(1)-\omega(-1)}{\omega(b+2a)-\omega(b-2a)}.
\]
If $r=\frac{p}{q}\in \mathbb Q$, where $\frac{p}{q}$ is an irreducible fraction, then the finite gap solution is periodic with period $p+q$ (for $a=\frac{1}{2}$ we get a solution of period~$2$). 
If $\omega(\la,\xi)$ is the quasi-momentum associated with $[b-2a, \gamma(\xi)]\cup [-1, 1]$ or $[b-2a, b+2a]\cup [\gamma(\xi), 1]$ and 
\[r(\xi)=\frac{\omega(1,\xi) - \omega(-1,\xi)}{\omega(\gamma(\xi), \xi) - \omega(b-2a, \xi)}\quad \mbox{or} \quad r(\xi)=\frac{\omega(1,\xi) - \omega(\gamma(\xi),\xi)}{\omega(b+2a, \xi) - \omega(b-2a, \xi)}
\]
is rational, then for such a point $\xi$  the solution is periodic; the period might be large as in area $B$. Area $A$ 
corresponds to a modulated wave near $\xi_0$ such that $\gamma(\xi_0)=-4$, where the solution is of period $2$.

\section{Statement of the Riemann--Hilbert problem}\label{s:RHP}

In this section we cover some basic facts of the inverse scattering transform and fix notation. For a detailed account of
scattering theory for Jacobi operators with steplike backgrounds see \cite{emtstp2, emt3, emt5},
with zero background see \cite[Chapter 10]{tjac}.

Under the assumption that the coefficients of the initial data \eqref{ini1} tend to the background constants sufficiently fast\footnote{For example, with finite first moments of perturbation.}, the spectrum of the Jacobi operator $H(t)$ consists of an absolutely continuous part of multiplicity one,
\[
\sigma_{ac}(H) = [b - 2 a, b+2a] \cup [-1,1],
\]
plus a finite simple pure point part,
$$
\{\lambda_j : j=1,\dots, N\} \subset \R \setminus \sigma_{ac}(H).
$$
To simplify further considerations we assume that the initial data \eqref{ini1} decay to their
backgrounds exponentially fast
\be \label{decay}
 \sum_{n = 1}^{\infty} \E^{\rho n} \big( |a(-n,0) - a| + |b(-n,0)-b|
+ |a(n,0) - \tfrac{1}{2}| + |b(n,0)|\big) < \infty,
\ee
where $\rho>0$ is a positive number.
The operator $H(t)$ is self-adjoint and the  diagonal elements of its Green's function $G(\la, n, m, t)$ (that is, the kernel of the resolvent operator $(H(t)-\la \id)^{-1}$) have the following expansion as $\la\to\infty$ (\cite[Sec.\ 6.1]{tjac})
\be\label{gerald1}
G(\la, n,n,t)=-\frac{1}{\la}\left( 1+\frac{b(n,t)}{\la} +\frac{a(n,t)^2 + a(n-1, t)^2 + b(n,t)^2}{\la^2} + O(\la^{-3})\right).
\ee

As mentioned in the introduction, instead of the spectral parameter $\la$ we use its Joukowsky transformation,
\[
z(\la)=\la-\sqrt{\la^2-1},
\]
which maps the two sides of the cut along the interval $[-1,1]$ to the unit circle $\mathbb T=\{z: |z|=1\}$.
The map $z\mapsto \la$ is one-to-one between the closed domains $\clos(\mathcal Q)$ and  $\clos(\C \setminus \sigma_{ac}(H(t)))$, where\footnote{
We define the closure by adding the upper and lower sides of the cuts as distinct points to the boundary.} 
\[ 
\mathcal Q :=\{z: |z|<1\}\setminus [q_1, q],
\]
The points $q_1=z(b+2a)$ and $q=z(b-2a)$
correspond to the edges of $\sigma(H_\ell)$ and $z=-1$ and $z=1$
correspond to the edges of $\sigma(H_r)$. The eigenvalues $\la_j$ are mapped to $z_j \in \left((-1, 0)\cup (0,1)\right) \setminus [q_1, q]$, for $j=1,\dots,N$; we denote them by 
\[
\si_d=\{z_j,  j=1,...,N\}.
\]
In addition to $z$ we also use the Joukowsky transformation $\zeta=\zeta(\la)$ associated with the left background and given by \eqref{spec5}.

Recall that the Jacobi equation \eqref{ht} has two Jost solutions $\psi(z,n,t)$, $\psi_\ell(z,n,t)$ for each $z \in \mathcal Q$
with asymptotic behavior
\[
\lim_{n\to \infty} z^{-n}\psi(z,n,t) =1,\quad |z|\leq 1; \qquad
\lim_{n\to -\infty} \zeta^{n}\psi_\ell(z,n,t) =1, \quad |\zeta|\leq 1.
\]
As functions of $z$ they have  slightly different properties on $\mathcal Q$. Indeed, since $\zeta^n \psi_\ell(z,n,t)$ is in fact an analytic function of $\zeta$ as $|\zeta|<1$,
this function has complex conjugated values on the sides of the cut along $[q_1, q]$, which we denote as $[q_1, q]\pm\I0$. It has equal real values at $z, z^{-1}\in\mathbb T$. The function $\psi(z,n,t)$ has complex conjugated values at conjugated points of $\mathbb T$, but \[\psi(z-\I0,n,t)=\psi(z+\I0,n,t)\in\R, \ \ \mbox{for}\ \ z\in [q_1, q].\]
Recall that $\psi_\ell(z,n,0)$ admits a representation via the transformation operator
\[ \psi_\ell(z,n,0)=\sum_{-\infty}^n K(n,m)\zeta^{-m}, \quad |\zeta|\leq 1.\]
Under  condition \eqref{decay} in the domain $1\leq |\zeta|< \E^\rho$ there exists an analytic function which is an extension of 
$\ol{\psi_\ell}$,
\be\label{uh} 
\breve\psi_\ell(z,n):=\sum_{-\infty}^n K(n,m)\zeta^{m}, \ \ \breve\psi_\ell(z\pm\I 0,n)=\overline{\psi_\ell(z\pm \I 0,n,0)},\quad z\in [q_1, q].
\ee 
The Jost solutions of \eqref{ht} are connected by the scattering relation
\be\label{pst}
T(z,t)\psi_\ell(z,n,t)=\ol {\psi(z,n,t)} + R(z,t)\psi(z,n,t), \quad |z|=1,
\ee
where $R(z,t)$ and  $T(z,t)$ are the right reflection and transmission coefficients. Their time evolution is given by
\[
R(z, t) = R(z)e^{(z-z^{-1})t}, \quad z \in \T, \qquad  |T(z, t)|^2 = |T(z)|^2 e^{(z-z^{-1})t}, \quad z \in [q_1, q],
\]
where $R(z)=R(z,0)$, $T(z)=T(z,0)$.
The right norming constants 
\[
\gamma_j(t)=\left(\sum_{n\in\Z}\psi^2(z_j,n,t)\right)^{-2}
\]
corresponding to $z_j\in\si_d$ evolve as 
$\gamma_j(t)=\gamma_j e^{(z_j-z^{-1}_j)t}$, $\gamma_j=\gamma_j(0)>0$. Let 
$$
W(z,t)=a(n-1,t)(\psi_\ell(z,n-1,t)\psi(z,n,t) - \psi_\ell(z,n,t)\psi(z,n-1,t))
$$ 
be the Wronskian of the Jost solutions and define $W(z):=W(z,0)$.

{\bf Resonant points}. {\it The point $\tilde q \in \{ -1, 1, q, q_1\}$ is called a resonant point if $W(\tilde q)=0$. If $W(\tilde q)\neq 0$, then $\tilde q$ is non-resonant. Note that $W(\ti q, t) =0$ iff $W(\ti q) =0$, that is, the property of being resonant (or not) is preserved with $t$.}
 
Under a much weaker decaying condition than \eqref{decay}, namely a finite first moment of perturbation
\be \label{decay1}
 \sum_{n = 1}^{\infty} n \big( |a(-n,0) - a| + |b(-n,0)-b|
+ |a(n,0) - \tfrac{1}{2}| + |b(n,0)|\big) < \infty,
\ee  
the set of the associated right initial  scattering data
\be\label{scatt}
\{R(z), z\in\mathbb T; \chi(z), z\in [q_1, q]; (z_j, \gamma_j), z_j\in\si_d \},
\ee
where
\be\label{defchi}
\chi(z)=-2a\frac{\zeta(z-\I 0)-\zeta^{-1}(z-\I 0)}{z-z^{-1}} |T(z)|^2, \quad z\in [q_1, q],
\ee
defines the solution of the Cauchy problem \eqref{tl}--\eqref{main} uniquely. For each $t$ this solution also has finite first moments of perturbations (\cite{emt3}). The scattering data \eqref{scatt} satisfy the following properties (\cite{dkkz, emtstp2}):
\begin{itemize}
 \item The function $R(z)$ is continuous on $\T$ and $R(z^{-1})=\ol{R(z)}=R^{-1}(z)$ for $z\in\T$. If $z=-1$ is non-resonant,  then $R(-1)=-1$, and if $z=-1$ is resonant, then $R(-1)=1$.
\item The function $T(z)$ can be restored uniquely for $z\in\mathcal Q$ from the data \eqref{scatt};  it is meromorphic with simple poles at $z_j$. 
\item The function $\chi(z)$ is continuous for $z\in (q_1, q)$ and vanishes at $\ti q \in \{q,q_1\}$ with
\be\label{chinon}\chi(z)=C(z-\ti q)^{1/2},\quad z\to\ti q \in \{q,q_1\},
\ee if
$\ti q$  is a non-resonant point. If $\ti q$ is a resonant point, then \be\label{chirez}\chi(z)= C(z-\ti q)^{-1/2} (1 + o(1)), \quad z \to \ti q \in \{q,q_1\}.
\ee
\end{itemize}

To apply the nonlinear steepest descent approach in the most general situation which assumes resonances, we 
choose the number $\rho>0$  in \eqref{decay} small such that 
\be
\label{nu} \rho>-\log|q|
\ee
in order to have the inclusion  $[q_1, q]\subset \{z:\,\E^{-\rho}<|z|<1\}$.
Under condition \eqref{decay} the scattering data have additional properties:
\begin{itemize}
\item
The function $R(z)$ admits an analytic continuation to $\{z:\E^{-\rho}<|z|<1\}\setminus [q_1, q]$ with simple poles at the points of the 
discrete spectrum located in this domain.
\item The function $\chi(z)$ has an analytic continuation $X(z)$ in a vicinity of $[q_1, q]$ with 
\[\chi(z)=\I|\chi(z)|=X(z-\I 0), \quad z\in [q_1, q],
\]
where
\be\label{Chi} X(z)= -\frac{a (\zeta - \zeta^{-1})(z-z^{-1})}{2 W(z)\, W(\breve\psi_\ell, \psi)(z)}.
\ee 
Here $\breve\psi_\ell(z,n)$ is defined by \eqref{uh} and $W(\breve\psi_\ell, \psi)(z)$ is the Wronskian of $\breve\psi_\ell(z,n,0)$ and $\psi(z,n,0)$. 
\end{itemize}

Treating the values $n$ and $t$ as parameters, we define a vector-valued function $m(z)=(m_1(z,n,t), m_2(z,n,t))$ on $\mathcal Q$ by
\be \label{defm}
m(z,n,t) =
\begin{pmatrix} T(z,t) \psi_{\ell}(z,n,t) z^n,  & \psi(z,n,t)  z^{-n} \end{pmatrix}.
\ee
The first component $m_1(z)$ is a meromorphic function in $\mathcal Q$ with poles at $z_j$. It has continuous limits as $z$ approaches the boundary of $\mathcal Q$ except (possibly) at $q$ and $q_1$, where a square root singularity may appear in the case of resonance.   The  second component of this vector is a holomorphic function in $\mathcal Q$ with continuous limits to the boundary.  Both functions have finite positive limits as $z\to 0$  (cf.\ \cite{emt14}). For our purpose it will be sufficient to control the product of the components.  
\begin{lemma}\label{asypm} For $z\to 0$, 
 \begin{align} \label{asympnew} \nonumber
& m_1(z,n,t) m_2(z,n,t) = 1 + 2 z b(n,t) \\
& \qquad + 4 z^2 \Big( a(n-1,t)^2 + a(n,t)^2 + b(n,t)^2 - \frac 12 \Big) + O(z^3).
\end{align}
\end{lemma}
\begin{proof} The Jost solutions $\psi$ and $\psi_\ell$ can be considered as the Weil solution of $H(t)$, and therefore the Green's function \eqref{gerald1}
considered as a function of $z$ can be represented as
\[G(\la(z),n,n,t)=\frac{\psi(z,n,t)\psi_\ell(z,n,t)}{W(z,t)}.
\]
Recall that \[T(z,t)=\frac{z-z^{-1}}{2 W(z,t)},\] that is,
\be\label{imp98}
m_1(z,n,t)m_2(z,n,t)=\frac{z-z^{-1}}{2}\,G(\la(z), n, n, t).
\ee
Taking into account that
$
\frac{1}{\la}= \frac{2 z}{1+z^2}$ and $\frac{z- z^{-1}}{2}=-\sqrt{\la^2 - 1}$,
we obtain \eqref{asympnew}.
\end{proof}
Let $\mathcal Q^*:=\{z:|z|>1\}\setminus [q^{-1}, q_1^{-1}]$ be the image of the domain $\mathcal Q$  under the map
$z \mapsto z^{-1}$.  We extend  $m$ to $\mathcal Q^*$ by $m(z^{-1}) = m(z) \si_1$,
where $\si_1=\left(\begin{smallmatrix} 0  & 1 \\
1 & 0 \end{smallmatrix} \right)$ is the first Pauli matrix. With this extension,  the second component $m_2(z)$ is a meromorphic function on $\mathcal Q^*$
with poles at $z_j^{-1}$, $z_j\in\si_d$, and $m_1(z)$ is holomorphic.
Being defined now on $\C\setminus \Sigma$, 
where \be\label{Gammac}\Sigma=\T\cup [q,q_1]\cup [q_1^{-1}, q^{-1}],\ee the function $m(z)$ can have jumps along $\Sigma$.
For convenience, from here on we encode the orientation of the contours in $\R$ as follows: assume  $-\infty\leq c<d\leq\infty$, then we write $[d,c]$ for the
interval $[c,d]$ with orientation right-to-left. In particular, the contours $[q, q_1]$, and $[q_1^{-1}, q^{-1}]$ in \eqref{Gammac} are  oriented right-to-left, and  the unit circle $\T$ is oriented counterclockwise.

Throughout this paper, plus $(+)$ and minus $(-)$ sides of a contour correspond to the left and right
sides by orientation, that is, the $+$ side of an oriented contour lies to the left as
one traverses the contour in the direction of its orientation. And
$m_\pm(z)$ denote the boundary values of $m(z)$ as $z$ tends to the contour from the $\pm$ side.
Using this notation implicitly assumes that the limit exists (in the sense that $m(z)$ extends to a
continuous function on the boundary except  probably at a finite number of points). In this paper, all contours are symmetric with respect to the map
$z\mapsto z^{-1}$, i.e.\ they contain with each point $z$ also $z^{-1}$. The symmetric part of the contour will be denoted by the same letter, the image of  a contour $\mathcal L\subset \{z: |z|<1\}$ is denoted by $\mathcal L^*$. 
Given the orientation on $\mathcal L$, the orientation on  the starred contour $\mathcal L^*$   can be chosen in two ways. For the convenience of tracking this orientation we use the following formal notation. If the points $z^{-1}$ and $z$ simultaneously move in the positive direction of $\mathcal L$ and $\mathcal L^*$, we encode this as $\mathcal L^* \,\uparrow\,\uparrow \,\mathcal L$. If $z^{-1}$ moves in the negative direction while $z$ moves in the positive direction, 
we use the notation $\mathcal L^* \,\downarrow\,\uparrow\, \mathcal L$. In particular, $[q_1^{-1}, q^{-1}]\,\downarrow\,\uparrow [q, q_1]$.   The following symmetry should be preserved
for the jump matrix of any vector RHP and for its solution.

{\bf Symmetry condition}. {\it  Let $\hat\Sigma$ be a symmetric oriented contour and $\mathcal L\cup \mathcal L^*\subset (\C\setminus\mathbb T)$ be any symmetric part. The jump matrix $v(z)$ of the
vector problem $m_+(z)=m_-(z)v(z)$, $z\in \Sigma$, satisfies 
\[\aligned
v(z)&=\sigma_1 (v(z^{-1}))^{-1} \sigma_1,\quad z\in\mathcal L\cup \mathcal L^*, \ \mbox{for}\ \mathcal L^* \,\downarrow\,\uparrow\, \mathcal L,\\
v(z)&=\sigma_1 v(z^{-1}) \sigma_1,\qquad z\in\mathcal L\cup \mathcal L^*, \ \mbox{for}\ \mathcal L^* \,\uparrow\,\uparrow\, \mathcal L.
\endaligned\]
If $\mathbb T\subset \Sigma$ then $v(z)=\sigma_1 (v(z^{-1}))^{-1} \sigma_1,$ $ z\in \mathbb T.$
Moreover,}
\be\label{symto}
m(z)=m(z^{-1})\sigma_1, \quad z\in\C\setminus\hat\Sigma.
\ee
To preserve the symmetry condition we will always choose symmetric deformations of the contours. 
Moreover, we will only use conjugations by diagonal matrices which respect the symmetry condition, 
as outlined in the next lemma.
\begin{lemma}[Conjugation, \cite{KTa}] \label{lem:conjug}
Let $m$ be the solution on $\C$ of the RH problem  $m_+(z)=m_-(z) v(z)$, $z \in \hat\Sigma$,
which satisfies the symmetry condition.
 Let $d: \C\setminus  \tilde\Sigma\to\C$ be a sectionally analytic function with jump on a symmetric contour 
 $\tilde\Sigma\subset \hat\Sigma$. Set
\be\label{mD}
\ti{m}(z) = m(z) \begin{pmatrix} d(z)^{-1} & 0 \\ 0 & d(z) \end{pmatrix}= m(z) [d(z)]^{-\sigma_3}, \quad 
\sigma_3= \begin{pmatrix} 1 & 0 \\ 0 & -1 \end{pmatrix}.
\ee
If $d$ satisfies   \[d(z^{-1}) = d(z)^{-1},\quad z \in \C\setminus \Sigma,\]
then \eqref{mD} respects the symmetry condition.
The jump matrix of $\ti{m}_+=\ti{m}_- \ti{v}$ is given by
\[
\ti{v} =
\left\{
\begin{array}{ll}
  \begin{pmatrix} v_{11} & v_{12} d^{2} \\ v_{21} d^{-2}  & v_{22} \end{pmatrix}, &
  \quad z \in \hat\Sigma \setminus  \Sigma, \\[3mm]
  \begin{pmatrix} \frac{d_-}{d_+} v_{11} & v_{12} d_+ d_- \\
  v_{21} d_+^{-1} d_-^{-1}  & \frac{d_+}{d_-} v_{22} \end{pmatrix}, &
  \quad z\in \Sigma.
\end{array}\right.
\]
\end{lemma}
The symmetry constraints described above will allow us to shorten notations and computations on starred parts of the contours. Indeed, if we know that $d(z)$ has a jump on $\mathcal L\cup \mathcal L^*$, with $d_+(z)=d_-(z) s(z)$ on $\mathcal L$, then the property $d(z^{-1})=d^{-1}(z)$ used in a vicinity of $\mathcal L$ gives a complete information about the jump on $\mathcal L^*$. The same is true for the jump matrices. 

\noprint{In particular,  the contour \eqref{Gammac} is symmetric with respect to $z\to z^{-1}$. To satisfy the symmetry condition for the extended function \eqref{defm}  we continue  \eqref{defchi} to $[q_1^{-1}, q^{-1}]$ as an odd function, $\chi(z)=-\chi(z^{-1})$.  This implies that we extend 
$X(z)$ given by \eqref{Chi} as an even function in a vicinity $\mathcal O_\rho$ of  $[q, q_1]\cup [q_1^{-1}, q^{-1}]$ (which is 
possible by the condition on $\rho$ in \eqref{nu}),
\be\label{Chi1}
X(z^{-1})=X(z),\quad z\in \mathbb C\setminus(\Sigma\cup\si_d\cup\si_d^*)\cap\mathcal O_\rho.
\ee
}
In $\C\setminus (-\infty, 0]$ introduce the  {\it phase function} 
\[
\Phi(z):=\Phi(z,\xi)=\frac{1}{2} \big(z - z^{-1}\big) + \xi\log z, \quad \xi:=\frac{n}{t},
\]
which is odd with respect to $z\to z^{-1}$, that is, $\Phi(z^{-1})=-\Phi(z)$. Note that $\E^{2 t\Phi(z)}=z^{2n}\E^{t(z - z^{-1})}$ is well defined in $\C\setminus\{0\}.$
The vector function  \eqref{defm} extended to $\mathcal Q^*$ by symmetry \eqref{symto}  solves  the following RHP (cf.\ \cite{dkkz, KTb,  EPT}):

\begin{RHproblem}[Initial meromorphic RHP statement] \label{RH1} 
Find a vector-valued function $m: \C \setminus \Sigma \to \C^{1\times2}$ which is  meromorphic in $\mathcal Q\cup \mathcal Q^*$ and continuous up to $\Sigma$ except at possibly the points $q, q_1, q^{-1},q_1^{-1}$. It has simple poles at $z_j^{\pm1}$, $j=1, \dots, N$, and satisfies:
\begin{itemize}
	\item  the jump condition $m_{+}(z)=m_{-}(z) v(z)$,  where
\[ 
v(z)=\left\{
\begin{array}{ll}
\begin{pmatrix}
0 & - \ol{R(z)} \E^{- 2 t \Phi(z)} \\
R(z) \E^{2 t \Phi(z)} & 1
\end{pmatrix}, & \quad z \in \T,\\
\begin{pmatrix}
1 & 0 \\
\chi(z) \E^{2t\Phi(z)} & 1
\end{pmatrix}, & \quad z \in [q, q_1],\\ [3mm]
\sigma_1 (v(z^{-1}))^{-1}\si_1, & \quad z \in  [q_1^{-1}, q^{-1}];
\end{array}\right.
\]
\item the residue conditions
\be \nn
\aligned
\res_{z=z_j} m(z) &= \lim_{z\to z_j} m(z)
\begin{pmatrix} 0 & 0\\ - z_j \gamma_j \E^{2 t\Phi(z_j)}  & 0 \end{pmatrix}, \quad j=1, \dots, N,\\
\res_{z=z_j^{-1}} m(z) &= \lim_{z\to z_j^{-1}} m(z)
\begin{pmatrix} 0 & z_j^{-1} \gamma_j \E^{2 t\Phi(z_j)} \\ 0 & 0 \end{pmatrix}, \quad j=1, \dots, N;
\endaligned
\ee
\item
the symmetry condition $m(z^{-1}) = m(z) \si_1$.
\item
the normalization condition $m_1(0) \cdot m_2 (0)= 1$ and $m_1(0) > 0$.
\item the resonant/non-resonant condition:
 If  $\chi(z)$ satisfies \eqref{chinon} at $\ti q $ then $m(z)$ has  finite limits  $m(\ti q^{\pm 1})\in \R^{1\times 2}$  as $z\to \ti q^{\pm 1}$, $\ti q\in \{q, q_1\}$. If \eqref{chirez} is fulfilled then\be\label{rezon}\aligned m(z)& = \left(\frac{C_1}{(z-\ti q)^{1/2}}, \, C_2\right)(1 + o(1)),\ \ C_1 C_2\neq 0, \ \mbox{or}\ \\  m(z)&=(C_1, C_2(z-\ti q))(1 +o(1)), \quad z\to \ti q, \quad C_1 C_2\neq 0.\endaligned\ee 
Respectively, at $\ti q^{-1}$ the analog of \eqref{rezon} holds by symmetry \eqref{symto}.
\end{itemize}
\end{RHproblem}

\begin{lemma} \label{th:RH1}
Suppose that the initial data of the Cauchy problem \eqref{tl}--\eqref{main} satisfy \eqref{decay1} and
let \eqref{scatt}  be the associated initial right scattering data.  Then the vector function $m(z)=m(z,n,t)$
defined by \eqref{defm}, \eqref{symto} is the unique solution of RHP~\ref{RH1}.
\end{lemma}
The proof of uniqueness is completely analogous to the KdV shock case (\cite{EPT}).

  \noprint{
Note that vector function \eqref{defm} has the following asymptotic behavior near point $z=0$ (see \cite{emt14}):
\be \label{asm1}
\aligned
m_1(z,n,t) & = \prod_{j=n}^\infty 2a(j,t) \Big(1 + 2 z \sum_{m=n}^\infty b(m,t)\Big) + O(z^2),\\
m_2(z,n,t)&=\prod_{j=n}^\infty \big(2a(j,t)\big)^{-1}\Big(1 - 2z \sum_{m=n+1}^\infty b(m,t)\Big)
+  O(z^2).
\endaligned
\ee
The uniqueness lemma guarantees that solving RH-1 and expanding the solution at $z=0$ we can get asymptotics for the solution of \eqref{tl}--\eqref{main}. In fact, for recovering $\{b(n,t)\}$ we use the more convenient formula
\begin{lemma}\label{asypm} For $z\to 0$ the following asymptotics holds:
 \begin{align} \nonumber
& m_1(z,n,t) m_2(z,n,t) = 1 + 2 z b(n,t) \\
& \qquad + 4 z^2 \Big( a(n-1,t)^2 + a(n,t)^2 + b(n,t)^2 - \frac 12 \Big) + O(z^3).
\end{align}
\end{lemma}
\begin{proof} First of all we observe that 
 \eqref{spec5} is equivalent to the equality
 \be\label{equ}2az\zeta^{-1} = 1 - 2bz - 2a\zeta z + z^2,\ee from which it follows
\be \label{zetaz}
\zeta = \zeta(z) = 2 a z + 4 a b z^2 + O(z^3), \quad z \rightarrow 0.
\ee
Recall now ( \cite[Equ.\ (10.28)]{tjac}) that

\[
T(z) = \frac{z^{-1} - z}{2 a(n) \big(\psi(z,n)\psi_\ell(z,n+1) - \psi(z,n+1)\psi_\ell(z,n) \big)}\to T(0)\neq 0,
\]
therefore
\[
m_1(z,n) m_2(z,n) = T(z) \psi_{\ell}(z,n) \psi(z,n)
= \frac{1-z^2}
{2z a(n) \Big(\frac{\psi_\ell(z,n+1)} {\psi_{\ell}(z,n)} - \frac{\psi(z,n+1)} {\psi(z,n)} \Big)}.
\]
We have therefore to expand the function 
\[
P(z) := 2z a(n) \Big(\frac{\psi_\ell(z,n+1)} {\psi_{\ell}(z,n)} - \frac{\psi(z,n+1)} {\psi(z,n)}\Big)
\]
up to $O(z^3)$ as $z\rightarrow 0$. It is evident from \eqref{asm1} that
\[
2z a(n) \frac{\psi(z,n+1)} {\psi(z,n)} = 4 a(n)^2 z^2 + O(z^3).
\]
By use of the transformation operators we get then
\[
P(z) = \frac{2 a z \zeta^{-1} + 2a K_{1}^-(n+1) z + 2a K_{2}^-(n+1) z \zeta}{1 + K_{1}^-(n) \zeta + K_{2}^-(n)  \zeta^2} - 4 a(n)^2 z^2 + O(z^3).
\]
Thus,
\[
P(z) = 1 - 2 b(n) z + z^2 \big(1 -  4 a(n)^2 - 4 a(n-1)^2 \big) + O(z^3).
\]
Therefore
\begin{align*}
& m_1(z,n) m_2(z,n) = (1-z^2) G(z)^{-1} \\
&\quad = 1 + 2 b(n) z + 4 z^2 \Big(a(n-1)^2 + a(n)^2 + b(n)^2 - \frac 12 \Big) + O(z^3).
\end{align*}
\end{proof}
}

The behavior of solutions of such RHPs is determined mostly by the behavior
of the real part of the phase function $\Phi(z, \xi)$  which depends on the value of the parameter $\xi=\tfrac{n}{t}$. 
The signature table of $\re \Phi(z, \xi)$  for the region  \eqref{mathcalIres} is depicted in
Fig.~\ref{fig:Phi}.
\begin{figure}[ht] 
\begin{tikzpicture}[scale=0.8]

\fill[black!5] (4.5,-2) circle (2.8cm);
\path [fill=white] [out=90, in=110] (3,-2) to (4.5, -2) [out=-90, in=250] (3,-2) to (4.5, -2);
\draw[thick] (4.5,-2) circle (2.8cm);
\path [fill=black!5] (-3.5,0.5) .. controls (-0.5,-1.5) and (-0.5,-2.5) .. (-3.5,-4.5) -- (-3.5,0.5); 

\draw [help lines,->] (-3.5,-2) -- (8.5,-2) coordinate (xaxis);
 \draw [help lines,->] (4.5,-5) -- (4.5,1) coordinate (yaxis);

\draw[thick] (2.5,-2) -- (3.8,-2);
\draw[thick] (-2.6,-2) -- (-0.5,-2);

\draw[thin] (4.2,-2) circle (0.10cm);
\filldraw (4.2,-2) circle (0.5pt); 
\draw (-3,-2) ellipse (0.12cm and 0.1cm); 
\filldraw (-3,-2) circle (0.5pt) node[above]{\small{${\T_j^*}$}};

\draw[thin] (2.1,-2) circle (0.1cm);
\filldraw (2.1,-2) circle (0.5pt) node[above]{\small{${\T_k}$}};
\draw (1,-2) ellipse (0.12cm and 0.1cm); 
\filldraw (1,-2) circle (0.5pt) node[above]{\small{${\T_k^*}$}};

\draw[gray, thick, densely dotted] (-3.5,0.5) .. controls (-0.5,-1.5) and (-0.5,-2.5) .. (-3.5,-4.5);
\draw[out=90, in=110, gray, thick, densely dotted] (3,-2) to (4.5, -2);
\draw[out=-90, in=250, gray, thick, densely dotted] (3,-2) to (4.5, -2);

\node at (-2.6,-1.1) {$\re \Phi < 0$};
\node at (3.5,-0.5) {$\re \Phi < 0$};
\node at (0,-0.5) {$\re \Phi > 0$};
\node at (6.7,0.2) {$\T$};
\filldraw (3.8,-2) circle (1pt) node[below]{\small{$q$}};
\filldraw (2.5,-2) circle (1pt) node[below]{\small{$q_1$}};
\filldraw[darkgray] (3,-2) circle (1pt) node[below]{\small{$z_0$}};
\filldraw (4.5,-2) circle (1pt) node[above right]{${0}$};
\filldraw (-0.5,-2) circle (1pt) node[below]{\small{$q_1^{-1}$}};
\filldraw (-2.6,-2) circle (1pt) node[below]{\small{$q^{-1}$}};
\filldraw[darkgray] (-1.25,-2) circle (1pt) node[above right]{\small{$z_0^{-1}$}};

\draw [->, thick] (3.4, -2) -- (3.3,-2);
\draw [->, thick] (-1.8, -2) -- (-1.9,-2);
\draw [->, thick] (4.5, 0.8) -- (4.49,0.8);
\draw [->, thin] (4.18,-1.9) -- (4.17,-1.9);
\draw [->, thin] (-3.03,-1.9) -- (-3.04,-1.9);
\draw [->, thin] (2.08,-1.9) -- (2.07,-1.9);
\draw [->, thin] (0.98,-1.9) -- (0.97,-1.9);

\end{tikzpicture}
\caption{Signature table for $\re \Phi(z, \xi)$ for $\xi \in (\xi^\prime_{cr},\xi_{cr})$.} \label{fig:Phi}
\end{figure}
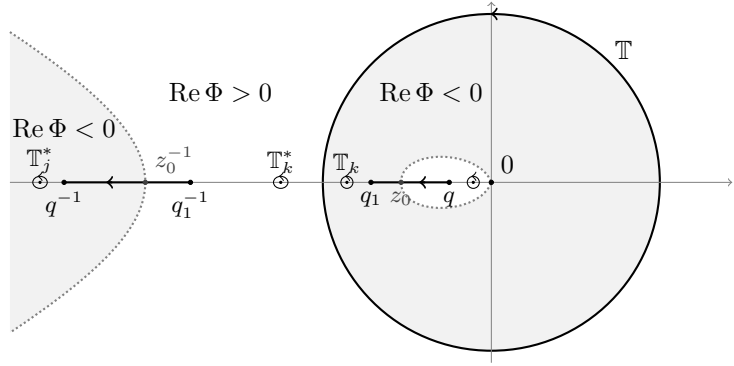
One part of the eigenvalues in $\mathcal Q$ lies in the set $\re\Phi(z)>0$ (namely $z_j\in (0,q)$),
while the remaining eigenvalues belong to the domain $\re\Phi(z)<0$ ($z_k\in (-1, q_1)\cup (0,1)$).
As outlined in \cite{dkkz, KTa}, one can redefine\footnote{An associated transformation on $\mathcal Q^*$   follows immediately from the symmetry condition.} $m(z)$ on $\mathcal Q$ by conjugating it with an invertible bounded matrix-function such that the residue conditions at $z_j \in \si_d$ are replaced by jump conditions along non-intersecting small circles around points of $\si_d$. The respective jump matrices will be exponentially close as $t\to\infty$ to the unit matrix for all further transformations of RHP~\ref{RH1}. By this transformation, the main contour $\Sigma$ is not changed and the structure of the jump matrices there remains qualitatively the same as in RHP~\ref{RH1} with respect to
decay/oscillation in $t$ and symmetry. 

Indeed, let $\epsilon>0$ be sufficiently small such that the circles $\mathbb T_j=\{ z : |z - z_j|=\epsilon \}$, $z_j\in \si_d$,
 do not intersect, do not contain the origin, and lie away from $\T \cup [q, q_1]$ (the precise value of $\epsilon$ will be chosen later).
Denote their images under the map $z\mapsto z^{-1}$ by $\mathbb T_j^*$. We orient  $\T_j$ and $\T_j^*$ counterclockwise, that is, $\T_j^*\,\uparrow\,\uparrow\, \T_j$.
Note that the curves $\mathbb T_j^*$
are not circles, but they surround $z_j^{-1}$ with minimal distance from the curve to $z_j^{-1}$ given by $\frac{\epsilon}{z_j(z_j - \epsilon)}$.
Introduce the Blaschke product
\be\label{Bla}
\Pi(z)= \prod_{z_j \in (q, 0)} |z_j| \frac{z - z_j^{-1}}{z-z_j},
\ee
and note that $\Pi(z^{-1})=\Pi^{-1}(z)$, $\Pi(0)>0$.  
Set $\mathbb D_\epsilon^j=\{z: \, |z-z_j|< \epsilon, z_j\in \si_d\}$ and
\[
A(z)= \begin{cases}
\begin{pmatrix}1& \frac{z-z_j}
{z_j \gamma_j  \E^{2t\Phi(z_j)} }\\ 0 &1\end{pmatrix}, & \quad  z\in \mathbb D_\epsilon^j, \quad z_j\in (q, 0),\\
\begin{pmatrix} 1 & 0 \\
\frac{z_j \gamma_j \E^{2t\Phi(z_j)} }{z-z_j} & 1\end{pmatrix}, & \quad z\in \mathbb D_\epsilon^j, \quad z_j\in (-1, q_1)\cup(0,1).
\end{cases}
\]
Let $m(z)$ be the solution of RHP~\ref{RH1} and define  $m^{\mathrm{ini}}(z)=m^{\mathrm{ini}}(z,n,t)$ as
\be\label{trans7}m^{\mathrm{ini}}(z) = \begin{cases} m(z) A(z)[\Pi(z)]^{-\sigma_3},&\quad  z\in \mathbb D_\epsilon^j, \quad z_j\in\si_d\\
m(z)[\Pi(z)]^{-\sigma_3},&\quad z\in\mathcal Q\setminus\bigcup_{z_j\in\si_d} \overline{\mathbb D_\epsilon^j}, \\
m^{\mathrm{ini}}(z^{-1})\si_1, &\quad z\in \mathcal Q^*.\end{cases}
\ee
This vector function is the unique solution of the following RHP (cf.\ \cite{KTb}):
\begin{RHproblem}[Holomorphic RHP for $\xi_{cr}^\prime\leq \xi \leq \xi_{cr}$]  \label{RH2} 

 Find a holomorphic vector function away from $\Sigma\cup \bigcup_{j=1}^N (\T_j\cup\T_j^*)$ that satisfies
	\begin{itemize}
		\item the jump condition $m_{+}^{\mathrm{ini}}(z)=m_{-}^{\mathrm{ini}}(z) v^{\mathrm{ini}}(z)$, where
\[
	v^{\mathrm{ini}}(z) = \left\{
	\begin{array}{lr}
\begin{pmatrix}	0 & - \frac{\Pi^2(z)\ol{R(z)}}{\E^{2 t \Phi(z)}} \\
	\frac{R(z) \E^{2 t \Phi(z)}}{\Pi^{2}(z)} & 1
	\end{pmatrix}, & z \in \T,\\[3mm]
\begin{pmatrix}
	1 & 0 \\
	\Pi^{-2}(z)\chi(z) \E^{2t\Phi(z)} & 1
	\end{pmatrix}, & z \in [q, q_1],\\[3mm]
\begin{pmatrix}1 & \frac{(z-z_j)\Pi^2(z)}{z_j \gamma_j  \E^{2t\Phi(z_j)} }\\
	0 &1\end{pmatrix},  &  z\in\T_j, \ z_j\in (q,0),	\\[3mm]
	\begin{pmatrix}1 & 0\\
\frac{z_j\gamma_j  \E^{2t\Phi(z_j)} }{(z-z_j)\Pi^2(z)} &1\end{pmatrix}, &  z\in\T_j, \  z_j \in \si_d\setminus(q, 0), \\
\si_1 v^{\mathrm{ini}}(z^{-1})\si_1, & z \in \bigcup_{j=1}^N \T_j^*,\\[1mm]
\si_1 (v^{\mathrm{ini}}(z^{-1}))^{-1}\si_1 & z\in [q_1^{-1}, q^{-1}];
	\end{array}\right.
 \]
\item $m^{\mathrm{ini}}(z^{-1}) = m^{\mathrm{ini}}(z) \si_1$;
\item $ m_1^{\mathrm{ini}}(0) \cdot m_2^{\mathrm{ini}}(0) = 1$,
	 $m_1^{\mathrm{ini}}(0) > 0$.
	\item The resonant/non-resonant condition of RHP~\ref{RH1} holds for $m^{\mathrm{ini}}(z)$ too.
	\end{itemize}
\end{RHproblem}

To summarize, for all values of $\xi\in [\xi_{cr}^\prime, \xi_{cr}]$ we  performed a one-to-one transformation and replaced the meromorphic RHP by the holomorphic RHP,
$$
[m(z,n,t); \mbox{RH problem~\ref{RH1}}] \longmapsto [m^{\mathrm{ini}}(z,n,t); \mbox{RH problem~\ref{RH2}}].
$$
In the next section we list some results established in \cite{emt14}. We represent them in terms of the variable $z$ 
and modify them to take the resonances and the additional discrete spectrum into account.

\section{Reduction to the model RH problem}\label{sec:tomod} 

Let $\xi \in [\xi_{cr}^\prime, \xi_{cr})$. Before we describe the transformations applicable in \eqref{mathcalIres} to obtain the model problem for this region, let us recall the $g$-function mechanism. For shock waves, the $g$-function
 proved its efficiency for several completely integrable equations (cf.\ \cite{KotlyaMinak, egkt}). In our case, the $z$-analog of the $g$-function constructed in \cite{emt14} looks as follows. Set
 \be\label{defQ}Q(z) = \sqrt{\frac{z-y}{z-q}\frac{z-y^{-1}}{z-q^{-1}}}, \qquad z\in\C\setminus \big([q, y]\cup [y^{-1}, q^{-1}]\big),
 \ee
 where $y$ is a point\footnote{The point $\gamma(\xi)$ in the introduction is connected with  $y=y(\xi)$ by $\gamma=\frac{y + y^{-1}}{2}$. In the present paper we use the notation $\la_y$ instead of $\gamma$ (see Remark \ref{rem6}). } which can be computed 
 implicitly from the condition 
 \be\label{defy}\int_{-1}^y P(s)\,Q(s)\frac{ds}{s}=0,
 \ee
with
\be\label{defPe}
P(s):=s+s^{-1} +2\xi +\frac{1}{2}\big(y + y^{-1} - q - q^{-1}\big).
\ee
 As is shown in \cite{emt14},  equation \eqref{defy} has the unique solution $y=y(\xi) \in (q_1, q)$ for any $\xi \in [\xi_{cr}^\prime, \xi_{cr})$. The function $y(\xi)$ is continuous and monotonous, with  $y(\xi^\prime_{cr})=q_1$ and $y(\xi_{cr})=q$.
 Moreover, $y(\xi)$ is differentiable with respect to $\xi \in (\xi_{cr}^\prime, \xi_{cr})$ (cf.\ \cite[App.]{emt14}).
For any $y$, the function \eqref{defQ}  satisfies $Q^2(z^{-1})=Q^2(z)$
which implies evenness,  $Q(z^{-1})=Q(z)$, because $Q(1^{-1})=Q(1)$. With the chosen orientation on $[q, y]\cup [y^{-1}, q^{-1}]$ we denote
 $Q_+(z)=Q(z-\I 0)$. From the evenness of $Q$ outside of $[q, y]\cup [y^{-1}, q^{-1}]$ we obtain oddness of $Q_+$, $Q_+(s)=-Q_+(s^{-1})$ for $s\in[q, y]\cup [y^{-1}, q^{-1}]$. Note that we choose the square root in \eqref{defQ} such that $Q(z)>0$ for $z\in (q, +\infty)$.
Introduce the $g$-function by 
\be\label{g}
g(z)=g(z,\xi) = \frac{1}{2} \int_{1}^z P(s)Q(s)\frac{ds}{s},\quad z\in \C\setminus (-\infty, 1).
\ee

\begin{lemma}\label{propg} The function $g(z)$ satisfies the following properties:
\begin{enumerate} 
[label=\rm{(\alph*)}]
\item $g(z)$ is  single valued on $\C \setminus [q^{-1}, q]$, moreover,
\be\label{oddnessg} g(z^{-1}) = - g(z) \ \ \mbox{for $z\in\C\setminus [q^{-1}, q]$};\ee
\item $\re g(z) = 0$ for $z \in [q,y] \cup [y^{-1}, q^{-1}]\cup \{z: |z|=1\}$; 
\item $g(q)=g(q^{-1})=0$;
\item $g_-(z)=-g_+(z)$ for $z\in [q,y] \cup [y^{-1}, q^{-1}]$; 
\item $\Phi(z)- g(z) =  K(\xi) +O(z)$ as $z \to 0$, where\ $K(\xi)\in \R$;
\item $g_+(z)-g_-(z)= 2\I B$ for $z \in [y, y^{-1}]$, where
\be \label{iB}
B:=-\I\int_q^y P(s)Q_+(s)\frac{ds}{s} \in  \R_+.
\ee
In particular, $g_\pm(y)=g_\pm(y^{-1})=\pm\I B$. 
\end{enumerate}
\end{lemma}

\begin{proof}
{\rm (a)--(c)} Since $Q(z^{-1})=Q(z)$ and $P(z^{-1})= P(z)$ for  $z\in\C\setminus [q^{-1}, q]$, then choosing a contour from $1$ to $z$ which does not have common points with the interval $[q^{-1}, q]$, we obtain \eqref{oddnessg} by the simple change of variables $s\to s^{-1}$.
\noprint{\[
g(z^{-1}) = \frac{1}{2} \int_{1}^{z^{-1}}P(s)Q(s)\frac{ds}{s}=- \frac{1}{2} \int_{1}^{z}P(s^{-1})Q(s^{-1})\frac{ds}{s}
=-g(z).
\]}
Condition \eqref{defy} implies $g_\pm(y)=g_\pm(y^{-1})$. 
Since the integrand $P(s)Q_\pm(s)s^{-1}$ is purely imaginary for $s \in [q,y]\cup [y^{-1}, q^{-1}]$, we have 
\be\label{forma}\re g(q)=\re g_\pm(y)=\re g_\pm(y^{-1})=\re g_\pm(q^{-1}).
\ee 
Oddness of $Q_+(s)$ also implies  $\im g(q)=\im g_\pm(q^{-1})=0$.  Together with \eqref{forma} and the oddness of $g$ this yields 
$g_+(q^{-1})=g_-(q^{-1})=g(q)=0$. Moreover, for $|s|=1$ we have $Q^2(\ol s)=Q^2(s)\in\R_+$, and therefore $\im Q(s)P(s)=0$. 
Since $\frac{ds}{s}\in\I \R$, this implies $\re g(z)=0$ for $|z|=1$, items {\rm (b)--(c)} are thus proved. They imply that  $g(z)$ does not have a jump along 
$(-\infty, q^{-1})$, and this shows {\rm (a)}. Note that these properties improve \cite[Lemma 3.2]{emt14} (see, e.g.\ \cite[Equ.\ (3.25)]{emt14}).
The above considerations imply an additional property,
 \[
 \re g_{\pm}(-1)=0.
 \]
Items {\rm (e)--(f)} are $z$-analogs of \cite[Equ.\ (3.21) and (3.24)]{emt14}, and can be obtained by a simple change of variable $(\la, +)\mapsto z$. The constant $B=B(\xi)$ in {\rm (f)} is the same as \cite[Equ.\ (3.21)]{emt14}.
\end{proof}
\begin{remark} \label{rem6}
The point $y(\xi)$ defines the edge $\la_y =\frac 1 2(y + y^{-1})$ of the  Whitham zone for the Toda shock case. The point $y(\xi)$  coincides with the stationary phase point $z_0(\xi)$ for $\Phi(z,\xi)$ at $\xi=\xi_{cr}$, that is, $z_0(\xi_{cr})=q$.  One can see that $y(\xi_{cr}^\prime)\neq z_0(\xi_{cr}^\prime)$. However, as it was shown  in \cite{emt14}, there are proper $g$-functions in the whole diapason $\xi\in (\xi_{cr}, \xi_{cr, 1})$, and the respective Whitham point $y_1(\xi)$ for $\xi\to \xi_{cr,1}$ will end at  $z_{0,\ell}(\xi_{cr, 1})=1$, where $z_{0,\ell}(\xi)$ is the stationary phase point for the left phase function $\Phi_\ell(z,\xi)$ in \eqref{phil} connected with the left initial scattering data.
\end{remark} 
The signature table for the real part of $g$ is depicted in Fig. \ref{fig:signg}. The points $y$ and $y^{-1}$ are nodal points for the curves $\re g(z)=0$.
\begin{figure}[ht] 
\begin{tikzpicture}

\clip (-3.5,-3.7) rectangle (8,-0.3);

\fill[blue!5] (4.5,-2) circle (2.8cm);
\path [fill=white]  [out=120, in=110] (2.8,-2) to (4.5,-2) [out=240, in=250] (2.8,-2) to (4.5,-2);
\draw[thick] (4.5,-2) circle (2.8cm);
\path [fill=blue!5] [out=60, in=340] (-1,-2) to (-3.5,0.5) -- (-3.5,-4.5) [out=300, in=20] (-1,-2) to (-3.5,-4.5);

\draw[thick] (2.5,-2) -- (3.8,-2);
\draw[thick] (-2.6,-2) -- (-0.5,-2);

\draw[thin] (4.2,-2) circle (0.10cm);
\filldraw (4.2,-2) circle (0.5pt); 
\draw (-3,-2) ellipse (0.12cm and 0.1cm); 
\filldraw (-3,-2) circle (0.5pt) node[above]{\small{${\T_j^*}$}};

\draw[thin] (2.1,-2) circle (0.1cm);
\filldraw (2.1,-2) circle (0.5pt) node[above]{\small{${\T_k}$}};
\draw (1,-2) ellipse (0.12cm and 0.1cm); 
\filldraw (1,-2) circle (0.5pt) node[above]{\small{${\T_k^*}$}};

\draw[gray, thick, densely dotted] (-3.5,0.5) .. controls (-0.5,-1.5) and (-0.5,-2.5) .. (-3.5,-4.5);
\draw[out=90, in=110, gray, thick, densely dotted] (3,-2) to (4.5, -2);
\draw[out=-90, in=250, gray, thick, densely dotted] (3,-2) to (4.5, -2);

\draw[out=60, in=340, blue, thick, densely dashdotted] (-1,-2) to (-3.5,0.5);
\draw[out=300, in=20, blue, thick, densely dashdotted] (-1,-2) to (-3.5,-4.5);
\draw[out=120, in=110, blue, thick, densely dashdotted] (2.8,-2) to (4.5,-2);
\draw[out=240, in=250, blue, thick, densely dashdotted] (2.8,-2) to (4.5,-2);

\node at (-2.6,-1) {$\re g < 0$};
\node at (3.5,-1) {$\re g < 0$};
\node at (0.5,-1) {$\re g > 0$};
\node at (6.7,0.2) {$\T$};
\filldraw (3.8,-2) circle (1pt) node[below]{\small{$q$}};
\filldraw (2.5,-2) circle (1pt) node[below]{\small{$q_1$}};
\filldraw[darkgray] (3,-2) circle (1pt) node[below]{\small{$z_0$}};
\filldraw (4.5,-2) circle (1pt) node[above right]{${0}$};
\filldraw (-0.5,-2) circle (1pt) node[below]{\small{$q_1^{-1}$}};
\filldraw (-2.6,-2) circle (1pt) node[below]{\small{$q^{-1}$}};
\filldraw[blue] (-1,-2) circle (1pt) node[above right, blue]{\small{$y^{-1}$}};
\filldraw[darkgray] (-1.25,-2) circle (1pt) node[below]{\small{$z_0^{-1}$}};
\filldraw[blue] (2.8,-2) circle (1pt) node[above left, blue]{\small{$y$}};

\draw [->, thick] (3.5, -2) -- (3.4,-2);
\draw [->, thick] (-1.8, -2) -- (-1.9,-2);
\draw [->, thick] (7.3, -1.98) -- (7.3,-1.96);
\draw [->, thin] (4.18,-1.9) -- (4.17,-1.9);
\draw [->, thin] (-3.03,-1.9) -- (-3.04,-1.9);
\draw [->, thin] (2.08,-1.9) -- (2.07,-1.9);
\draw [->, thin] (0.98,-1.9) -- (0.97,-1.9);

\end{tikzpicture}
\caption{Signature table of $\re g(z, \xi)$ for $\xi \in (\xi_{cr}^\prime, \xi_{cr})$.} \label{fig:signg}
\end{figure}
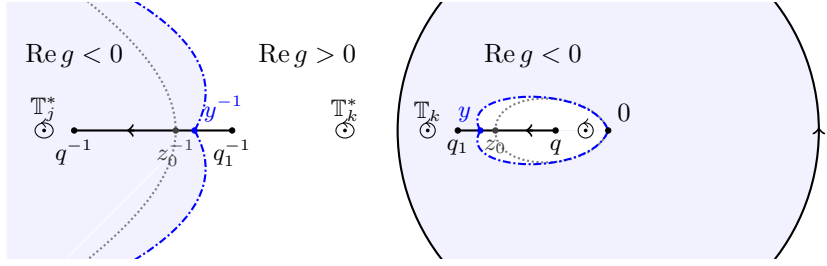
This signature table allows us to choose the radius $\epsilon$ of the circles $\T_j$ so small that
for $ \Phi_j(z):=\Phi(z_j) - \Phi(z) + g(z)$
we will have 
\be \label{est88} \mathrm {sign}\re \Phi_j(z)= \mathrm {sign} \re \Phi(z_j) \ \ \mbox{for all}\ \ z_j\in \si_d\ \ \mbox{and}\ \ z\in \T_j.
\ee
The radius $\epsilon$ should also satisfy
 \[
 8\epsilon< \min\{ \min_{j\neq k} |z_j - z_k|; \min_j|z_j+1|; \min_j |z_j - 1|; \min_j |z_j - q|; \min_j |z_j - q_1|\}.
 \]
 Moreover, since we  intend to justify the asymptotics uniformly in the regions \eqref{mathcalI} or \eqref{mathcalIres}
  for arbitrary small but fixed positive $\varepsilon$, we also assume that
 \begin{align} 
 \label{est91} 4\epsilon &< |y(\xi_{cr} - \varepsilon) - q|, \\
\label{est94}4\epsilon&< |y(\xi_{cr}^\prime + \varepsilon) - q_1|.
 \end{align}
With such a value of $\epsilon$ chosen, we next perform three transformations which lead to a model problem. The transformations are 
analogous to those in \cite{emt14}, modified by additional deformations to wipe out the non-$L^2$ singularities of the jump 
matrix in case of resonances at $q$ or $q_1$. 

{\it Step 1:} On $\T$ one can factorize $v^{\mathrm{ini}}$ using Schur complements
	\[
	v^{\mathrm{ini}}=\begin{pmatrix}
	1 & - \Pi^2\ol{R}\E^{-2t \Phi} \\
	0&1
	\end{pmatrix}
	\begin{pmatrix} 1&0\\
	\Pi^{-2}R\E^{2t\Phi} & 1
	\end{pmatrix}.
	\]
Let  $\xi\in\mathcal I_\varepsilon$, where $\mathcal I_\varepsilon$ is defined by \eqref{mathcalIres}.  Let $y=y(\xi)$ and let $\mathfrak r$, $q_1 <  \mathfrak r < y$, be a point in a small vicinity of $y$ as depicted in Fig.~\ref{fig:step1} with 
\be\label{est99}\frac{\epsilon}{2}\leq|\mathfrak r - y|\leq \epsilon.
\ee  
Introduce a closed contour $\mathcal C_\mathfrak r$ oriented counterclockwise, which starts at $\mathfrak r$ and encloses the interval $[q_1, \mathfrak r]$ passing through the point $q_1-\epsilon$. Denote the domain inside this contour 
  by $\Omega_\mathfrak r$ (with $[q_1, \mathfrak r]$ excluded). 
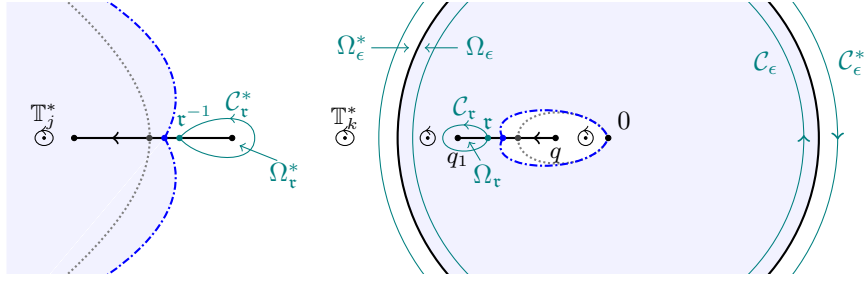
\begin{figure}[ht] 
\begin{tikzpicture}
\clip (-3.5,-3.8) rectangle (8,-0.2);

\fill[blue!5] (4.5,-2) circle (2.8cm);
\path [fill=white]  [out=120, in=110] (3.1,-2) to (4.5,-2) [out=240, in=250] (3.1,-2) to (4.5,-2);
\draw[thick] (4.5,-2) circle (2.8cm);
\path [fill=blue!5] [out=60, in=340] (-1.4,-2) to (-3.5,0.5) -- (-3.5,-4.5) [out=300, in=20] (-1.4,-2) to (-3.5,-4.5);

\draw[thick] (2.5,-2) -- (3.8,-2);
\draw[thick] (-2.6,-2) -- (-0.5,-2);

\draw[thin] (4.2,-2) circle (0.10cm);
\filldraw (4.2,-2) circle (0.5pt);
\draw (-3,-2) ellipse (0.12cm and 0.10cm); 
\filldraw (-3,-2) circle (0.5pt) node[above]{\small{${\T_j^*}$}};

\draw[thin] (2.1,-2) circle (0.10cm);
\filldraw (2.1,-2) circle (0.5pt);
\draw (1,-2) ellipse (0.12cm and 0.10cm); 
\filldraw (1,-2) circle (0.5pt) node[above]{\small{${\T_k^*}$}};

\draw[gray, thick, densely dotted] (-3.5,0.5) .. controls (-0.96,-1.5) and (-0.96,-2.5) .. (-3.5,-4.5);
\draw[out=90, in=110, gray, thick, densely dotted] (3.3,-2) to (4.5, -2);
\draw[out=-90, in=250, gray, thick, densely dotted] (3.3,-2) to (4.5, -2);

\draw[out=60, in=340, blue, thick, densely dashdotted] (-1.4,-2) to (-3.5,0.5);
\draw[out=300, in=20, blue, thick, densely dashdotted] (-1.4,-2) to (-3.5,-4.5);
\draw[out=120, in=110, blue, thick, densely dashdotted] (3.1,-2) to (4.5,-2);
\draw[out=240, in=250, blue, thick, densely dashdotted] (3.1,-2) to (4.5,-2);

\draw[teal] (4.5,-2) circle (2.6cm);
\draw[teal] (4.5,-2) circle (3.05cm);
\draw[out=100, in=90, teal] (2.9,-2) to (2.3,-2);
\draw[out=260, in=270, teal] (2.9,-2) to (2.3,-2);
\draw[out=50, in=90, teal] (-1.2,-2) to (-0.2,-2);
\draw[out=310, in=270, teal] (-1.2,-2) to (-0.2,-2);

\node[teal] at (6.6,-1) {$\mathcal C_\epsilon$};
\node[teal] at (7.75,-1) {$\mathcal C_\epsilon^*$};
\node[teal] at (2.8,-0.8) {$\Omega_\epsilon$};
\draw[->,teal] (2.5,-0.8) -- (2.06, -0.8);
\node[teal] at (1.1,-0.8) {$\Omega_\epsilon^*$};
\draw[->,teal] (1.35, -0.8) -- (1.85, -0.8);
\node[teal] at (2.6,-1.6) {$\mathcal C_{\mathfrak r}$};
\node[teal] at (-0.4,-1.55) {$\mathcal C_{\mathfrak r}^*$};
\node[teal] at (2.9,-2.5) {$\Omega_{\mathfrak r}$};
\draw[->,teal] (2.8,-2.3) -- (2.65, -2.05);
\node[teal] at (0.2,-2.5) {$\Omega_{\mathfrak r}^*$};
\draw[->,teal] (0.1,-2.3) -- (-0.4, -2.1);

\filldraw (3.8,-2) circle (1pt) node[below]{\small{$q$}};
\filldraw (2.5,-2) circle (1pt); 
\node at (2.5,-2.3) {\small{$q_1$}};
\filldraw[darkgray] (3.3,-2) circle (1pt); 
\filldraw (4.5,-2) circle (1pt) node[above right]{${0}$};
\filldraw (-0.5,-2) circle (1pt);
\filldraw (-2.6,-2) circle (1pt); 
\filldraw[blue] (-1.4,-2) circle (1pt); 
\filldraw[darkgray] (-1.6,-2) circle (1pt); 
\filldraw[blue] (3.1,-2) circle (1pt); 
\filldraw[teal] (2.9,-2) circle (1pt) node[above, teal]{\small{$\mathfrak r$}};
\filldraw[teal] (-1.2,-2) circle (1pt); 
\node[teal] at (-1,-1.7) {\small{$\mathfrak r^{-1}$}};

\draw [->, thick] (3.6, -2) -- (3.5,-2);
\draw [->, thick] (-2, -2) -- (-2.1,-2);
\draw [->, thick, teal] (7.1, -1.98) -- (7.1,-1.96);
\draw [->, thick, teal] (7.55, -2) -- (7.55,-2.02);
\draw [->, thin] (4.18,-1.9) -- (4.17,-1.9);
\draw [->, thin] (-3.03,-1.9) -- (-3.04,-1.9);
\draw [->, thin] (2.08,-1.9) -- (2.07,-1.9);
\draw [->, thin] (0.98,-1.9) -- (0.97,-1.9);
\draw [->, thin, teal] (-0.52,-1.74) -- (-0.53,-1.74);
\draw [->, thin, teal] (2.59,-1.825) -- (2.58,-1.825);
\end{tikzpicture}

\caption{Contour deformation of Step 1.} \label{fig:step1}
\end{figure} 
 Let $\Omega_\epsilon$  be an open annulus between the circles $\T$ and $\mathcal C_\epsilon=\{z: |z|=1-\epsilon\}$
oriented counterclockwise, with $\Omega_{\mathfrak r}^*$ and $\Omega_\epsilon^*$ the images of these domains under the map 
$z\mapsto z^{-1}$. According to \eqref{decay} 	the reflection coefficient $R(z)$ can
be continued as a meromorphic function in the domain $\{z:1>|z|>\E^{-\rho}\}$,  which covers the interval $[q_1, y(\xi_{cr}-\varepsilon)]$
by \eqref{nu}. Thus $R(z)$ is a holomorphic function in $\Omega_\epsilon\cup\Omega_{\mathfrak r}$, because  these domains do not contain  points of the discrete
spectrum by our choice of $\epsilon$. 
We extend $R(z)$ to $\Omega_\epsilon^*\cup\Omega_{\mathfrak r}^*$
 by $R(z)= \ol{R(z^{-1})}$. 
Redefine $m^{\mathrm{ini}}$ by
\be\label{step1}
m^{(1)}(z) = m^{\mathrm{ini}}(z)\begin{cases}
\begin{pmatrix} 1&0\\
-\Pi^{-2}(z) R(z) \E^{2t\Phi(z)} & 1
\end{pmatrix}, & z \in \Omega_{\mathfrak r}\cup 
\Omega_\epsilon,\\
\begin{pmatrix} 1& -\Pi^{-2}(z) R(z^{-1}) \E^{-2t\Phi(z)}\\
0 & 1
\end{pmatrix}, & z \in \Omega_{\mathfrak r}^*\cup 
\Omega_\epsilon^*,\\
\id, & \text{else},
\end{cases}
\ee
and orient $\mathcal C_{\mathfrak r}$ and $\mathcal C_{\mathfrak r}^{*}$ counterclockwise.
Then the jump along $\T$ disappears as well as the jump along $[\mathfrak r, q_1]$, since the Pl\"ucker identity
implies that $R_-(z) - R_+(z) + \chi(z)=0$ for $z \in [\mathfrak r, q_1]$ (compare
\cite[Lemma~3.2]{egkt}). Moreover, since the continuation of $R(z)$ is in agreement with the scattering relation \eqref{pst} and \eqref{defm} is the unique solution of RHP~\ref{RH1}, it is straightforward to check that $m^{(1)}(z)$ given by \eqref{trans7}, \eqref{step1} does not have singularities at $q_1$ and $q_1^{-1}$ both in  the resonant and non-resonant case. In summary, $m^{(1)}(z)$ satisfies
\begin{RHproblem}\label{3}
\footnote{From here on we (mostly) state RHPs only in terms of their conditions.}
\noindent	\begin{itemize}
		\item $m_{+}^{(1)}(z,n,t)=m_{-}^{(1)}(z,n,t) v^{(1)}(z,n,t)$, where
\[
	v^{(1)}(z) = \left\{
	\begin{array}{ll}
\begin{pmatrix}
	1 & 0 \\
	\Pi^{-2}(z)\chi(z) \E^{2t\Phi(z)} & 1
	\end{pmatrix}, & \quad z \in [q, \mathfrak r],\\[3mm]
\begin{pmatrix} 1&0\\
\Pi^{-2}(z) R(z) \E^{2t\Phi(z)} & 1
\end{pmatrix}, & \quad z \in \mathcal C_\epsilon\cup\mathcal C_{\mathfrak r},\\[3mm]
\si_1 v^{(1)}(z^{-1})\si_1, & \quad z \in  \mathcal C_{\mathfrak r}^*\cup\mathcal C_\epsilon^*,\\[2mm]
\si_1 (v^{(1)}(z^{-1}))^{-1}\si_1, & \quad z \in [\mathfrak r^{-1}, q^{-1}],\\[2mm]
	v^{\mathrm{ini}}(z), & \quad z \in \bigcup_j(\T_j\cup\T_j^*);
	\end{array}\right.
 \]
\item $m^{(1)}(z^{-1}) = m^{(1)}(z) \si_1$;
\item $m_1^{(1)}(0) \cdot m_2^{(1)}(0) = 1$, $m_1^{(1)}(0) > 0$;
\item The resonant/non-resonant condition of RHP~\ref{RH1} holds for $m^{(1)}(z)$ only at $q$, $q^{-1}$.

\end{itemize}
\end{RHproblem}

{\it Step 2:} The jump matrix on  $[q, \mathfrak r]\cap \{z: \re\Phi(z)>0\}$ contains 
off-diagonal elements which are exponentially increasing in time. One can get rid of this exponential growth by 
replacing the phase function with the $g$-function, which is purely imaginary on $[q,y]$ and has 
negative real part on $[y,\mathfrak r]$. For  $z \in \C$ set
\[
m^{(2)}(z)=m^{(1)}(z)\E^{-t(\Phi(z) - g(z))\sigma_3}= m^{(1)}(z)\begin{pmatrix} \E^{-t(\Phi(z) - g(z))}& 0\\ 0&\E^{t(\Phi(z) - g(z))}\end{pmatrix}.
\] 
Then Lemma~\ref{propg} and \eqref{iB} imply that $m^{(2)}(z)$ is the unique holomorphic solution  of the following problem: 
\begin{RHproblem}\label{4}
\noindent		\begin{itemize}
		\item $m_{+}^{(2)}(z,n,t)=m_{-}^{(2)}(z,n,t) v^{(2)}(z,n,t)$, where
\[
	v^{(2)}(z) = \left\{
	\begin{array}{ll}
\begin{pmatrix}
	\E^{t(g_+-g_-)} & 0 \\
	\Pi^{-2}\chi & \E^{-t(g_+-g_-)}
	\end{pmatrix}, & \quad z \in [q,y],\\[3mm]
	\begin{pmatrix}
		\E^{2\I tB} & 0 \\
		\Pi^{-2}\chi \E^{2t \re g} & \E^{-2\I t B}
		\end{pmatrix}, & \quad z \in [y,\mathfrak r],\\[3mm]
		\begin{pmatrix}
			\E^{2\I tB} & 0 \\
			0 & \E^{-2\I t B}
			\end{pmatrix}, & \quad z \in [\mathfrak r,\mathfrak r^{-1}],\\[3mm]
\si_1 (v^{(2)}(z^{-1}))^{-1}\si_1, & \quad z\in [\mathfrak r^{-1}, q^{-1}],\\[2mm]
\E^{-t(\Phi - g)\sigma_3}v^{(1)}\E^{t(\Phi - g)\sigma_3}, &\quad z\in  \Gamma,
	\end{array}\right.
 \]
and 
\be\label{defK}
\Gamma:=\mathcal C_{\mathfrak r}\cup \mathcal C_{\mathfrak r}^*\cup \mathcal C_{\epsilon}\cup \mathcal C_{\epsilon}^*\cup\bigcup_{j=1}^N\left(\T_j\cup\T_j^*\right);
\ee
\item $m^{(2)}(z^{-1}) = m^{(2)}(z) \si_1$;
\item $m_1^{(2)}(0) \cdot m_2^{(2)}(0) = 1$,  $m_1^{(2)}(0) > 0$.
\item The vector $m^{(2)}(z)$ does not have  singularities at $q_1$ and $q_1^{-1}$.
It has bounded values at $\mathfrak r$ and $\mathfrak r^{-1}$. Its behavior at $q$ and $q^{-1}$ is the same as for $m(z)$ in \eqref{rezon}.
	\end{itemize}
	\end{RHproblem}
\begin{remark} Our choice of $y$, $\mathfrak r$ and $\epsilon$ in \eqref{est88}--\eqref{est91}, \eqref{est99}  guarantees that 
\be\label{estforv2}
v^{(2)}(z)=\id +O(\E^{-c(\varepsilon)t}),\ \ \mbox{for}\ \ z\in\Gamma,\ \ \mbox{as}\ \ t\to\infty,
\ee
uniformly for $\xi\in \mathcal I_\varepsilon$.
\end{remark}

{\it Step 3:} The last step involves the lense mechanism to remove the oscillating terms (with respect to $t$) 
in the jump matrix  on $[q, y]\cup [y^{-1}, q^{-1}]$. 
To this end introduce the function 
\[
\Omega(z,s)=\frac {1}{2 s}\frac{s+z}{s-z}, 
\]
which can be considered as the Cauchy kernel for symmetric contours, because $\Omega(z,s)=\frac{1}{z-s}(1 + o(1))$ as $z\to s$, and  \[\Omega(z,s^{-1})d(s^{-1})=\Omega(z^{-1},s) ds.\] This property implies that for any "good" function $f(s)$ such that $f(s^{-1})=f(s)$ and 
\be\label{intprop}\int_{-1}^q f(s)\frac{ds}{s}=\int_{-1}^q f(s)\Omega(0,s)ds=0,\ee the function
  \be \label{p(z)}
	p(z) =\frac{1}{2\pi \I} \int_{q}^{q^{-1}} \Omega(z,s)f(s)ds, \ee
solves the jump problem
\begin{align}
\label{pi1}	p_+(z) &= p_-(z) + f(z), \quad z \in [q,q^{-1}],\\
\label{pi2}	p(z^{-1}) &= - p(z), \quad z \in \C\setminus [q^{-1}, q]\\
\label{pi3}	p(z)& = O(z), \quad z\to 0.
 \end{align}
By "good" function we mean a function $f\in C^1((q^{-1}, y^{-1})\cup(y^{-1}, y)\cup (y,q))$ which has the following behavior in the node points: \[f(s)=\frac{C(\kappa)}{\sqrt{s-\kappa}} (1 +o(1)), \quad s\to \kappa\in \{q^{-1}, y^{-1}, y, q\}, \quad C(\kappa)\neq 0.\]
Then (cf.\ \cite{Musch}) 
\[p(z)= \frac{C_1(\kappa)}{\sqrt{z-\kappa}}(1 +o(1)), \quad z\to \kappa\in \{q^{-1}, y^{-1}, y, q\}, \quad C_1(\kappa)\neq 0.\]
Set   
\be  \label{w(z)}
	\mathcal S(z) = \sqrt{\frac{(z-q)(z-y)(z-y^{-1})(z-q^{-1})}{z^2}}, \quad z \in  \C \setminus ([q,y] \cup [y^{-1}, q^{-1}]).\ee 
	This function satisfies the following symmetries:
	$\mathcal S(z^{-1})=\mathcal S(z)$ for $z \notin [q,y] \cup [y^{-1}, q^{-1}]$ and $\mathcal S_-(z) = \mathcal S_+(z^{-1})=-\mathcal S_+(z)$ for $z \in [q,y] \cup [y^{-1}, q^{-1}]$. 
	Define \be\label{mathcalF} \mathcal F(z):=\E^{\mathcal S(z) p(z)},\quad z\in \C\setminus[q,q^{-1}].\ee
\begin{lemma} The function $\mathcal F(z)$ is holomorphic in $\C\setminus[q,q^{-1}]$ and satisfies the property $\mathcal F(z^{-1})=\mathcal F^{-1}(z)$. It has bounded limits on the sides of the contour $[q, q^{-1}]$ and solves the jump problem \[\aligned\mathcal F_+(z)\mathcal F_-(z)&=\E^{f(z)\mathcal S_+(z)}, \quad z\in [q,y]\cup [y^{-1}, q^{-1}]; \\ \mathcal F_+(z)&=\mathcal F_-(z)\E^{f(z)\mathcal S(z)},\quad z\in [y,y^{-1}].\endaligned\]
\end{lemma}
\begin{proof} The proof is immediate from \eqref{pi1}--\eqref{w(z)}, the properties of $\mathcal S$ and the 
Sokhotski-Plemelj theorem.
\end{proof}
Let now $f(z)$ be defined as
	\be\label{defef1}
	f(s) : = \begin{cases}
		\frac{ \log(\Pi^{-2}(s)|\chi(s)\mathcal V_+^2(s)|)}{\mathcal S_+(s)}, & s \in [q,y], \\
		\frac{\I (\Delta-\frac{\pi\ell_{\mathcal V}}{2})}{\mathcal S(s)}, & s \in [y, -1],\\
		f(s^{-1}), & s \in [-1, q^{-1}], \\
		\end{cases} 
\ee
where
\be\label{Delta}
\Delta= -  \I \int_{q}^{y} \frac{\log(\Pi^{-2}(s)|\chi(s)\,\mathcal V_+^2(z)|)}{\mathcal S_+(s)}\frac{ds}{s} 
\left(\int_{y}^{-1} \frac{ds}{s\, \mathcal S(s)} \right)^{-1} +\frac{\pi\ell_{\mathcal V}}{2}\in\R,
\ee
and 
\be\label{defVell}
\begin{cases}\mathcal V(z):=\left(\frac{z-q^{-1}}{z-q}\right)^{1/4},\  \ell_{\mathcal V}=1, & \mbox{if}\ \chi(z)\ \mbox{satisfies \eqref{chinon} at}\ q;\\
\mathcal V(z):=\left(\frac{z-q}{z-q^{-1}}\right)^{1/4},\ \ell_{\mathcal V}=-1, & \mbox{if}\ \chi(z)\ \mbox{satisfies \eqref{chirez} at}\ q;
\end{cases}\ee
with $\mathcal V(0)>0$. We observe that $\mathcal V(z^{-1})=\mathcal V^{-1}(z)$ for $z\notin [q,y] \cup [y^{-1}, q^{-1}]$ and   \[\mathcal V_+(z)=\mathcal V_-(z)\E^{\I\frac{\pi\ell_{\mathcal V}}{2}},\qquad \mathcal V_+(z)\mathcal V_-(z)=|\mathcal V_+(z)|^2, \quad z\in [q, q^{-1}].\] 
It is straightforward to see that $f(z)$ satisfies \eqref{intprop}, and is a "good" function. In addition, $f(s)\in\I\R$, therefore if $p(z)$ is its Cauchy type integral \eqref{p(z)}, then
\[\lim_{z\to 0} p(z)\mathcal F(z)>0.\] The above considerations imply for $F(z)$ defined by
\be\label{F(z)}
F(z) = \mathcal F(z)\mathcal V^{-1}(z), \quad z \in \C \setminus [q,q^{-1}],
\ee
the following properties.
\begin{lemma}\label{lemma:F} The function $F(z)$ satisfies
\begin{enumerate} 
[label=\rm{(\alph*)}]
\item $F_+(z) F_-(z) = \Pi^{-2}(z)|\chi(z)|$ for $z \in [q,y]$;
\item $F_+(z) F_-(z) = \Pi^{-2}(z) |\chi(z)|^{-1}$ for $z \in [y^{-1}, q^{-1}]$;	
\item $F_+(z) = F_-(z) e^{\I \Delta}$ for $z \in [y, y^{-1}]$;
\item $F(z^{-1})=F^{-1}(z)$ for $z \in \C \setminus [q,q^{-1}]$;
\item $F(0)> 0$;
\item If $\chi(z)$ satisfies \eqref{chinon} at $q^{\pm 1}$ then $F(z)=C(z-q^{\pm 1})^{\pm 1/4}(1 + o(1))$ as $z\to q^{\pm 1}$.
\item If $\chi(z)$ satisfies \eqref{chirez} at $q^{\pm 1}$ then $F(z)=C(z-q^{\pm 1})^{\mp 1/4}(1 + o(1))$ as $z\to q^{\pm 1}$.
\end{enumerate}
\end{lemma}
Set
\[
G^F(z) = \begin{pmatrix}
	F^{-1}(z) & -\frac{\Pi^2 (z) F(z)}{X(z)} e^{-2tg(z)} \\
	0 & F(z)
	\end{pmatrix},
\]
where the function $X(z)$ is well defined by \eqref{Chi} in a vicinity of $[q,y]$ and satisfies the property $X_\pm(z) = \pm\I |\chi(z)|$  for $z \in [q,y]$. Recalling that  $g_+(z) = -g_-(z)$ for $z\in [q,y]$, we observe that $v^{(2)}(z)$ can be factorized by
\[
v^{(2)}(z) = 
	G_-^F(z)
\begin{pmatrix}
	0 &\I \\
	\I &0
	\end{pmatrix}G_+^F(z)^{-1}, \quad z \in [q,y].
\]
In the domain of existence of $X(z)$ introduce a subdomain  $\Omega$
 as depicted in Fig.~\ref{fig:Omega} with $\Omega^*=\{z: z^{-1}\in\Omega\}$. 
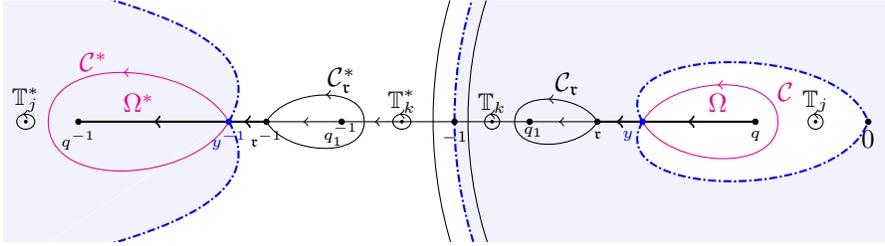
\begin{figure}[ht] 
\begin{tikzpicture}
\clip (-5.2,-1.6) rectangle (7,1.6);
\path [fill=blue!5] [out=60, in=345] (-2,0) to (-5,2) -- (-5,-2) [out=300, in=15] (-2,-0) to (-5,-2);
\path [fill=blue!5] (1.4,2.1) arc (158:202:5.5) -- (6.8, -2.02) -- (6.8, 2.1) -- (1.4, 2.1);
\path [fill=white] [out=120, in=110] (3.5,0) to (6.5,0) [out=240, in=250] (3.5,0) to (6.5,0);
\draw[out=60, in=345, blue, thick, densely dashdotted] (-2,0) to (-5,2);
\draw[out=300, in=15, blue, thick, densely dashdotted] (-2,0) to (-5,-2);
\draw[out=120, in=110, blue, thick, densely dashdotted] (3.5,0) to (6.5,0);
\draw[out=240, in=250, blue, thick, densely dashdotted] (3.5,0) to (6.5,0);
\draw[blue, thick, densely dashdotted] (1.4,2.1) arc (158:202:5.5);
\draw (1.6,2.1) arc (157:202.7:5.3);
\draw (1.1,2.1) arc (159:200.6:5.8);

\draw[out=120, in=90, magenta] (-2,0) to (-4.4,0);
\draw[out=240, in=270, magenta] (-2,0) to (-4.4,0);
\draw[out=60, in=90, magenta] (3.5,0) to (5.3,0);
\draw[out=300, in=270, magenta] (3.5,0) to (5.3,0);

\draw[thick] (2.9,0) -- (5,0);
\draw[thick] (-4,0) -- (-1.5,0);
\draw (-1.5,0) -- (2.9,0);

\draw[out=120, in=90] (2.9,0) to (1.8,0);
\draw[out=240, in=270] (2.9,0) to (1.8,0);
\draw[out=60, in=90] (-1.5,0) to (-0.2,0);
\draw[out=300, in=270] (-1.5,0) to (-0.2,0);

\draw[thin] (1.5,0) circle (0.10cm);
\filldraw (1.5,0) circle (0.5pt) node[above]{\small{${\T_k}$}};
\draw (0.3,0) ellipse (0.12cm and 0.10cm); 
\filldraw (0.3,0) circle (0.5pt) node[above]{\small{${\T_k^*}$}};
\draw[thin] (5.8,0) circle (0.10cm);
\filldraw (5.8,0) circle (0.5pt) node[above]{\small{${\T_j}$}};
\draw (-4.7,0) ellipse (0.12cm and 0.10cm); 
\filldraw (-4.7,0) circle (0.5pt) node[above]{\small{${\T_j^*}$}};

\filldraw (-4,0) circle (1pt) node[below]{\tiny{$q^{-1}$}};
\filldraw[blue] (-2,0) circle (1pt) node[below]{\tiny{$y^{-1}$}};
\filldraw (-1.5,0) circle (1pt) node[below]{\tiny{$\mathfrak r^{-1}$}};
\filldraw (-0.5,0) circle (1pt);
\node at (-0.51, -0.17) {\tiny{$q_1^{-1}$}};
\filldraw (1,0) circle (1pt) node[below]{\tiny{$-1$}};
\filldraw (2,0) circle (1pt); 
\node at (2.05,-0.14) {\tiny{$q_1$}};
\filldraw (2.9,0) circle (1pt) node[below]{\tiny{$\mathfrak r$}};
\filldraw[blue] (3.5,0) circle (1pt) node[below left]{\tiny{$y$}};
\filldraw (5,0) circle (1pt) node[below]{\tiny{$q$}};
\filldraw (6.5,0) circle (1pt) node[below]{${0}$};
\node[magenta] at (4.5,0.25) {$\Omega$};
\node[magenta] at (5.4,0.4) {$\mathcal C$};
\node[magenta] at (-3.2,0.25) {$\Omega^*$};
\node[magenta] at (-3.8,0.8) {$\mathcal C^*$};
\node at (2.5,0.5) {$\mathcal C_{\mathfrak r}$};
\node at (-0.5,0.55) {$\mathcal C_{\mathfrak r}^*$};

\draw [->, thick] (3.25, 0) -- (3.15,0);
\draw [->, thick] (4.2, 0) -- (4.1,0);
\draw [->, thick] (-1.7, 0) -- (-1.8,0);
\draw [->, thick] (-2.9, 0) -- (-3,0);
\draw [->, magenta] (4.51, 0.49) -- (4.5,0.49);
\draw [->, magenta] (-3.39, 0.655) -- (-3.4,0.655);
\draw [->, thin] (1.48,0.1) -- (1.47,0.1);
\draw [->, thin] (0.28,0.1) -- (0.27,0.1);
\draw [->, thin] (5.78,0.1) -- (5.77,0.1);
\draw [->, thin] (-4.73,0.1) -- (-4.74,0.1);
\draw [->, thin] (2.26,0.3) -- (2.25,0.3);
\draw [->, thin] (-0.7,0.36) -- (-0.71,0.36);
\draw [->] (-1, 0) -- (-1.01,0);
\draw [->] (-0.05, 0) -- (-0.06,0);
\draw [->] (2.4, 0) -- (2.39,0);

\end{tikzpicture}
\caption{Contour deformation of Step 3.} \label{fig:Omega}
\end{figure}
These domains and their boundaries $\mathcal C$ and $\mathcal C^*$ should not contain or 
intersect $\mathbb T_j$ and $\mathbb T_j^*$ and should be situated inside the regions $\re g>0$ and $\re g<0$, respectively. We add $\mathcal C$ and $\mathcal C^*$ (both oriented counterclockwise) to the contour $\Gamma$ and denote
	\be\label{defXi}\Xi:=\mathcal C\cup\mathcal C^*\cup\mathcal C_{\mathfrak r}\cup \mathcal C_{\mathfrak r}^*\cup \mathcal C_{\epsilon}\cup \mathcal C_{\epsilon}^*\cup\bigcup_{j=1}^N\left(\T_j\cup\T_j^*\right).\ee
Define $m^{(3)}(z)$ by
\be\label{m3}
m^{(3)}(z)=\begin{cases}
	 m^{(2)}(z)G^F(z), & \quad z \in \Omega,\\
	m^{(3)}(z^{-1})\si_1, & \quad z \in \Omega^*,\\
	m^{(2)}(z)(F(z))^{-\si_3}, &  \quad z \in \C \setminus (\overline{\Omega} \cup \ol{\Omega^*}).
	\end{cases}
	\ee
\begin{theorem} \label{thm:v3}
For $\xi \in \mathcal I_\varepsilon$, RH problem~\ref{RH2} is equivalent to the following RH problem: to find a vector function
holomorphic in $\C\setminus(\Xi\cup [q, q^{-1}])$ which has continuous limits 
on the sides of the contour \ $\Xi\cup[q, q^{-1}]$ except for the points $q, q^{-1}$  and satisfies
\begin{itemize}
		\item the jump condition $m_{+}^{(3)}(z,n,t)=m_{-}^{(3)}(z,n,t) v^{(3)}(z,n,t)$, where
\be\label{vi30} 
v^{(3)}(z)=
\begin{cases}
	\I \si_1, & \quad z \in [q,y],\\
	\begin{pmatrix}
		\E^{2 \I tB- \I \Delta} & 0 \\
		\frac{\chi(z)\E^{2t (g_+(z) + g_-(z))}}{\Pi^2(z)F_+(z)F_-(z)}  & \E^{-2 \I tB + \I \Delta}
		\end{pmatrix}, & \quad z \in [y, \mathfrak r],\\
		\begin{pmatrix}
		\E^{2 \I tB- \I\Delta} & 0 \\
		0 & \E^{-2\I t B + \I\Delta}
		\end{pmatrix}, & \quad z \in [\mathfrak r, -1],
\end{cases}\ee
\be\label{vi3}v^{(3)}(z)=
\begin{cases}		
					\begin{pmatrix}
			1 & \frac{\Pi^2(z)F^2(z)}{X(z)} e^{-2tg(z)} \\
			0 & 1
			\end{pmatrix}, & \quad z \in \mathcal C,\\[3mm]
\si_1 v^{(3)}(z^{-1})\si_1, &\quad z\in \mathcal C^*,\\[3mm]
						\si_1 (v^{(3)}(z^{-1}))^{-1}\si_1, &  \quad z\in [-1, q^{-1}],\\[3mm]
	[F(z)]^{-\sigma_3}v^{(2)}(z)[F(z)]^{\sigma_3}, & \quad z\in \mathcal \Gamma;						
	\end{cases}
\ee
\item the symmetry condition $m^{(3)}(z^{-1}) = m^{(3)}(z) \si_1$;
\item the normalization condition $m_1^{(3)}(0) \cdot m_2^{(3)}(0) = 1$, $m_1^{(3)}(0) > 0$.
\item At the points  $\{q, q^{-1}, y, y^{-1}, \mathfrak r, \mathfrak r^{-1}\}$ of discontinuity of the jump matrix, 
  $m^{(3)}(z)$ has the following behavior: it has at most a fourth root singularity
 \begin{align}\label{sing4}m^{(3)}(z)&=O(z-\kappa)^{-1/4},\quad \mbox{as}\ z\to \kappa\in \{q,q^{-1}\}, \ \mbox{and} \\ \nn 
 m^{(3)}(z)&=O(1), \quad \mbox{as}\ z\to \kappa\in \{\mathfrak r, \mathfrak r^{-1}, y, y^{-1}\}.
 \end{align}
\end{itemize}
Here $B$ and $\Delta$ are defined by \eqref{iB} and \eqref{Delta}, $g(z)$ by \eqref{g}, and $F(z)$  by \eqref{w(z)}, \eqref{p(z)}, \eqref{F(z)}. 
 For large $z$, $m^{(3)}(z)$ and the solution $m(z)$ of the initial RH problem~\ref{RH1} are connected via
\be\label{connection6}m^{(3)}(z) = m(z) \Big[\Pi(z)F(z) \E^{t(\Phi(z) - g(z))}\Big]^{-\si_3}.
\ee
\end{theorem}
\begin{proof}The jump condition is immediate from Lemmas~\ref{lemma:F},~\ref{propg}, which imply
\be\label{imp33}
\frac{F_-(z)}{F_+(z)} \E^{t(g_+(z) - g_-(z))} =  \E^{2 \I t B - \I \Delta}, \quad z \in [y, y^{-1}].
\ee
The  claim to be discussed in more detail is \eqref{sing4}. Transformation \eqref{m3} implies that in a 
vicinity of $q$,
\[
m^{(3)}(z)=\begin{pmatrix}
F^{-1}(z) m_1^{(2)}(z), & - \frac{\Pi^2(z)F(z)}{X(z)}m_1^{(2)}(z)\E^{-2tg(z)} + F(z) m_2^{(2)}(z)\end{pmatrix}. 
\] 
In the non-resonant case \eqref{chinon} we have three possibilities for $m$, and therefore for $m^{(2)}$, which 
include possible zeros of the Jost solutions, 
\begin{enumerate} 
\item $m^{(2)}( q)= (C_1, C_2)$
\item $m^{(2)}(z)= (C_1(z- q)^{1/2}, C_2)(1 +o(1))$ 
\item $m^{(2)}(z)=(C_1, C_2 (z - q))(1 +o(1))$, where $ C_1C_2\neq 0$. 
\end{enumerate}
The symmetry condition implies the respective behavior at $q^{-1}$.  In the resonant case \eqref{chirez} we have \eqref{rezon}. 
By use of (f) and (g) of Lemma \ref{lemma:F} we obtain \eqref{sing4}.
\end{proof}

In summary, we have transformed the initial RH problem $[m^{\mathrm{ini}}(z,n,t); \mbox{RHP~\ref{RH2}}]$ by {\it Steps 1-3}
to an equivalent RH problem $[m^{(3)}(z,n,t); \mbox {Theorem \ref{thm:v3}}]$
with jump matrix $v^{(3)}$ of the form $v^{(3)}= v^{mod} + v^{err}$, where
\begin{equation}\label{vmodDef}
	v^{mod}=
	\begin{cases}
		\I \si_1, & \quad z \in [q,y],\\
		-\I \si_1, & \quad z \in [y^{-1}, q^{-1}],\\
		\begin{pmatrix}
			\E^{2\I tB - \I \Delta} & 0 \\
			0 & \E^{-2\I t B + \I \Delta}
			\end{pmatrix}, & \quad z \in [y, y^{-1}],\\
			\id, & \quad z\in\Xi.
			\end{cases}
	\end{equation}	
 The matrix $v^{mod}$ on 
$[q,q^{-1}]$ is the jump matrix of an explicitly solvable RHP and its solution will yield the principal term of the
long-time asymptotic expansion of the solution for the initial value problem \eqref{tl}--\eqref{main}, \eqref{decay}. We will solve this model RHP in the next section.
Note that the jump matrix $v^{(3)}$ on the contour $\Xi\cup [\mathfrak r, \mathfrak r^{-1}]$ is exponentially close to the identity matrix  as $t \rightarrow \infty$
except for small neighborhoods of the critical (parametrix) points $y, y^{-1}$. To estimate the error term one has to rescale the equivalent RHP in neighborhoods of the parametrix points and solve the respective local problems, which can be analyzed and controlled individually. This will be done in Section \ref{s:par}.

\section{Solution of the vector model RH problem}\label{sec:4}

We have to solve the following jump problem

\begin{modelRHP} \label{mRHP}
Find a holomorphic vector function in $\C \setminus [q^{-1},q]$ satisfying
	\begin{itemize}
	\item the jump condition $m^{mod}_+(z)=m^{mod}_-(z)v^{mod}(z)$ with $v^{mod}(z)$ given by \eqref{vmodDef};
	\item $m^{mod}(z^{-1}) = m^{mod}(z) \si_1$;
	\item $ m^{mod}_1(0) \cdot m^{mod}_2 (0) = 1$,  $m^{mod}_1(0) > 0$;
	\item The vector $m^{mod}(z)$ has continuous limits as $z$ approaches the jump contour except for $q$ and $q^{-1}$ and the points of discontinuity of the jump matrix, $y$, $y^{-1}$, where the forth-root singularities are admissible. 
\end{itemize}
\end{modelRHP}
Uniqueness of the solution to this problem is proved in \cite{emt14}.

Consider the two-sheeted Riemann surface $\mathbb X$ associated with the function
\[
\mathcal R(z) = \sqrt{(z-q)(z-y)(z-y^{-1})(z-q^{-1})}
\]
such that $\mathcal R(1) \in \R_+$ and $\mathcal R(-1) \in \R_-$. The sheets of $\mathbb X$ are glued along the cuts $[q^{-1},y^{-1}]$ and $[y,q]$.
Points on $\mathbb X$ are denoted by $(z, \pm)$.
We first choose a canonical homology basis of cycles $\{\mathfrak a, \mathfrak b\}$ on $\mathbb X$, see Fig.~\ref{fig:Hom}.
The $\mathfrak b$ cycle surrounds the interval $[y, q]$ counterclockwise on the upper sheet and the $\mathfrak a$ cycle passes
from $y$ to $y^{-1}$ on the upper sheet and back from $y^{-1}$ to $y$ on the lower sheet.

\begin{figure}[ht] 
\begin{tikzpicture}[scale=0.9]

\draw[thick] (2,0) -- (4,0);
\draw[thick] (-3.5,0) -- (-0.5,0);

\filldraw (-3.5,0) circle (1pt) node[below]{$q^{-1}$};
\filldraw(-0.5,0) circle (1pt) node[below right]{$y^{-1}$};
\filldraw(2,0) circle (1pt) node[below]{$y$};
\filldraw(4,0) circle (1pt) node[below]{$q$};
\node at (1, -0.7) {$\mathfrak a$};
\node at (3.3, -0.55) {$\mathfrak b$};

\draw[out=90, in=90, dotted] (-0.7,0) to (2.3,0);
\draw[out=270, in=270] (-0.7,0) to (2.3,0);
\draw[out=90, in=90] (1.7,0) to (4.3,0);
\draw[out=270, in=270] (1.7,0) to (4.3,0);

\draw [->] (3.01, 0.76) -- (3,0.76);
\draw [->] (3, -0.76) -- (3.1,-0.76);
\draw [->] (0.75, -0.88) -- (0.74,-0.88);
\draw [->, black!50] (0.75, 0.88) -- (0.76,0.88);

\end{tikzpicture}
\caption{Homology basis on $\mathbb X$. Solid curves lie on upper sheet, dotted curve lies on lower sheet.} \label{fig:Hom}
\end{figure}
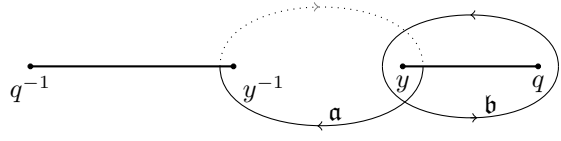

Consider the normalized holomorphic Abel differential
\[ 
\zeta = \frac{dz}{\Gamma \mathcal R(z)}, \qquad  \Gamma = \int_{\mathfrak a} \frac{dz}{\mathcal R(z)} = 2 \int_y^{y^{-1}} \frac{dz}{\mathcal R(z)} > 0,
\]
then  $\int_{\mathfrak a} \zeta = 1$ and $\tau = \tau (\xi) = \int_{\mathfrak b}\zeta \in \I \R_+$.
From here on we work on the upper sheet of $\mathbb X$ and identify it with the domain $\C\setminus([q^{-1}, y^{-1}]\cup [y,q]).$ On  $\C\setminus [q^{-1}, q]$
 introduce the Abel map $A(z)=\int_q^z \zeta$. Its properties are
determined by those of $\mathcal R(z)$, that is, we will take into account that
\[
\frac{dz}{\mathcal R(z)}=-\frac{d(z^{-1})}{\mathcal R(z^{-1})},\quad z\in \C\setminus([q^{-1}, y^{-1}]\cup [y,q]);
\]
\be\label{propW2} 
\frac{ds}{\mathcal R_-(s)}=\frac{d(s^{-1})}{\mathcal R_-(s^{-1})},\quad s\in [ y^{-1},q^{-1}]\cup [q,y],\ \ \mathcal R_-(z)=\mathcal R(z+\I 0).
\ee
\noprint{\begin{lemma} The function \eqref{root}, augmented with definition \eqref{W_-}, exhibits the following symmetry
\[
\int_a^b \frac{dz}{\mathcal R_-(z)} =  \left\{ \begin{array}{ll}
- \int_{a^{-1}}^{b^{-1}} \frac{dz}{W_-(z)},  & z \in [-\infty, q^{-1}]\cup [y^{-1},y] \cup [q, \infty] \\
\int_{a^{-1}}^{b^{-1}} \frac{dz}{W_-(z)}, & z \in [q^{-1}, y^{-1}] \cup [y, q].
\end{array}\right.
\]
\end{lemma}
\begin{proof}
By applying the substitution $z=s^{-1}$.
\end{proof}}

\begin{lemma}\label{propabel} The Abel map $A(z)$ satisfies
\begin{align}\label{propA} & A(z^{-1}) = - A(z) + \frac 12, \quad A(q^{-1}) =  \frac 12, \quad A(y) =  \frac {\tau}{2}\ (\Mod \tau),\\ \label{propA0}
& A_\pm(-1) = \frac 14 \mp \frac {\tau}{2}, \quad A(1) =  \frac 14, \quad A(0) = \frac 12 - A(\infty),\\ \label{propA1}
& A_+(z) = - A_-(z),  \quad z \in [q,y], \\ \label{propA2}
&  A_+(z) = A_-(z) - \tau, \quad z \in [y,y^{-1}], \\ \label{propA3}
&  A_+(z) = - A_-(z) + 1,  \quad z \in [y^{-1}, q^{-1}].
\end{align}
\end{lemma}
Associated with $\mathbb X$ is the Riemann theta function
\be\label{Thetta}
\theta(z)=\theta(z\mid \tau)=\sum_{k\in\Z} \exp\big(\pi\I k^2\tau + 2\pi\I k z\big).
\ee
It satisfies $\theta(-z)=\theta(z)$ and 
$\theta(z+ l + k \tau)=\exp(-2\pi \I k z - \pi \I k^2 \tau)  \theta(z)$ for $l, k \in \Z$.
\noprint{To add the missing jump factor $\pm \I$ along $[q,y] \cup [y^{-1},q^{-1}]$ we utilize the function
\[
\gamma(z)= \sqrt[4]{\frac{(q+1)(q^{-1}+1)}{(y+1)(y^{-1}+1)}}
\sqrt[4]{\frac{(y-z)(y^{-1}-z)}{(q-z)(q^{-1}-z)}}.
\]
It satisfies $\gamma(-1)=1$, $\gamma(z^{-1}) = \gamma(z)$,  $(\gamma - \gamma^{-1})(-1)=0$, and solves
the jump problem
\[
\gamma_+(z) =
\begin{cases}
		\I \gamma_-(z)  & \quad z \in [q,y], \\
		-\I \gamma_-(z) , & \quad z \in [y^{-1},q^{-1}].
			\end{cases}
\]}
\begin{lemma}\label{lemvec3}[Vector solution of the model RHP]
On $\C \setminus [q, q^{-1}]$ define
\begin{align}\label{deli}
\delta(z) &= \frac{\theta\big(A(z) -\frac 12+ \frac{tB}{2\pi}-\frac{\Delta}{4\pi}\big)\theta\big(A(z) +\frac{tB}{2\pi} -\frac{\Delta}{4\pi}\big)}
{\theta\big(A(z) -\frac 12\big)\theta\big(A(z)\big)}
\end{align}
and
\be\label{defash}
H(z) = \sqrt[4]{\frac{(y-z)(y^{-1}-z)}{(q-z)(q^{-1}-z)}}.
\ee
Then 
 the vector function
\be\label{uzhe}
m^{mod}(z) = \begin{pmatrix}
\delta(z), & \delta(z^{-1})
\end{pmatrix} \frac{H(z)}{\sqrt{\delta(0)\delta(\infty)}}
\ee
is the unique solution of   the model RH problem. Moreover,
\be\label{defdelta1}
\delta(z) = \frac{\theta\big(2A(z) -\frac 12+ \frac{tB}{\pi}-\frac{\Delta}{2\pi}\mid 2 \tau \big)}
{\theta\big(2A(z) -\frac 12\mid 2\tau\big)}.
\ee
\end{lemma}
\begin{proof} Using the formula 
$\theta (v \mid \tau) \theta(v-\frac12 \mid \tau) = \theta (2v - \frac 12 \mid 2\tau) \theta(\frac12 \mid 2\tau)$  (cf.\ \cite{Dubr})
we can rewrite $\delta(z)$ as a quotient of two theta functions with double period $2\tau$ as it is written in \eqref{defdelta1}. Applying  Lemma \ref{propabel}  to \eqref{deli} and using that
\[ 
\aligned &H(z^{-1})=H(z),\ z\in \C\setminus [q, y]\cup [y^{-1}, q^{-1}]; \qquad H(0)=1;\quad  \\
&H_+(z)=\I H_-(z),\ z\in [q, y];\qquad  H_+(z)=-\I H_-(z),\ z\in [y^{-1}, q^{-1}];
\endaligned\]
it is straightforward to check that $m^{mod}$ \eqref{uzhe} satisfies the jump \eqref{vmodDef} as far as the symmetry and normalization conditions. In fact, \eqref{defash}--\eqref{uzhe} are the $z$-analog of \cite[Equ.\ (5.22)]{emt14}, where the vector model problem solution \eqref{uzhe} was computed on the Riemann surface
$\mathbb M(\xi)$ of the function
\be \label{Mxi}
R^{1/2}(\la)=-\sqrt{(\la^2 -1)(\la-\la_q)(\la - \la_y)}
\ee
with
$\la_q=b-2a=\frac{1}{2}(q+q^{-1})$ and $\la_y=\frac{1}{2}(y+y^{-1})$.
\end{proof}
Recall that $B=B(\xi)$ depends on $n$ and $t$. By \cite[Lem.~5.3]{emt14}, 
\[
2\I t B=-n\Lambda - tU, 
\]
where $\Lambda$ and $U$ are the $\mathfrak b$-periods of the normalized Abel differentials $\Omega_0$ and $\omega_{\infty_+, \infty_-}$ of the second and third kind on $\mathbb M(\xi)$ (cf.\ \cite[Ch.\ 9]{tjac}).
They do not correspond to the respective Abel differentials\footnote{In fact, $\Omega_0=\frac{\pa}{\pa\la}\Omega(\la, \xi) d\la$ and $\omega_{\infty_+, \infty_-}=\frac{\pa}{\pa\la}\omega(\la, \xi) d\la$ from \eqref{gfunc}.} on $\X$, but due to \eqref{propW2} the constants $\Lambda$ and $U$ can easily be expressed in terms of the variable $z$. In particular, 
\[
\Lambda=2\int_y^q \frac{(s-h)(s-h^{-1})}{W_-(s)}\frac{ds}{s},
\]
where $\la_h=\frac{h + h^{-1}}{2}$ is the zero of $\omega_{\infty_+,\infty_-}$.
Note that $\Lambda$ is connected with the Abel map $A(z)$ by (cf.\ \cite{tjac})
\[ 
A(\infty)=A(0) -\frac{\Lambda}{4\pi\I}.
\]

\begin{remark} Observe that by \eqref{propA0},
\[
\delta_{\pm}(-1) = \frac{\theta\big(-\frac 14 \mp \frac \tau 2 + \frac{tB}{2\pi}-\frac{\Delta}{4\pi}\big)\theta\big(\frac14  \mp \frac \tau 2  +\frac{tB}{2\pi} -\frac{\Delta}{4\pi}\big)}
{\theta\big(-\frac 14 \mp \frac \tau 2 \big)\theta\big(\frac 14 \mp \frac \tau 2 \big)}.
\]
This implies that $m^{mod}_{\pm}(-1)= (0,0)$ if $2 t B - \Delta=
 \pi + 2\pi k$, $k\in \Z$. As it is shown in \cite{EPT}, for those pairs of $n$ and $t$ which satisfy
\[
n\Lambda + tU=\I(2k+1)\pi,\quad k\in\Z,
\]
a bounded and invertible  matrix solution of the model jump problem \eqref{vmodDef} with integrable isolated  singularities on the jump contour $[y, y^{-1}]$ does not exist.
\end{remark}

In the next section we propose a matrix model solution $M^{mod}(z)$ which is invertible for all $n$ and $t$, but has poles at the edges of the right background (at $z=1$ and $z=-1$).
We will establish that the determinant of this matrix does not have singularities at these points, and therefore is a nonzero constant. Moreover, $m^{(3)}(z)[M^{mod}(z)]^{-1}$ does not have singularities at these points too, and hence is a suitable vector for the 
conclusive asymptotic analysis.

\section{The matrix model RH problem}\label{sec:5}
Let $\omega(p)=\int_{b-2a}^p\omega_{\infty_+\infty_-}$ be the Abel integral of the third kind on the Riemann surface $\M(\xi)$ of \eqref{Mxi}, as introduced in \cite{emt14}. Let $I(\xi)$ be the closed contour on $\M(\xi)$ with projection on the interval $[\la_y, -1]$, which starts at $\la_y$, passes to 
$-1$ on the upper sheet and returns on the lower sheet. Then 
\be\label{Ome}\omega_+(p)-\omega_-(p)=-\Lambda, \quad p\in I(\xi).
\ee 
We associate $z\in \mathcal Q(\xi):=\{z: |z|<1\}\setminus [q,y]$ with $p=(\la, +)$ on the upper sheet of $\M(\xi)$, and $z^{-1}$, $z\in\mathcal Q(\xi)$, with $p^*=(\la,-)$ on the lower sheet. Calculating the $z$-analog of \eqref{Ome} and taking into account the symmetry property $\omega(p)=-\omega(p^*)$, we obtain that $\E^{ \omega(p)}=:G(z)$, defined on  $z\in \C\setminus [q^{-1},q] $ if and only if $p\in \M(\xi)\setminus I(\xi)$,
admits the representation
\[
G(z)=\exp\left(\int_q^z \frac {(s-h)(s-h^{-1})}{\mathcal R(s)}\frac{ds}{s}\right),\quad z\in \C\setminus [q^{-1},q],
\]
and has the following properties.
\begin{itemize}
\item The function $G(z)$  is holomorphic on $\mathcal E\setminus [y^{-1}, y]$ and satisfies $G(z^{-1})= G^{-1}(z)$.
\item Its jumps are given by
\[\aligned &G_+(z) = G_-(z)\E^{-\Lambda},\qquad z\in [y, y^{-1}],\\ 
& G_\pm(z) = [G_\mp(z^{-1})]^{-1},  \qquad z\in[q,y]\cup[y^{-1}, q^{-1}].
\endaligned\]
\item The following asymptotic expansion is valid,
\be\label{fromG}
G(z)=-\frac{\ti a}{2 z}\Big(1 + 2\ti b z+ O(z^2)\Big), \quad
G(z^{-1})=-\frac{2z}{\ti a}\Big(1 -2\ti b z+O(z^2)\Big).
\ee
Here $\ti a$ and $\ti b$ are the coefficients of the asymptotic expansion for $\omega(p)$ as $p\to\infty_\pm$ 
(cf.\ \cite[Equ.\ (9.44)]{tjac}),
\[
\E^{ \omega(p)}=-\left(\frac{\ti  a}{\la}\right)^{\pm 1}\left(1 +\frac{\ti b}{\la} +O(\la^{-2})\right).
\]
\end{itemize}
Note that in all our considerations the values $y$, $\tau$, $h$ etc.\ depend on $n$ via $\xi$. 
To emphasize the dependence of \eqref{defdelta1} on $n$, we abbreviate
\be\label{defalpha}
\aligned
& \alpha_n(z):=\delta(z)\frac{H(z)}{\sqrt{\delta(0)\delta(\infty)}}=\alpha_n H(z)\frac{\theta\big(2A(z) -\frac 12 - \frac{n\Lambda}{2\pi\I}   - \frac{tU}{2\pi\I}-\frac{\Delta}{2\pi}\mid 2 \tau \big)}
{\theta\big(2A(z) -\frac 12\mid 2 \tau \big)},\\
& \alpha_n:=\frac{\theta \big(2A(\infty) -\frac 12\mid 2 \tau \big) }{\sqrt{\theta\big(2A(\infty) -\frac 12 - \frac{n\Lambda}{2\pi\I}   - \frac{tU}{2\pi\I}-\frac{\Delta}{2\pi}\mid 2 \tau \big)\theta\big(2A(0) -\frac 12 - \frac{n\Lambda}{2\pi\I}   - \frac{tU}{2\pi\I}-\frac{\Delta}{2\pi}\mid 2 \tau\big)}}.
\endaligned\ee
We fix $y$, $\tau$ and $h$ in \eqref{defalpha} and consider this expression shifted to $n+1$,  
\[
\alpha_{n+1}(z)=\alpha_{n+1} H(z)\frac{\theta\big(2A(z) -\frac 12 - \frac{(n+1)\Lambda}{2\pi\I}   - \frac{tU}{2\pi\I}-\frac{\Delta}{2\pi} \mid  2 \tau \big)}
{\theta\big(2A(z) -\frac 12 \mid 2\tau\big)}.
\]
\begin{lemma}\label{lemmat}
The vector function
\be\label{nem}
m^\#(z)=\begin{pmatrix}\beta_n(z), & \beta_n(z^{-1})\end{pmatrix}, \quad \mbox{where $\beta_n(z):=\alpha_{n+1}(z)G(z^{-1})$},
\ee
solves the jump problem of the model RH problem, that is,
\[m^\#_+(z,n,t)=m^\#_-(z,n,t)v^{mod}(z),\quad \mbox{ where}\]
	\begin{equation}\label{vmodDef1}
	v^{mod}(z)=
	\begin{cases}
		\I \si_1, & \quad z \in [q,y],\\
		-\I \si_1, & \quad z \in [y^{-1}, q^{-1}],\\
		\begin{pmatrix}
			\E^{-n\Lambda - t U - \I \Delta} & 0 \\
			0 & \E^{n\Lambda + t U + \I \Delta}
			\end{pmatrix}, & \quad z \in [y, y^{-1}].\\
			\end{cases}
	\end{equation}
It satisfies the symmetry condition	
 $m^\#(z^{-1}) = m^\#(z) \si_1$.
The normalization condition is not fulfilled, instead
	$\lim_{z\to 0} m^\#_1(z)m^\#_2(z)=1$. More precisely, 
	\[ 
	m^\#_1(z)=-\frac{2z}{\ti a}\alpha_{n+1}(0)\big(1+ O(z)\big),\quad 
	m^\#_2(z)=-\frac{\ti a}{2 z}\alpha_{n+1}(\infty)\big(1 + O(z)\big), 
	\ \mbox{as $z\to 0$}.
	\]
\end{lemma}
Introduce two functions defined on $\C\setminus([q^{-1},q]\cup\{1, -1\})$:
\be\label{defPsi}
\aligned
\Psi_1(z)&=\frac 12 m_1^{mod}(z) +\rho(z)m^\#_1(z)=\frac 12 \alpha_n(z) +\rho(z)\beta_n(z),\\
\Psi_2(z)&=\frac 12 m_2^{mod}(z)+\rho(z)m^\#_2(z)=\frac 12 \alpha_n(z^{-1}) +\rho(z)\beta_n(z^{-1}),
\endaligned
\ee
where
\be \label{nui}\rho(z)=-\rho(z^{-1})=\frac{2 K_n}{\ti a\,(z^{-1} - z)}, 
\qquad K_n^{-1}=\alpha_n(0)\alpha_{n+1}(\infty).\ee
\begin{lemma}\label{aboutmatrix} 
\begin{enumerate}
\item The matrix
\be\label{defmmatr}
M^{mod}(z)=\begin{pmatrix} \Psi_1(z) & \Psi_2(z)\\ \Psi_2(z^{-1}) & \Psi_1(z^{-1})\end{pmatrix}, \quad z\in \C\setminus \left([q^{-1}, q]\cup\{1,-1\}\right),
\ee
 is a meromorphic matrix solution for the model jump problem 
 \be\label{mateq} M^{mod}_+(z)= M^{mod}_-(z)v^{mod}(z),\quad z\in[q, q^{-1}],
 \ee with $v^{mod}(z)$ given by \eqref{vmodDef} or by the equivalent matrix \eqref{vmodDef1}. 
 It has simple poles at $z=\pm 1$.
\item  $M^{mod}(z)$ satisfies the symmetry
\be\label{symmat}
M^{mod}(z^{-1})=\sigma_1 M^{mod}(z)\sigma_1.
\ee
\item  The vector function 
$(1,\, 1)M^{mod}(z)$ has removable singularities at $1, -1$ and  integrable singularities at 
$\{q, q^{-1}, y, y^{-1} \}$ of order $O((z- \kappa)^{- \frac14})$ as $z\to\kappa \in \{q, q^{-1}, y, y^{-1} \}$.
\item  The determinant of $M^{mod}(z)$ is constant,
\be\label{determ} \det M^{mod}(z)=1,\quad z\in \C.
\ee
\end{enumerate}
\end{lemma}

\begin{proof}Items (i) and (ii) follow from Lemmas \ref{lemvec3}, \ref{lemmat}. To prove (iii) observe that
\be\label{defPsiti}
\Psi_2(z^{-1})=m_1^{mod}(z) -\rho(z)m^\#_1(z), \quad
\Psi_1(z^{-1})=m_2^{mod}(z)-\rho(z)m^\#_2(z).
\ee
Thus,
\be\label{svyaz}m^{mod}(z)=\begin{pmatrix}m_1^{mod}(z), &m_2^{mod}(z)\end{pmatrix}= (1,\, 1)M^{mod}(z).\ee
The vector function $m^{mod}(z)$ given by \eqref{deli}, \eqref{defash} and \eqref{uzhe}, does not have singularities at $z=\pm 1$, and it has fourth root singularities at $\{q, q^{-1}, y, y^{-1} \}$ which proves (iii). We emphasize that \eqref{svyaz} provides a connection between the unique solution of the vector model RHP and the matrix model problem solution.

(iv) Evaluating $\det M^{mod}(z)$ as $z\to 0$ by use of \eqref{fromG} and \eqref{nui} we get
\be\det M^{mod}(z)=\rho(z)\left(m^\#_1(z)m_2^{mod}(z) - m^\#_2(z) m_1^{mod}(z)\right)\label{if},\ee
that is,
\begin{align*}\det M^{mod}(z)&=-2\rho(z)\alpha_{n+1}(z^{-1})G(z)\alpha_n(z)+ O(z^2),\\
& =\frac{2K_n z}{\ti a}\left(-\frac{\ti a}{2z}\right)\alpha_n(0)\alpha_{n+1}(\infty)(1 + O(z))\\
&= 1 + O(z), \quad z\to 0.
\end{align*}
By \eqref{mateq}, $\det M^{mod}(z)$ does not have jumps in $\C$ and by \eqref{symmat}, it is an even function,
\be\label{ivot}
\det M^{mod}(z^{-1})=\det M^{mod}(z).
\ee 
Therefore, $\det M^{mod}(\infty)=1$. This function is even and bounded outside of small vicinities of $1$ and $-1$ and may have simple poles at $\pm 1$. The singularities at the points $\{q, q^{-1}, y, y^{-1} \}$ are at most of square root order, and therefore  removable.
Since the Abel map $A(z)$ and $G(z)$ are single valued functions in a vicinity of $1$ with $A(1)=1/4$ and $G(1)=1$, we have
\[m^{mod}_1(1)=m^{mod}_2(1),\quad m^\#_1(1)=m^\#_2(1).
\]
This implies together with \eqref{if} the absence of a singularity at $1$. Hence we have a function 
holomorphic in $\C\setminus \{-1\}$ and bounded at infinity, which has at most a simple pole at the only point $z=-1$ and the additional symmetry \eqref{ivot}. The only function which satisfies these properties is a constant. 
\end{proof}
\begin{corollary} The following equality holds,
\be\label{imp22}\lim_{z\to 0}\rho(z)m_2^\#(z)=-\frac{1}{m_1^{\textrm{mod}}(0)}.\ee
\end{corollary}
\begin{proof} By \eqref{if} and \eqref{determ} 
\[
\rho(z)\left(m^\#_1(z)m_2^{mod}(z) - m^\#_2(z) m_1^{mod}(z)\right)\equiv 1.\]
But $m^\#_1(z)\to 0$ as $z\to 0$. This proves \eqref{imp22}.\end{proof}
 Let $\mathcal B$ and $\mathcal B^*$ be small symmetric vicinities of the points $y$ and $y^{-1}$. The precise shape of their boundaries will be chosen in the next section. In $\C\setminus (\ol{\mathcal B}\cup \ol{\mathcal B^*})$ introduce the vector function $\nu(z)=m^{(3)}(z) [M^{mod}(z)]^{-1}$. This function satisfies the symmetry condition $\nu(z^{-1})=\nu (z)\si_1$ but the normalization condition is not identified yet. Instead, due  to symmetry of $\nu$ and the properties of $[M^{mod}(z)]^{-1}$ we have
 \begin{lemma}\label{lem0}
 \be\label{00}\nu_2(0)=\nu_1(\infty)=\frac 12\left(\frac{m^{(3)}_1(0)}{m^{\text{mod}}_1(0)} +\frac{m^{\text{mod}}_1(0)}{m^{(3)}_1(0)} \right):=\tau>0.\ee
 \end{lemma}
 \begin{proof} From \eqref{defPsi}, \eqref{imp22} and the normalization $m_2^{\text{mod}}(0)=[m_1^{\text{mod}}(0)]^{-1}>0$  we observe that 
 \[\Psi_2(0)=\frac 12 m_2^{\text{mod}}(0) -\frac{1}{m_1^{\text{mod}}(0)}=-\frac{1}{2\, m_1^{\text{mod}}(0)},\qquad
 \Psi_1(0) =\frac 12 m_1^{\text{mod}}(0).
 \]
Then \eqref{00} follows from the normalization $m^{(3)}_2(0)=[m^{(3)}_1(0)]^{-1}>0$ and 
\[\nu_2(0)=-\Psi_2(0)m^{(3)}_1(0) + \Psi_1(0)[m^{(3)}_1(0)]^{-1}.
\]
 \end{proof}
\begin{lemma}\label{properr} The  function $\nu(z)$
 does not have singularities in vicinities of the points $q, q^{-1}, 1, -1$.  
\end{lemma}
\begin{proof} By \eqref{vmodDef} and \eqref{m3} the vector $\nu(z)$ does not have jumps in small vicinities of 
$q, q^{-1}$ and $1$. By \eqref{sing4}, \eqref{defash}, \eqref{defalpha}, \eqref{nem}, \eqref{defPsi}, \eqref{defmmatr} and \eqref{determ} we conclude that $\nu(z)=O(z- q^{\pm 1})^{-1/2}$, and therefore it has no singularities at $q$ and $q^{-1}$.

At $z=1$, both $m^{(3)}(z)$ and $M^{mod}(z)$ have no jumps. The same is true for $m^{mod}(z)$,  $m^\#(z)$ and $G(z)$, which means that the equalities 
\begin{align*}
G(z^{-1})&=G^{-1}(z),\quad m^{(3)}_{1}(z)=m^{(3)}_{2}(z^{-1}),\\ 
m^{mod}_{1}(z)&=m^{mod}_{2}(z^{-1}),\quad m^\#_{1}(z)=m^\#_{2}(z^{-1}),
\end{align*}
 can be applied in  a vicinity of $z=1$. The differences
 \[ 
 G(z^{-1})-G(z),\  m^{(3)}_{1}(z)-m^{(3)}_{1}(z^{-1}),\ 
m^{mod}_{1}(z)-m^{mod}_{1}(z^{-1}),\ 
m^\#_{1}(z)-m^\#_{1}(z^{-1}),\] are all of order $O(z-1)$ as $z\to 1$.
Thus, from \eqref{defPsi} and \eqref{defPsiti} it follows that
\[ \aligned 
\Psi_2(z) - \Psi_1(z)&=O(z-1) +\rho(z)(m^\#_1(z) - m^\#_1(z^{-1}))\to\Psi_0,\\
\Psi_1(z^{-1}) - \Psi_2(z^{-1})&= O(z-1) -\rho(z)(m^\#_1(z) - m^\#_1(z^{-1}))\to -\Psi_0,\endaligned\quad z\to 1,
\]
where
\[
\Psi_0=\lim_{z\to 1}\frac{K_n}{\ti a(z^{-1} - z)}\left(\beta_n(z) - \beta_n(z^{-1})\right).
\]
On the other hand, $m^{(3)}(z)$ does not have a jump in a vicinity of $z=1$. Therefore, by the symmetry property, $m_1^{(3)}(1)=m_2^{(3)}(1^{-1})=m_2^{(3)}(1)$.
Hence
\[
\aligned\nu(z)&=m^{(3)}(z)\begin{pmatrix} \Psi_1(z^{-1}) & -\Psi_2(z)\\ -\Psi_2(z^{-1}) & \Psi_1(z)\end{pmatrix}\\ &=\begin{pmatrix}m_1^{(3)}(z)\Psi_1(z^{-1}) - m_2^{(3)}(z)\Psi_2(z^{-1}), &  m_2^{(3)}(z)\Psi_1(z) -m_1^{(3)}(z)\Psi_2(z)
\end{pmatrix}\\
& \to \left(\Psi_0 m_1^{(3)}(1)\right) \begin{pmatrix}1, & 1\end{pmatrix}, \quad z\to 1.\endaligned
\]
It remains to investigate the behavior of $\nu(z)$ near $z=-1$. Since $m^{(3)}(z)$ and $M^{mod}(z)$ have the same constant jump $v^{3}(z)=v^{mod}(z)=\E^{(2\I t B - \I\Delta)\si_3}$ in a vicinity of this point, we conclude that $\nu(z)$ does not have jumps here, and therefore $z=-1$ is an isolated singularity, which is at most a simple pole.   From the symmetry condition it follows that both components $\nu_1(z)$ and $\nu_2(z)$ of $\nu(z)$ have the same behavior, either simple poles or removable singularities. To prove that $-1$ is in fact a removable singularity, it suffices to check that
\[\aligned f(z)&=\nu_1(z)\nu_2(z)\\
&=\left(m_1^{(3)}(z)\Psi_1(z^{-1}) - m_2^{(3)}(z)\Psi_2(z^{-1})\right)\left(  m_2^{(3)}(z)\Psi_1(z) -m_1^{(3)}(z)\Psi_2(z)\right)
\endaligned\]
increases not faster than $o((z+1)^{-2})$ from some direction. The behavior of $f(z)$ is determined by the summand which contains $\rho^2(z)$ (cf.\ \eqref{defPsi} and \eqref{defPsiti}). Computing this term we get
\[\aligned 
f(z) & \sim \rho^2(z)\Big([m^\#_2(z)]^2[m_1^{(3)}(z)]^2 + [m^\#_1(z)]^2[m_2^{(3)}(z)]^2\\
& - 2 m^\#_1(z)m^\#_2(z)m_1^{(3)}(z)m_2^{(3)}(z)\Big)=\rho^2(z)\ti f(z).
\endaligned\]
The function $\ti f(z)$ has finite limiting values on the sides of the contour $[\mathfrak r, \mathfrak r^{-1}]$, and in particular at $z=-1$.
Using the symmetry condition we get 
$m^\#_{1,\pm}(-1)=m^\#_{2,\mp}(-1)$, $m_{1,\pm}^{(3)}(-1)= m_{2,\mp}^{(3)}(-1)$, therefore
\[m^\#_{1, \pm}(-1)=m^\#_{2,\pm}\E^{\mp(2\I t B-\I\Delta)}, \quad  m_{1, \pm}^{(3)}(-1)= m_{2,\pm}^{(3)}\E^{\mp(2\I t B-\I\Delta)},
\]
that is, $m^\#_{2,\pm}(-1)m_{1,\pm}^{(3)}(-1)= m^\#_{1,\pm}(-1)m_{2,\pm}^{(3)}(-1)$.
Thus
\[\aligned\ti f_\pm(-1)&=[m^\#_{2,\pm}(-1)]^2[m_{1,\pm}^{(3)}(-1)]^2 + [m^\#_{1,\pm}(-1)]^2[m_{2,\pm}^{(3)}(-1)]^2\\
& - 2m^\#_{1,\pm}(-1)m^\#_{2,\pm}(-1)m_{1,\pm}^{(3)}(-1)m_{2,\pm}^{(3)}(-1)=0.
\endaligned\]
\end{proof}
\begin{remark}\label{remim} 
The jump matrix $v^{(3)}(z)$ given by \eqref{vi30}, \eqref{vi3} satisfies the symmetry 
\be\label{mainv3} v^{(3)}(z)=\si_1 v^{(3)}(z^{-1})\sigma_1\ee on the contour $\Xi$ (cf.\ \eqref{defXi}), while on  $[q, q^{-1}]$ it satisfies   $[v^{(3)}(z)]^{-1}=\si_1 v^{(3)}(z^{-1})\sigma_1$. 
Therefore we reverse the orientation on the part $[-1, q^{-1}]$ such that the property  \eqref{mainv3} is satisfied on the whole jump contour of RHP~\ref{3}.
\end{remark}

\section{Solution of the parametrix RH problems}\label{s:par}
 In this section we solve local RHPs in vicinities of the points $y, y^{-1}$, where the error jump matrix (introduced at the end of Section \ref{sec:tomod}) is not small as $t\to\infty$. Recall that on the contours with $y^{-1}$ as a nodal point we have (with the new orientation on $[q^{-1}, -1]$ as introduced in Remark \ref{remim})
\[ 
v^{(3)}(z)=
\begin{cases}
\I \si_1, & \quad z \in [q^{-1}, y^{-1}],\\
		\E^{(-2 \I tB + \I\Delta)\sigma_3}& 
		\quad z \in [ \mathfrak r^{-1}, -1],\\
	\begin{pmatrix}
		\frac{F_+(z)}{F_-(z)}\E^{t(g_+(z) - g_-(z))} & U(z)\E^{-2t \re g(z)} \\
		0 &  \frac{F_-(z)}{F_+(z)}\E^{t(g_-(z) - g_+(z))}
		\end{pmatrix}, & \quad z \in [  y^{-1}, \mathfrak r^{-1}],\\[3mm]
					\begin{pmatrix}
			1 & 0 \\
			\frac{e^{2tg(z)}}{U_1(z)}  & 1
			\end{pmatrix}, & \quad z \in \mathcal C^*,
				\end{cases}
\]
where we denoted
\be\label{UU1} U(z)=\I |\chi(z)|\Pi^2(z)F_+(z)F_-(z),\qquad U_1(z)=\Pi^2(z)F^2(z)X(z),\ee
and used \eqref{imp33}.
Respectively,
\[v^{err}(z)=v^{(3)}(z) -  v^{mod}(z)=
\begin{cases}
	\begin{pmatrix}
		0 & U(z)\E^{-2t \re g(z)} \\
		0  & 0
		\end{pmatrix}, & \quad z \in [  y^{-1}, \mathfrak r^{-1}],\\[2mm]
					\begin{pmatrix}
			0 & 0 \\
			\frac{e^{2tg(z)}}{U_1(z)}  & 0
			\end{pmatrix}, & \quad z \in \mathcal C^*,
				\end{cases}\]
does not vanish as $t\to\infty$ since $\re g(y^{-1})=0$, $U(y^{-1})U_1(y^{-1})\neq 0$. 
The local (parametrix) RHPs are similar to those of the KdV shock wave 
analysis (see e.g.\ \cite[Sec.\ 7]{EPT}).  
Consider first the point $y^{-1}$. Let $\mathcal B^*=\mathcal B^*(\varepsilon)$ be a neighborhood of $y^{-1}$ such that its boundary contains $\mathfrak r^{-1}$ given by \eqref{est99}.
 To describe the boundary of $\mathcal B^*$, introduce a local change of variables
 \begin{equation}\label{wDef}
w^{3/2}(z) = \frac{3  t}{2}\left(g(z)-g_{\pm}(y^{-1})\right), \quad z \in \mathcal{B}^*,
\end{equation}
with the cut along the interval $J=[q^{-1},y^{-1}]\cap\ol{\mathcal B^*}.$
From \eqref{g} and item (b) of Lemma \ref{propg} we have
\[\aligned
\frac{3}{2}(g(z) - g_{\pm}(y^{-1})) &= \frac{3}{2}\int \limits_{y^{-1}\pm\I 0}^{z}\, P(s) \ti Q(s)\ \sqrt{y^{-1}-s}\,\frac{ds}{2s}\\
& = \frac{2}{ y}P(y)\ti Q(y)(z-y^{-1})^{3/2}\left( 1 + o(1) \right),\quad z\to y^{-1},\endaligned
\]
with $P(s)$ given by \eqref{defPe} and \[\ti Q(s):= \sqrt{\frac{s-y}{(s-q)(s-q^{-1})}}.\] Evidently, $\ti Q(y)>0$. For $\xi\in (\xi_{cr}^\prime, \xi_{cr})$ we have (see \cite{emt14}) $P(s)=s +s^{-1} - \zeta - \zeta^{-1}$, where $\zeta= \zeta(\xi)\in (-1, y)$.
Thus, $P(y^{-1})<0$ and $ \mathcal T:= 2 y^{-1}P(y)\ti Q(y)>0$. 
  Then 
  \be\label{double ve} 
  w(z)= \mathcal T^{2/3}t^{2/3}\,(z-y^{-1})(1 + o(1)),\quad \mbox{as}\ \ z\to y^{-1}, \quad (\mathcal T t)^{2/3}>0.
  \ee 
  Hence $w(z)$ is a holomorphic function in $\mathcal B^*$.
\begin{figure}[ht]
\begin{tikzpicture}[scale=0.8]

\draw (2.5,0) circle (1.5cm);
\draw (1,0) -- (4, 0);
\filldraw (2.5,0) circle (1pt) node[below right]{$y^{-1}$};
\draw[out=120,in=340] (2.5,0) to (1.39,1);
\draw[out=240,in=20] (2.5,0) to (1.39,-1);
\draw (4.2,-1) node {$\partial \mathcal B^*$};
\draw [->] (2.5, 1.5) -- (2.49,1.5);

\draw [->] (3.2, 0) -- (3.3, 0);
\draw [->] (1.6, 0) -- (1.7, 0);
\draw [->] (2.01, 0.65) -- (2, 0.659);
\draw [->] (2, -0.659)-- (2.01, -0.65);

\node at (3.3, 0.3) {$J^\prime$};
\node at (1.7, 0.3) {$J$};
\node at (2.2,0.9) {$\mathcal L_1$};
\node at (2.2,-0.9) {$\mathcal L_2$};

\draw[->] (5,0.5) arc (120:60:20mm);
\draw (6, 1) node {$w$};

\draw (9.5,0) circle (1.8cm);
\filldraw (9.5,0) circle (1pt) node[below right]{$0$};
\draw (7.7,0) -- (11.3,0);
\draw (9.5,0) -- (8.5, 1.5);
\draw (9.5,0) -- (8.5, -1.5);

\draw [->] (8.5, 0) -- (8.6, 0);
\draw [->] (10.3, 0) -- (10.4, 0);
\draw [->] (8.91, 0.889) -- (8.9, 0.9);
\draw [->](8.9, -0.9) -- (9.03, -0.7);
\draw [->] (9.5, 1.8) -- (9.49,1.8);

\draw (11.5,-1) node {$\partial \mathcal O$};

\end{tikzpicture}
\caption{The local change of variables $w(z)$.}
\label{fig:5}
\end{figure}
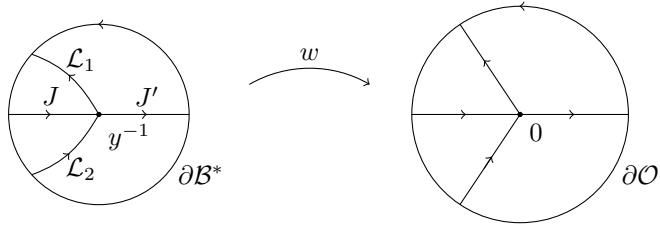

Untill now we did not specify the particular shape of the boundary of $\mathcal B^*$ and the shape of the contour $\mathcal C^*$ inside $\mathcal B^*$. Treating $w(z)$  as a conformal map, let us think of $\mathcal B^*$ as a pre-image of a disc $\mathcal O$ of radius $\mathcal T^{2/3}|y^{-1} - \mathfrak r^{-1}| t^{2/3}$ centered at the origin. Then $w(z)$ maps the interval $J=[q^{-1}, y^{-1}]\cap\ol {\mathcal B^*}$ to the negative half axis and $J^{\prime}=[y^{-1}, \mathfrak r^{-1}]$ to the positive half axis.  We also choose the contour $\mathcal C^*\cap\mathcal B^*$ to be contained in the pre-image of the rays 
$\arg w=\pm\frac{2\pi \I}{3}$ and divide it in two parts $\mathcal L_1$  and $\mathcal L_2$, with orientation as depicted in Fig. \ref{fig:5}. 
With the new orientation,
\be\label{vmodn} v^{mod}(z)=\begin{cases}
\I\sigma_1&\quad z\in J\\
\E^{(-2\I tB + \I \Delta)\sigma_3} &\quad z\in J^\prime\end{cases}.
\ee
In $\mathcal B^*$ we introduce the function
\be\label{defr}
r(z):=\frac{\E^{\mp\frac{\I\pi}{4}}\,\E^{\mp \I t B}}{\sqrt{X(z)}\Pi(z)F(z)}, \quad z\in\mathcal B^*\cap\{z: \pm \im z>0\}.\ee
 Since \be\label{chipm}X_\pm(z)=\mp\I |\chi(z)|,\quad z\in J\cup J^\prime, \ee 
 we see that 
  \[ r_\pm(z)=\frac{\E^{\mp \I t B}}{\sqrt{|\chi(z)|}\Pi (z) F_\pm(z)},\quad  z\in J\cup J^\prime.\]
By items (b) and (c) of Lemma \ref{lemma:F}, taking into account the change of direction for the contour in (c), we get 
\beq\label{imp45} r_+(z)r_-(z)=1, \quad z\in J; \qquad r_+(z)=r_-(z)\E^{\I \Delta -2\I t B},\quad z\in J^\prime.\eeq
Due to \eqref{chipm},  
$\sqrt{X_+(z)X_-(z)}=|\chi(z)|$ for $z\in J^\prime$. By use of \eqref{UU1} and \eqref{wDef}, we have for the off-diagonal 
elements of $v^{err}$ 
\be\label{toto}\aligned &U(z)\E^{-2t\re g(z)}=\I|\chi(z)|\Pi^2(z)F_+(z)F_-(z)\E^{-t(g_-(z) +g_+(z))}\\
&=\I\,\frac{\E^{-t(g_+(z) -g_+(y^{-1}))}\E^{-t(g_-(z) -g_-(y^{-1}))}}{r_+(z) r_-(z)}=\frac{\I\,\E^{-\frac 4 3 w^{3/2}(z)}}{r_+(z) r_-(z)},
\endaligned\ee
\be\label{toto1}\aligned \frac{\E^{2 t g(z)}}{U_1(z)}&= \pm \I\, r^2(z)\E^{\pm 2 \I t B + 2 t g(z)}=\pm \I\, r^2(z)\E^{2t g(z)-g_\pm(y^{-1})}
\\ &= \pm\I r^2(z) \E^{\frac 4 3 w^{3/2}(z)},\quad  \pm \im z>0.\endaligned\ee
We redefine $m^{(3)}(z)$, $m^{mod}(z)$ and the matrix $M^{mod}(z)$ inside $\mathcal B^*$ by 
\[
\widehat m^{(3)}(z)= m^{(3)}(z)r(z)^{-\si_3}, \quad \widehat m^{mod}(z)=m^{mod}(z)r(z)^{-\si_3},
\] 
\be\label{1} \widehat M^{mod}(z)=M^{mod}(z)r(z)^{-\si_3},\quad z\in\mathcal B^*.
\ee
Using \eqref{vmodn}, \eqref{imp45}, \eqref{toto} and \eqref{toto1}, we obtain that inside $\mathcal B^*$ 
\[ \widehat m^{(3)}_+(z)=\widehat m^{(3)}_-(z) \widehat v^{(3)}(z), \quad \widehat M_+^{mod}(z)=\widehat M_-^{mod}(z)\widehat v^{mod}(z),\]
where \[\widehat v^{mod}(z)=\I\sigma_1, \quad z\in J;\qquad \widehat v^{mod}(z)=\id, \quad z\in J^\prime;\]
\[
\widehat v^{(3)}(z)=\begin{cases} \I\si_1, &\quad z\in J,\\
\begin{pmatrix} 1 & \I \E^{-\frac 4 3 w^{3/2}(z)}\\0&1\end{pmatrix}, &\quad z\in J^\prime,\\
\begin{pmatrix} 1 & 0\\ -\I\E^{4/3 w^{3/2}(z)}& 1\end{pmatrix}, &\quad z\in \mathcal L_1,\\
\begin{pmatrix} 1&0\\ \I\E^{4/3 w^{3/2}(z)}& 1\end{pmatrix}, &\quad z\in \mathcal L_2.
\end{cases}
\]
By \eqref{double ve} we conclude that $w^{1/4}(z)$ has the following jump along $J$,
\[w_+^{1/4}(z)=w_-^{1/4}(z) \I,\quad z\in J.\]
Recall that $ \mathcal O=w(\mathcal B^*)$. It is now straightforward to check that  the matrix 
\[N(w)=\frac{1}{\sqrt{2}}\begin{pmatrix} w^{1/4}& w^{1/4}\\  - w^{-1/4} & w^{-1/4}\end{pmatrix}, \quad w\in {\mathcal O},\]
 solves the jump problem 
 \[N_+(w(z))=\I N_-(w(z))\sigma_1,\quad z\in J.\]
Therefore, in  $ \mathcal B^*$ we have $\widehat M^{mod}(z)=  \mathcal H(z) N(w(z))$, where $\mathcal H(z)$ is a holomorphic matrix function in $ \mathcal B^*$. 
Since $\det N(w)=\det [r(z)^{\sigma_3}]=1$, we have
 \beq\label{obr3}\det  \mathcal H(z)=\det M^{\text{mod}}(z)=\det \widehat M^{mod}(z)=1.
 \eeq
According to \eqref{1} we get
\[
M^{mod}(z)= \mathcal H(z)N(w(z))r(z)^{\sigma_3},\quad z\in \pa\mathcal B^*.
\]
By property $\rm{(c)}$ of Lemma \ref{propg},  $w_+(z)^{3/2} =-w_-(z)^{3/2}$ for $z \in J$, that is,
\[\widehat v^{(3)}(z)=d_-(z)^{-\sigma_3} \mathcal S\, d_+(z)^{\sigma_3},\quad z\in \mathcal B^*,
\] 
where
\[
d(z):= \tilde d(w(z)), \quad \tilde d(w)=\E^{ 2/3 w^{3/2}},
\]
and
\[ \mathcal S=\left\{\begin{array}{ll} \mathcal S_1, & \quad z\in\mathcal L_1,\\
\mathcal S_2, & \quad z\in J,\\
\mathcal S_3, & \quad z\in\mathcal L_2,\\
\mathcal S_4, & \quad z\in J^\prime.\end{array}\right.
\]
Here
\[\mathcal S_1=\begin{pmatrix} 1& 0\\ \I &1\end{pmatrix};\quad \mathcal S_2=\begin{pmatrix}0&\I\\ \I & 0\end{pmatrix};\quad \mathcal S_3=\begin{pmatrix} 1& 0\\ -\I &1\end{pmatrix};\quad \mathcal S_4=\begin{pmatrix}1&\I\\ 0& 1\end{pmatrix}.\]
Consider $\mathcal S$ as the jump matrix on the contour $w(J\cup J^\prime\cup \mathcal L_1\cup \mathcal L_2)$
in $\mathcal O$.  Let  $\mathcal A(w)$ be the matrix solution of the jump problem 
\[
\mathcal A_+(w)=\mathcal A_-(w) \mathcal S,\quad w\in w(J\cup J^\prime\cup \mathcal L_1\cup \mathcal L_2),
\]
satisfying the boundary condition
\be\label{constc}
\mathcal A(w)=  N(w)\Psi(w)\tilde d(w)^{-\sigma_3},\quad w\in\partial \mathcal O,\quad t\to\infty,
\ee
where \be\label{ppsi}\Psi(w)=\id +\frac{C}{w^{3/2}}\Big(1 + O\big(w^{-3/2}\big)\Big),\quad w\to \infty,\ee
is an invertible matrix, and
 $C$ is a constant matrix with respect to $w$, $t$ and $\xi$.
The solution $\mathcal A(w)$ can be expressed via  Airy functions and their derivatives in a standard manner (see, for example, \cite{dkmvz}, \cite[Ch.\ 3]{Bleher}, \cite{GGM} or \cite{AELT}). In particular,  in the sector between the contours $w(J^\prime)$ and $w(\mathcal L_1)$ in $\mathcal O$ we have
\[ \mathcal A(w)=\mathcal A_1(w)=\sqrt{2\pi}\begin{pmatrix} -y_1^\prime(w)&\I y_2^\prime(w)\\-y_1(w)& \I y_2(w)\end{pmatrix},\]
where
$y_1(w)=\mathrm{Ai}(w)$ and $y_2(w)=\E^{-\frac{2\pi \I}{3}} \mathrm{Ai}(\E^{-\frac{2\pi \I}{3}} w)$. In the sector between the lines $w(\mathcal L_1)$ and $w(J)$ we get 
\[\mathcal A(w)=\mathcal A_2(w)=\mathcal A_1(w) \mathcal S_1=\sqrt{2 \pi}\begin{pmatrix}y_3^\prime(w) & \I y_2^\prime(w)\\ y_3(w) & \I y_2(w)\end{pmatrix},\]
where 
$y_3(w)=\E^{\frac{2\pi \I}{3}} \mathrm{Ai}(\E^{\frac{2\pi \I}{3}} w)$. Here we used the standard equality
$y_1(w) + y_2(w) + y_3(w)=0$.
Changing orientation on $J$ and $\mathcal L_2$ we obtain between $w(J)$ and $w(\mathcal L_2)$ 
\[
\mathcal A(w)=\mathcal A_3(w)= -\I \mathcal A_2(w)\sigma_1,
\]
and between the lines $w(\mathcal L_2)$ and $w(J^\prime)$, correspondingly,
\[\mathcal A(w)=\mathcal A_4(w)=\mathcal A_3(w) \mathcal S_3^{-1}.\] The last conjugation with the  matrix $S_4$ will lead to matrix $\mathcal A_1(w)$ again, because
$\mathcal S_1\mathcal S_2^{-1}\mathcal S_3^{-1}\mathcal S_4=\id.$ Note that the constant matrix $C$ in \eqref{constc} is the same for all regions, 
\[ C=\frac{1}{48}\begin{pmatrix} -1& 6\\ -6& 1\end{pmatrix}.\]
The precise formulas for $\mathcal A_j(w)$  are in fact not important for us. 
The matrix  
\[M^{par}(z):=  \mathcal H(z)\mathcal A(w(z))d(z)^{\sigma_3},\quad z\in\mathcal B^*,
\] 
solves in $\mathcal B^*$ the jump problem
\[ M_+^{par}(z)= M_-^{par}(z)\widehat v^{(3)}(z),\quad z\in J\cup J^\prime\cup \mathcal L_1\cup\mathcal L_2,\]
and satisfies for sufficiently large $t$ the boundary condition 
\[\aligned M^{par}(z)&=\mathcal  H(z) N(w(z))\Psi(w(z))=\widehat M^{mod}(z)\Psi(w(z))\\ &
=M^{mod}(z)r(z)^{-\si_3} \Psi(w(z)),\quad z\in\pa\mathcal B^*.\endaligned\]
Note that \eqref{ppsi} and \eqref{wDef} yield
\[\det \Psi(w(z))=1 + O(t^{-1}), \quad z\in \mathcal B^*,\quad t\to\infty,
\]
uniformly with  respect to $\xi\in \mathcal I_\varepsilon$.  This implies with \eqref{obr3} invertibility of $M^{par}(z)$ in $\ol{\mathcal B^*}$.
In summary,  we constructed a  matrix with the following properties:
\begin{lemma}\label{lempar} 
The vector function 
\be\label{nusa}\nu(z)=\widehat  m^{(3)}(z)M^{par}(z)^{-1}=m^{(3)}(z)r(z)^{-\si_3}M^{par}(z)^{-1},\quad z\in \mathcal B^*,\ee
does not have jumps and isolated singularities in $\mathcal B^*$, it is holomorphic there. The function $\nu(z)$ has piecewise continuous limiting values as $z$ approaches $\pa\mathcal B^*$ from inside, given by 
\be\label{main67}
\nu(z)=m^{(3)}(z)r(z)^{-\si_3} \Psi(w(z))^{-1} r(z)^{\si_3} M^{mod}(z)^{-1},\quad z\in \pa\mathcal B^*.
\ee
\end{lemma}
Let   $\mathcal B:=\{z: z^{-1}\in \mathcal B^*\}$. Define $\nu(z)$ in $\mathcal B^*$ by the symmetry 
$\nu(z)=\nu(z^{-1})\sigma_1$, $z\in \mathcal B$. With  this extension,  $\nu(z)$ is holomorphic in $\mathcal B$.
Let us extend the definition of $\nu(z)$ to $\C \setminus (\ol{\mathcal B^*}\cup
\ol{\mathcal B})$ by 
\be\label{defnui} \nu(z)=m^{(3)}(z)M^{mod}(z)^{-1},\quad  z\in\C\setminus (\ol{\mathcal B^*}\cup
\ol{\mathcal B}).\ee

Theorem \ref{properr}  implies that this function does have jumps on the contour
$[q, -1]\cup [q^{-1},-1]$ outside of $\ol{\mathcal B^*}\cup
    \ol{\mathcal B}$. We label the parts of $\mathcal C$ and $\mathcal C^*$ outside $\mathcal B$ and $\mathcal B^*$ 
by $\mathcal C_{\mathcal B}$ and $\mathcal C_{\mathcal B}^*$, see Fig.~\ref{fig:Omega}. The jumps of  $\nu(z)$ on 
$\Gamma\cup\mathcal C_{\mathcal B}\cup\mathcal C_{\mathcal B}^*$ (cf.\ \eqref{defK})  are exponentially small with respect to $t\to \infty$. Let us compute the jump of this vector on the boundaries $\pa\mathcal B$ and $\pa\mathcal B^*$, which we treat as clockwise oriented contours. Since
neither $m^{(3)}(z)$, $r(z)=r^{-1}(z^{-1})$ nor $M^{mod}(z)$ have jumps on these contours, we obtain  from \eqref{main67} and \eqref{defnui}
\[m^{(3)}(z)=\nu_-(z) M^{mod}(z) r(z)^{-\si_3}\Psi(w(z)) r(z)^{\si_3}=\nu_+(z) M^{mod}(z), \quad z\in \pa\mathcal B^*.\]
Taking into account \eqref{ppsi} we find the jump
\[ \nu_+(z)=\nu_-(z) (\id + W(z)), \quad z\in \pa\mathcal B^*\cup\pa\mathcal B,\]  
where 
\be\label{defW}
\aligned  W(z)&= M^{mod}(z) r(z)^{-\si_3}\big(\Psi(w(z))-\id\big) r(z)^{\si_3}M^{mod}(z)^{-1}, \quad z\in \pa\mathcal B^*;\\
 W(z)&=\si_1 W(z^{-1})\si_1, \qquad z\in \pa\mathcal B.
 \endaligned
 \ee
 The jump contour 
\[
 \mathcal K:= \Gamma\cup\mathcal C_{\mathcal B}\cup\mathcal C_{\mathcal B}^*\cup\pa \mathcal B\cup\pa\mathcal B^*,
 \] 
 for the RHP associated with the error vector $\nu(z)$ is depicted in Fig.\ \ref{fig:finalcontour}.
\vskip 3 mm 

\begin{figure}[ht] 
\begin{tikzpicture}
\clip (-3.5,-1.7) rectangle (8.6,1.7);
\draw[black!30, densely dotted] (4.8,0) circle (3.5cm);
\draw (4.8,0) circle (3.35cm);
\draw (4.8,0) circle (3.67cm);
\draw [->, thin] (8.15, 0.02) -- (8.15,0.03);
\draw [->, thin] (8.47, 0) -- (8.47,-0.01);


\draw[out=60, in=90] (3.3,0) to (4.5,0);
\draw[out=300, in=270] (3.3,0) to (4.5,0);
\draw [->, thin] (4.01, 0.33) -- (4,0.33);

\draw[black!30, densely dotted] (2.1,0) -- (4.3,0);
\draw[black!30, densely dotted] (-2.7,0) -- (-0.1,0);

\draw[out=120, in=90] (2.9,0) to (1.9,0);
\draw[out=240, in=270] (2.9,0) to (1.9,0);
\draw [->, thin] (2.3, 0.275) -- (2.29,0.275);
\draw[out=60, in=90] (-1.1,0) to (0.2,0);
\draw[out=300, in=270] (-1.1,0) to (0.2,0);
\draw [->, thin] (-0.29, 0.35) -- (-0.3,0.35);

\draw[thin] (4.7,0) circle (0.10cm);
\filldraw (4.7,0) circle (0.5pt) node[above]{\small{${\T_j}$}};
\draw [->, thin] (4.68, 0.1) -- (4.67,0.1);
\draw[thin] (1.7,0) circle (0.10cm);
\filldraw (1.7,0) circle (0.5pt) node[above]{\small{${\T_k}$}};
\draw [->, thin] (1.68, 0.1) -- (1.67,0.1);
\draw (0.7,0) ellipse (0.12cm and 0.10cm); 
\filldraw (0.7,0) circle (0.5pt) node[above]{\small{${\T_k^*}$}};
\draw [->, thin] (0.68, 0.1) -- (0.67,0.1);
\draw (-3.2,0) ellipse (0.12cm and 0.10cm); 
\filldraw (-3.2,0) circle (0.5pt) node[above]{\small{${\T_j^*}$}};
\draw [->, thin] (-3.22, 0.1) -- (-3.23,0.1);

\draw (-1.6,0) circle (0.5cm);
\draw (3.3,0) ellipse (0.4cm and 0.5cm);
\draw [->, thin] (-1.58, 0.5) -- (-1.57,0.5);
\draw [->, thin] (3.32, 0.5) -- (3.33,0.5);

\filldraw (-2.7,0) circle (1pt); 
\node at (-2.5,-0.14) {\tiny{$q^{-1}$}};
\filldraw(-1.6,0) circle (1pt); 
\node at (-1.45,-0.14) {\tiny{$y^{-1}$}};
\filldraw (-1.1,0) circle (1pt); 
\node at (-0.98,-0.37) {\tiny{$\mathfrak r^{-1}$}};
\filldraw (-0.1,0) circle (1pt);
\node at (-0.07, -0.14) {\tiny{$q_1^{-1}$}};
\filldraw (1.3,0) circle (1pt) node[below]{\tiny{$-1$}};
\filldraw (2.1,0) circle (1pt); 
\node at (2.15,-0.14) {\tiny{$q_1$}};
\filldraw (2.9,0) circle (1pt); 
\node at (2.84,-0.2) {\tiny{$\mathfrak r$}};
\filldraw (3.3,0) circle (1pt) node[below]{\tiny{$y$}};
\filldraw (4.3,0) circle (1pt); 
\node at (4.3,-0.14) {\tiny{$q$}};
\filldraw (8.3,0) circle (1pt) node[below]{\tiny{${1}$}};
\node at (-2.4,0.6) {$\mathcal \mathcal C_{\mathcal B}^*$};
\draw [->, thin] (-2.31, 0.35) -- (-2.32,0.35);
\node at (-1.5,0.8) {$\pa \mathcal B^*$};
\node at (-0.5,0.6) {$\mathcal C_{\mathfrak r}^*$};
\node at (2.3,0.5) {$\mathcal C_{\mathfrak r}$};
\node at (3.3,0.8) {$\pa \mathcal B$};
\node at (4,0.6) {$\mathcal C_{\mathcal B}$};
\node at (1.9,-0.9) {$\mathcal C_\epsilon$};
\node at (1,-1) {$\mathcal C_\epsilon^*$};

\draw[out=120, in=90] (-1.6,0) to (-2.9,0); 
\draw[out=240, in=270] (-1.6,0) to (-2.9,0);
\path[clip, draw] (-1.6,0) circle (0.5cm) (3.3,0) ellipse (0.4cm and 0.5cm);
\path[fill=white] (-1.6, 0) circle (0.5cm) (3.3,0) ellipse (0.4cm and 0.5cm);
\draw[black!30, densely dotted] (-2.7,0) -- (-0.1,0); 
\draw[black!30, densely dotted] (2.1,0) -- (4.3,0);
\filldraw (-1.1,0) circle (1pt); 
\filldraw(-1.6,0) circle (1pt);
\filldraw (2.9,0) circle (1pt);
\node at (-1.45,-0.14) {\tiny{$y^{-1}$}};
\filldraw (3.3,0) circle (1pt) node[below]{\tiny{$y$}};

\path[clip, draw] (-1.6,0) circle (0.5cm) (3.3,0) ellipse (0.4cm and 0.5cm) ;
                  \draw[out=120, in=90, densely dotted] (-1.6,0) to (-2.9,0); 
		\draw[out=240, in=270, densely dotted] (-1.6,0) to (-2.9,0);
		\draw[out=60, in=90, densely dotted] (3.3,0) to (4.5,0);
		\draw[out=300, in=270, densely dotted] (3.3,0) to (4.5,0);
\draw [->, thin] (-1.58, 0.5) -- (-1.57,0.5);
\draw [->, thin] (3.32, 0.5) -- (3.33,0.5);		
\end{tikzpicture}
\caption{$\mathcal K = \bigcup_{j=1}^N(\T_j \cup \T_j^*) \cup \mathcal C_\epsilon \cup \mathcal C_\epsilon^* \cup \mathcal C_{\mathfrak r} \cup \mathcal C_{\mathfrak r}^* 
\cup \mathcal C_{\mathcal B}\cup\mathcal C_{\mathcal B}^*\cup \pa \mathcal B\cup\pa\mathcal B^*$.} \label{fig:finalcontour}
\end{figure}
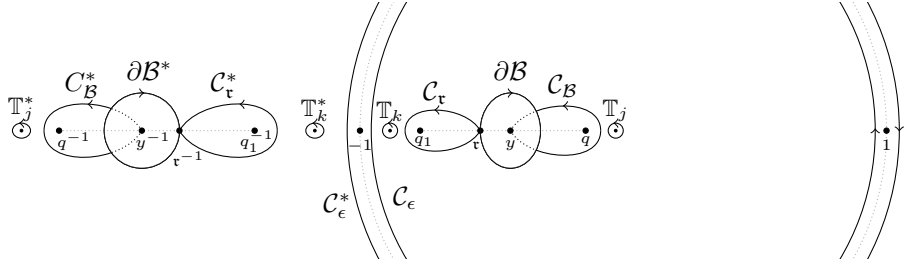

Merging the results of this  section with Lemmas \ref{properr} and \ref{lem0} we have proven 
 \begin{theorem}\label{lemnu}
 The vector $\nu(z)$ is a holomorphic  function in the domain  $\C\setminus \mathcal K$
 and bounded on the closure of this domain.  On the contour $\mathcal K$,  $\nu(z)$ has the jump 
 \be\label{merr}\nu_+(z)=\nu_-(z) (\id + W(z)),
 \ee 
 where $W(z)$ is given by \eqref{defW} on $\pa\mathcal B\cup\pa\mathcal B^*$ and 
 \be \label{we}
 W(z)=M^{mod}_-(z) \big(v^{(3)}(z)-\id\big)M^{mod}_+(z)^{-1},\quad z\in \Gamma\cup\mathcal C_{\mathcal B}\cup\mathcal C_{\mathcal B}^*=\mathcal K\setminus (\pa\mathcal B\cup\pa\mathcal B^*).
 \ee
 On $\mathcal K$, $W(z)$ has the symmetry \be\label{symW}W(z^{-1})=\si_1 W(z)\si_1,\quad z\in \mathcal K.\ee
 The vector $\nu(z)$ satisfies \be\label{symer}\nu(z^{-1})=\nu(z)\si_1.\ee Moreover, 
 \be\label{symnu}\nu_-(z^{-1})=\nu_-(z)\si_1, \quad z\in \mathcal K.\ee
 In addition,
 \[
 \nu_2(0)=\nu_1(\infty)=\tau>0.
 \]
 \end{theorem}
We also observe the following estimate. By definition, 
\[\min\textrm{dist}\, \left(y^{-1}, \pa \mathcal B^*\right)> C(\varepsilon) >0, \ \mbox{uniformly with respect to}\ \xi\in \mathcal I_\varepsilon.\] 
Respectively,
\be\label{liu}\min_{z\in \pa \mathcal B*}\frac{1}{|g(z)-g_\pm(y^{-1})|}>C_1(\varepsilon)>0,\ \mbox{uniformly with respect to}\ \xi\in \mathcal I_\varepsilon.\ee
Moreover, the matrix functions $M^{mod}(z)$, $[M^{mod}(z)]^{-1}$, $r(z)$, $r^{-1}(z)$ are bounded uniformly with respect to $z\in\mathcal B\cup\mathcal B^*$ and $\xi\in\mathcal I_\varepsilon$. From \eqref{ppsi}, 
\eqref{defr} and \eqref{defW} we conclude that
\be\label{imp44}
\sup_{\xi\in\mathcal I_\varepsilon}\sup_{z\in \pa\mathcal B\cup\pa\mathcal B^*}\|W(z)\|\leq\frac{C_2(\varepsilon)}{t}.\ee
On the other hand, \eqref{liu},  \eqref{we}, \eqref{vi3} and \eqref{estforv2} imply
\be\label{imp55}
\sup_{\xi\in\mathcal I_\varepsilon}\sup_{z\in\mathcal K\setminus (\pa\mathcal B\cup\pa\mathcal B^*)}\|W(z)\|\leq O\Big(\E^{-C_3(\varepsilon)t)}\Big).
\ee

\section{Completion of the asymptotic analysis}\label{s:completion}
The aim of this section is to establish that the solution $m^{(3)}(z)$  is well approximated by   $m^{mod}(z)=\rI M^{\mathrm{mod}}(z)$ as $z\to 0$.
We follow the well-known approach via singular integral equations (see e.g., \cite{dz}, \cite{GT}, \cite[Ch.\ 4]{its}, \cite{len}). 
A peculiarity of this approach applied to the Toda equation is generated by the type of normalization condition of the vector RHP and the symmetry condition.  
In particular, if we want to preserve the symmetry condition \eqref{symer} in the Cauchy-type formula for 
$\nu(z)$, we should use a matrix Cauchy kernel (cf.\ \cite[Equ.\ (B.8)]{KTb},
\[
\hat \Omega(s,z)=\begin{pmatrix}\frac{1}{s-z} & 0\\ 0 & \frac{1}{s-z} -\frac{1}{s}\end{pmatrix}ds, 
\quad s\in \mathcal K, \quad z\notin \mathcal K.\]
Since (\cite[Equ.\ (B.9)]{KTb})
\[\hat\Omega(s, z^{-1})=\si_1\hat \Omega(s^{-1},z)\si_1,
\]
this implies with \eqref{symnu} and \eqref{symW} the symmetry
\[\int_{\mathcal K}\nu_-(s)W(s)\hat\Omega(s,z)= \int_{\mathcal K}\nu_-(s)W(s)\hat \Omega(s,z^{-1})\si_1.
\]
Note that the $1,1$-entry of the Cauchy kernel $\hat\Omega(s,z)$ has zero at $z=\infty$ while the $2,2$-entry has zero at $z=0$. From \eqref{merr} and \eqref{00} it follows that
\[\aligned 
\nu(z)&=\begin{pmatrix}\nu_1(\infty),& \nu_2(0)\end{pmatrix} + \frac{1}{2\pi\I}\int_{\mathcal K}\nu_-(s)W(s)\hat\Omega(s,z)\\ &=\tau \begin{pmatrix}1, & 1\end{pmatrix} + \frac{1}{2\pi\I}\int_{\mathcal K}\nu_-(s)W(s)\hat\Omega(s,z).\endaligned\]

Let $\mathfrak C$ denote the Cauchy operator associated with $\mathcal K$,
\[
(\mathfrak C h)(z)=\frac{1}{2\pi\I}\int_{\mathcal K}h(s)\hat\Omega(s,z), \qquad s\in\C\setminus\mathcal K,
\]
where $h= \begin{pmatrix} h_1, & h_2 \end{pmatrix}\in L^2(\mathcal K)$ and satisfies the symmetry $h(s)=h(s^{-1})\si_1$.  
Let  $(\mathfrak C_+ h)(z)$ and $(\mathfrak C_- h)(z)$ be the non-tangential limiting values of $ (\mathfrak C h)(z)$ from
 the left and right sides of $\mathcal K$, respectively.
As usual, we introduce the operator $\mathfrak C_{W}:L^2(\mathcal K)\cap L^\infty(\mathcal K)\to
L^2(\mathcal K)$ by  $\mathfrak C_{W} h=\mathfrak C_-(h W)$. 
By virtue of \eqref{imp44} and \eqref{imp55} we obtain
\[
\|\mathfrak C_{W}\|=\|\mathfrak C_{W}\|_{L^2(\mathcal K)\to L^2(\mathcal K)}\leq C\|W\|_{L^\infty(\mathcal K)}= O\Big(\frac{1}{t}\Big)
\] 
as well as
\[ 
\|(\id - \mathfrak C_{W})^{-1}\|=\|(\id - \mathfrak C_{W})^{-1}\|_{L^2(\mathcal K)\to L^2(\mathcal K)}\leq \frac{1}{1-O(t^{-1})}
\]
for sufficiently large $t$. Consequently, for $t\gg 1$, on $\mathcal K$ we define a vector function
\[
\mu(s) =(\tau, \ \tau) + (\id - \mathfrak C_{W})^{-1}\mathfrak C_{W}\big((\tau, \ \tau)\big)(s),
\]
with $\tau$ given  by \eqref{00}.
Then 
\begin{align}\nn
\|\mu(s) - (\tau, \ \tau)\|_{L^2(\mathcal K)} &\leq \|(\id - \mathfrak C_{W})^{-1}\| \|\mathfrak C_{-}\| \|W\|_{L^\infty(\mathcal K)}\\
&= O(t^{-1}).\label{estmu}
\end{align}
With the help of $\mu$, the vector function $\nu(z)$ can be represented as 
\[
\nu(z)=(\tau, \ \tau) +\frac{1}{2\pi\I}\int_{\mathcal K}\mu(s) W(s) \hat \Omega(s,z),
\]
and by virtue of \eqref{estmu}, \eqref{imp44} and \eqref{imp55}  we obtain as $z\to 0$
\be\label{nush}
\nu(z)=(\tau, \ \tau) + \frac{1}{2\pi\I } \int_{\mathcal K} (\tau, \ \tau) W(s)\begin{pmatrix} s^{-1}+zs^{-2}&0\\0&zs^{-2}\end{pmatrix}ds + E(z).
\ee
Here $E(z)$ is a holomorphic vector function in a vicinity of $z=0$ with 
\[
\|E(z)\|\leq \|W\|_{L^2(\mathcal K)}\|\mu (s)- (\tau, \ \tau)\|_{L^2(\mathcal K)}(1 +O(z))=O(t^{-2})(1 +O(z)),
\] 
and $O(z)$ is uniformly bounded for $\xi\in\mathcal I_\varepsilon$.
From \eqref{nush} and \eqref{svyaz} we get
\be\label{psun}
m^{(3)}(z) = \nu(z) M^{mod}(z)=
\tau m^{\text{mod}}(z) + \tau O(t^{-1}) E_1(z),\ee 
where $E_1(z)$ is a holomorphic vector function in a vicinity of $z=0$, uniformly bounded with respect to $\xi\in\mathcal I_\epsilon$.
The normalization conditions for $m^{(3)}$ and $m^{\text{mod}}$ imply that $\tau^2(1 + O(t^{-1}))=1$, 
that is, \[\tau=1 + O(t^{-1}).\] Together with \eqref{psun}, \eqref{defash}--\eqref{uzhe} and \eqref{connection6} this implies 

 \begin{theorem}\label{lemaprox} The following representation holds for $t\to\infty$ and $n\to\infty$ 
\be\label{finalapr}m_1(z)m_2(z)=H^2(z)\frac{\delta(z)\delta(z^{-1})}{\delta(0)\delta(\infty)}+ \beta_1(\xi, t) +\beta_2(\xi, t)z + \beta_2(\xi, t)O(z^2), \quad z\to 0,
 \ee
 where $|\beta_j(\xi, t)|\leq \frac{C(\varepsilon)}{t}$ uniformly with respect to $\xi\in\mathcal I_\varepsilon$.
\end{theorem} 
Our next aim is to clarify the properties of the function \be\label{as77}\mathcal Y(z):=\frac{\delta(z)\delta(z^{-1})}{\delta(0)\delta(\infty)}.\ee To simplify notations denote $\frac{tB}{2\pi} - \frac{\Delta}{4\pi}=:x\in\R.$ 
Then by \eqref{deli}, 
\[\mathcal Y(z) = \frac{\theta\big(A(z) -\frac 12+ x\big)\,\theta\big(A(z) +x\big)\,
\theta\big(A(z^{-1}) -\frac 12+ x\big)\,\theta\big(A(z^{-1}) +x)}
{\delta(0)\delta(\infty)\,\theta\big(A(z) -\frac 12\big)\,\theta\big(A(z)\big)\,\theta\big(A(z^{-1}) -\frac 12\big)\,\theta\big(A(z^{-1})\big)}.\]
Since $x\in \R$, the properties of the Abel integral $A(z)$ listed in Lemma \ref{propabel} imply that $\theta\big(A(z) -\frac 12+ x\big)\theta\big(A(z) +x\big)$ has a simple zero  at $(\mu(x), \pm)$ on one of the sheets of the Riemann surface $\mathbb X$ with projection $\mu(x)\in[y^{-1}, y]$, and a zero $(\mu^{-1}(x), \mp)$ on the other sheet. The function $\theta\big(A(z^{-1}) -\frac 12+ x\big)\theta\big(A(z^{-1}) +x)$ has zeros at $(\mu(x), \mp), (\mu^{-1}(x), \pm)$. The denominator  of $\mathcal Y(z)$ has double zeros (as points on the Riemann surface) at $y$ and $y^{-1}$. Since $y$ and $y^{-1}$ are the branching points on $\mathbb X$,  then these zeros are simple in the variable $z$ on the complex plane. We observe that $\mathcal Y(z)$, being considered as a function in the domain $\C\setminus [q^{-1}, q]$ identified with the upper sheet of $\mathbb X$,  does not have jumps on $[q^{-1}, q]$ and tends to $1$ as $z\to\infty$.  Hence $\mathcal Y(z)$ is a rational function with simple poles at $y$ and $y^{-1}$ and simple zeros at $\mu(x)$ and $\mu^{-1}(x)$. Therefore,
\[
\mathcal Y(z)=\frac{(z-\mu(x))(z-\mu^{-1}(x))}{(z-y)(z - y^{-1})},
\] 
and
\be\label{as78}
\aligned &m_1^{mod}(z)m_2^{mod}(z)=H^2(z)\mathcal Y(z)=\frac{(z-\mu(x))(z-\mu^{-1}(x))}{\sqrt{(z-q)(z-q^{-1})(z-y)(z-y^{-1})}}\\
&= \frac{z+z^{-1} - \mu(x) -\mu^{-1}(x)}{\sqrt{(z + z^{-1} - q - q^{-1})(z + z^{-1} - y - y^{-1})}}=\frac{\la - \la (n,t)}{\sqrt{(\la - (b - 2 a))(\la - \la_y)}},
\endaligned
\ee
where $\la(n,t)=\frac{\mu(x) + \mu^{-1}(x)}{2}\in [\la_y, -1]$. We emphasize that $\mu(x)=\mu(x(n,t))$ depends on $n$ and $t$ via 
\[x= x(n,t)=-\frac{n\Lambda}{4\pi \I} - \frac{t U}{4\pi \I} -\frac{\Delta}{4\pi}.\]
In particular,
$\mu(x(0,0))=\mu(-\frac{\Delta}{4\pi})$.

Let $\Psi(n,t,p,\xi)$, $p\in \mathbb M(\xi)$, be the Baker--Akhiezer function of a finite gap Toda lattice solution $\{\hat a(n,t,\xi), \hat b(n,t,\xi)\}$ associated with the spectrum on the set $[b-2a, \la_y]\cup [-1, 1]$
and with initial divisor point $(\la(0,0), \pm)$, where we choose sign $+$ if $\mu(-\frac{\Delta}{4\pi})\in [-1, y]$ and sign $-$ if $\mu(-\frac{\Delta}{4\pi})\in [y^{-1}, -1]$. Here we took into account that the set $ \{z: |z|<1\}\setminus [y, q]$ is in one-to-one correspondence with the upper sheet of $\mathbb M(\xi)$ (cf.\ Section \ref{sec:5}). 
Then (cf.\ \cite{tjac})
\[\Psi(n, t, p,\xi)\Psi(n,t,p^*,\xi)=\frac{\la - \la(n,t)}{\la - \la(0,0)},\]
where $\la(n,t)$ is the   projection of the zero divisor for $\Psi$. Equation \eqref{defdelta1} and our considerations above justify this claim.
In particular, according to the trace formulas we have
\begin{align*}
& \hat b(n,t,\xi)=\frac{1}{2}\left(b-2a +\la_y(\xi) - 2\la(n,t)\right), \\
&\hat a(n,t,\xi)^2 + \hat a(n-1, t,\xi)^2=\\
&\qquad  \frac{1}{4} \left( 2+ (b-2a)^2 +\la_y(\xi)^2 - 2 \la(n,t)^2 - \frac{1}{2}\left(b-2a +\la_y(\xi) - 2\la(n,t)\right)^2 \right).
\end{align*}
The operator $\hat H(t)=\hat H(t,\xi)$ associated with these coefficients is reflectionless (since it is finite gap) and has the following Green's function (cf.\ \cite{tjac})
\[
\hat G(\la, n,n, t)=-\frac{\la -\la(n,t)}{\sqrt{(\la^2 -1)(\la - (b - 2 a))(\la - \la_y)}}=-\frac{m_1^{mod}(z)m_2^{mod}(z)}{\sqrt{\la^2 -1}}.
\]
Combining \eqref{gerald1}, \eqref{asympnew}, \eqref{imp98}, \eqref{finalapr}, \eqref{as77} and \eqref{as78} we arrive at 
Theorem~\ref{thor7}, where we switched back to the notation $\gamma(\xi)$ instead of $\lambda_y$.

 \section{Discussions} \label{sec:disc}
In this section we briefly discuss how to derive and justify the asymptotics in the left region $\mathcal I_{1, \varepsilon}:=[ 
 \xi_{cr,1}+\varepsilon, \xi_{cr,1}^\prime-\varepsilon]$. A justification of the asymptotics in the middle region $(\xi_{cr,1}^\prime, \xi_{cr}^\prime)$ which takes into account the presence of resonances and the discrete spectrum in the gap 
 $(b+2a, -1)$ is given in \cite{EM}.
 
First of all, if left and right background spectra are of equal length, that is, in  case $a=1$, there is no need for an independent extensive  study. Indeed, for arbitrary $a>0$ let us consider the Toda lattice associated with the functions
\[
\breve a(n, t)= \frac{1}{2a} a\big(-n-1, \tfrac{t}{2a}\big), \quad \breve b(n, t)= 
\frac{1}{2a} \Big(b-b\big(-n, \tfrac{t}{2a}\big)\Big),
\]
where $\{a(n,t), b(n,t)\}$ is the solution of \eqref{tl}--\eqref{main}, \eqref{decay}.
It is straightforward to check that $\{\breve a(n,t), \breve b(n,t)\}$ satisfy  \eqref{tl} associated with 
the initial profile
\begin{align*} \label{ini2}
\begin{split}	
& \breve a(n,0)\to \frac {1}{2a}, \quad \breve b(n,0) \to \frac{b}{2a}, \quad \mbox{as $n \to -\infty$}, \\
& \hat a(n,0)\to \frac{1}{2}, \quad \hat b(n,0) \to 0, \quad \mbox{as $n \to +\infty$}.
\end{split}
\end{align*}
If $a=1$, the region $(\xi_{cr}^\prime, \xi_{cr})$ for $\{\breve a (n,t), \breve b(n,t)\}$ coincides with $(\xi_{cr,1}^\prime, \xi_{cr,1})$ for $\{a(n,t), b(n,t)\}$, and therefore we can simply apply the results of Theorem \ref{thor7}. This approach is applicable for arbitrary $a$ when $b-2a>1$, that is, for rarefaction waves. Unfortunately, for the shock waves and $a\neq 1$ we  are not able to match these regions. Moreover, even if they would match, this approach would still require cumbersome computations. Indeed,  applying the asymptotics from Theorem~\ref{thor7} to 
\[
 a(m, \breve t) = 2a \breve a(n,t),\quad 
b(m,\breve t) =b - 2a\breve b(n+1,t),\quad  
\breve t = \frac{t}{2a},\quad m=-n-1,\]
 one has to take into account the new time  and space variables and recompute $\mathfrak b$ 
and $\mathfrak \tau$ periods in theta functions, as far as the initial Dirichlet eigenvalues. 

As we see from the above considerations, the form of the $g$-function in $\mathcal I_{1, \varepsilon}$ is dictated by the spectrum $[b-2a, b+2a]\cup [\gamma(\xi), 1]$, where $\gamma(\xi)\in [-1, 1]$. Therefore on the $z$-plane the images of this point will belong to  $\mathbb T\setminus \{-1, 1\}$. Denote them by $z_0$ and $\overline{z_0}$. It is clear that the piece-wise constant jump matrix for the respective model problem will appear on the union of the real interval  and the arc,
\[ [q, q^{-1}]\cup \{z\in \mathbb T: \re z<\re z_0\}.
\]
A construction of the vector and matrix model solutions in theta functions for such a contour in terms of $z$ is quite  bulky and not transparent for further analysis.
For this reason it is more convenient to study the asymptotics for $\xi\in\mathcal I_{1, \varepsilon}$  using the other vector RHP stated with respect to  the left scattering data $R_\ell(\zeta, t)$, $T_\ell(\zeta, t)$  on the $\zeta$-plane (cf.\ \eqref{spec5}).
The left phase function \eqref{phil} is used and replaced by a suitable $g$-function; the structure of the jump matrices and the further analysis is completely analogous to the one given in this paper. It allows us to conclude that the error term in this region is described in terms of Airy functions and is of order $O(t^{-1})$. The error term in the middle region is of order $O(\E^{-C(\varepsilon) t})$, \cite{EM}.

Our last remark concerns  condition \eqref{decay}. The value  of $\rho$ given by \eqref{nu} can be significantly reduced up to any $\rho>0$ if the point $q_1$ is non-resonant, because in the non-resonant case we do not need to apply the lens mechanism around the domains
$\Omega_{\mathfrak r}$ and $\Omega_{\mathfrak r}^*$. It was used to remove a possible singularity of $m$ at $q_1$. Moreover,  the condition $\rho>-\log|q_1|$ is sufficient to remove the singularity. Condition \eqref{nu} was chosen to 
achieve less cumbersome formulas for the jump matrices. Note that condition \eqref{est94} is only essential in the resonant case, and one can expect that the asymptotics in Theorem \ref{thor7} hold in the region \eqref{mathcalI} if $q_1$ is non-resonant.
\vskip 4mm
\noindent{\bf Acknowledgments.} We are grateful to Alexander Minakov for useful discussions. I.E. and A.P. are indebted to the Department of Mathematics at the University of Vienna for its hospitality and support during the autumn of 2019, when this work was done.

\end{document}